\documentclass[runningheads,a4paper,draft]{llncs}
\usepackage{mathtools,amssymb}
\usepackage{nicefrac}
\usepackage{graphicx}
\usepackage{array}
\usepackage{stmaryrd}
\usepackage[textsize=small]{todonotes}
\usepackage{xspace}
\usepackage{infer}
\usepackage{marginnote}
\usepackage{tablefootnote}
\usepackage{xcolor}
\usepackage[final,colorlinks=true,linktocpage=true,hyperfootnotes]{hyperref}
\usepackage{url}
\urldef{\mailsa}\path|{benjamin.kaminski,katoen,matheja,federico.olmedo}|
\urldef{\mailsc}\path|@cs.rwth-aachen.de|    
\newcommand{\keywords}[1]{\par\addvspace\baselineskip
\noindent\keywordname\enspace\ignorespaces#1}

\usepackage{tikz}
\usetikzlibrary{arrows,decorations.pathmorphing,positioning,fit,trees,shapes,shadows,automata,calc}

\allowdisplaybreaks

\newcommand{\ie}{i.e.\@\xspace}
\newcommand{\wrt}{w.r.t.\@\xspace}
\newcommand{\eg}{e.g.\@\xspace}

\DeclareMathAlphabet{\mathpzc}{OT1}{pzc}{m}{it}

\newcommand{\pProgs}{\textnormal{\textsf{pProgs}}\xspace}   
\newcommand{\Vars}{\ensuremath{\mathsf{Var}}\xspace}   
\newcommand{\Vals}{\ensuremath{\mathsf{Val}}\xspace}    
\newcommand{\DExprs}{\ensuremath{\mathsf{DExp}}\xspace}    
\newcommand{\Dists}{\ensuremath{\mathcal{D}}}

\newcommand{\appProgs}{\Stmt}

\newcommand{\E}{\ensuremath{\mathbb{T}}}

\newcommand{\States}{\ensuremath{\Sigma}}
\newcommand{\Rpos}{\ensuremath{\mathbb{R}_{{}\geq 0}}}
\newcommand{\Rposinf}{\ensuremath{\Rpos^{\infty}}}

\newcommand{\Nats}{\ensuremath{\mathbb N}\xspace}

\newcommand{\PReals}{\ensuremath{\mathbb R_{\geq 0}}\xspace}

\newcommand{\ZO}{[0,\,\! 1]}                             

\newcommand{\lrule}[1]{\textnormal{\small \textsf{[#1]}}\xspace}

\newcommand{\SKIP}{\ensuremath{\textnormal{\texttt{skip}}}}
\newcommand{\EMPTY}{\ensuremath{\textnormal{\texttt{empty}}}}
\newcommand{\HALT}{\ensuremath{\textnormal{\texttt{halt}}}}
\newcommand{\AssignSymbol}{\mathrel{\textnormal{\texttt{:=}}}}
\newcommand{\ASSIGN}[2]{\ensuremath{#1 \AssignSymbol #2}}
\newcommand{\AppAssignSymbol}{\mathrel{\textnormal{\texttt{:}}\hspace{-.1em}{\approx}}}
\newcommand{\APPASSIGN}[2]{\ensuremath{#1 \AppAssignSymbol\hspace{.1em} #2}}
\newcommand{\COMPOSE}[2]{\ensuremath{#1\textnormal{\texttt{;}}\: #2}}

\newcommand{\NDCHOICE}[2]{\ensuremath{\left\{{#1}\right\}\mathrel{\Box}\left\{{#2}\right\}}}

\newcommand{\ITE}[3]{\ensuremath{\textnormal{\texttt{if}}\:(#1)\:\{#2\}\:\textnormal{\texttt{else}}\:\{#3\}}}
\newcommand{\If}{\ensuremath{\textnormal{\texttt{if}}}}
\newcommand{\Else}{\ensuremath{\textnormal{\texttt{else}}}}

\newcommand{\WHILEDO}[2]{\ensuremath{\textnormal{\texttt{while}}\:(#1)\:\{#2\}}}
\newcommand{\WHILE}{\ensuremath{\textnormal{\texttt{while}}}}
\newcommand{\Do}{\ensuremath{\textnormal{\texttt{do}}}}
\newcommand{\BOUNDEDWHILE}[3]{\ensuremath{\WHILE^{<{#1}}\:(#2)\:\{#3\}}}

\newcommand{\eetsymbol}{\ensuremath{\textnormal{\textsf{ert}}}}
\newcommand{\boldeetsymbol}{\ensuremath{\textnormal{\textsf{\textbf{ert}}}}}

\newcommand{\eet}[2]{\ensuremath{\eetsymbol\left[{#1}\right]\left({#2}\right)}}

\newcommand{\Exp}[2]{\ensuremath{\textnormal{\textsf{E}}_{#1}\left({#2}\right)}}
\newcommand{\subst}[2]{\ensuremath{\left[{#1}/{#2}\right]}}
\newcommand{\ind}[1]{\ensuremath{{\chi_{#1}}}}
\newcommand{\zero}{\ensuremath{\boldsymbol{0}}}
\newcommand{\one}{\ensuremath{\boldsymbol{1}}}
\newcommand{\infinity}{\ensuremath{\boldsymbol{\infty}}}

\newcommand{\lfp}{\ensuremath{\textnormal{\textsf{lfp}}\,}}

\newcommand{\ExpRew}[2]{\ensuremath{\textnormal{\textsf{ExpRew}}^{#1}\left(#2\right)}}
\newcommand{\eval}[1]{\ensuremath{\llbracket {#1} \rrbracket}}
\newcommand{\probof}[2]{\ensuremath{\llbracket {#1} \colon {#2} \rrbracket}}

\newcommand{\probofInstance}[3]{\ensuremath{\probof{#1}{#2}({#3})}}

\newcommand{\pto}{\ensuremath{\rightharpoonup}}
\newcommand{\ctert}[1]{\ensuremath{\mathbf{#1}}}
\newcommand{\pexprdeno}[1]{\llbracket #1 \rrbracket}
\newcommand{\ExpToFun}[1]{\eval{#1}}

\newcommand{\true}{\ensuremath{\mathsf{true}}\xspace}
\newcommand{\false}{\ensuremath{\mathsf{false}}\xspace}

\providecommand{\eqdef}{\raisebox{-.2ex}[.2ex]{$\:\stackrel{\scriptscriptstyle \vartriangle}{=}\:$}}
\newcommand{\To}{\rightarrow}                          
\newcommand{\ev}[2]{\mathbb{E}_{#1}\left(#2\right)}  

\newcommand{\supn}{\sup\nolimits_{n}}

\newcommand{\rmdp}{\ensuremath{\mathfrak{M}}}
\newcommand{\mdpStates}{\ensuremath{\mathcal S}\xspace}
\newcommand{\mdpTPF}{\ensuremath{\textnormal{\textbf{P}}}\xspace} 
\newcommand{\mdpAct}{\ensuremath{\mathit{Act}}\xspace}
\newcommand{\mdpInit}{\ensuremath{s_0}\xspace}

\newcommand{\mdpRew}{\ensuremath{\mathit{rew}}\xspace}
\newcommand{\mdpSched}{\ensuremath{\mathfrak{S}}\xspace}
\newcommand{\rmdpMc}{\ensuremath{\rmdp_{\mdpSched}}\xspace}
\newcommand{\rmdpStmt}[3]{\ensuremath{\rmdp^{#1}_{#2}\sem{#3}}}

\newcommand{\ipath}{\pi}
\newcommand{\fpath}{\hat\ipath}

\newcommand{\mdpPaths}[1]{\ensuremath{\textrm{Paths}^{#1}}\xspace}
\newcommand{\mdpFPaths}[1]{\ensuremath{\textrm{Paths}_{\mathit{fin}}^{#1}}\xspace}
\newcommand{\mdpSPaths}[2]{\ensuremath{\Pi(#1,#2)}\xspace}

\newcommand{\mdpPr}[2]{\mathrm{Pr}^{#1}\{#2\}}

\newcommand{\mdpEventually}{\Diamond}

\newcommand{\mdpRewSink}{\ensuremath{f}\xspace}
\newcommand{\mdpSink}{\ensuremath{\,\mathpzc{sink}\,}\xspace}
\newcommand{\mdpSinkState}{\ensuremath{\langle \mdpSink \rangle}\xspace}

\newcommand{\sem}[1]{\llbracket #1 \rrbracket}

\newcommand{\Stmt}{\ensuremath{\textnormal{\textsf{pProgs}}}\xspace}
\newcommand{\DetStmt}{\ensuremath{\mathnormal{\textsf{dProgs}}}\xspace}
\newcommand{\stmt}{\ensuremath{C}\xspace}
\newcommand{\opSem}[2]{\ensuremath{\mathbb{C}\sem{#1}({#2})}\xspace}

\newcommand{\Terminated}{\ensuremath{\downarrow}\xspace}
\newcommand{\actL}{\ensuremath{\mathit{L}}\xspace}
\newcommand{\actN}{\ensuremath{\tau}\xspace}
\newcommand{\actR}{\ensuremath{\mathit{R}}\xspace}
\newcommand{\mdpState}[2]{\ensuremath{\langle {#1},\, {#2} \rangle}\xspace}
\newcommand{\ProgramStates}{\Sigma}
\newcommand{\ps}{\sigma}

\newcommand{\transb}[4]{{#1} \xrightarrow{#2}  {#3} \vdash {#4}}
\newcommand{\trans}[6]{\transb{\mdpState{#1}{#2}}{#3}{\mdpState{#4}{#5}}{#6}}

\newcommand{\eetF}{F_{\mdpRewSink}}

\newcommand{\pguard}{\ensuremath{\xi}\xspace}
\newcommand{\pexpr}{\ensuremath{\mu}\xspace}

\newcommand{\mydot}{                         
  \raisebox{-0.5pt}{${\scriptscriptstyle \bullet}~$}}

\newcommand{\htriple}[4]{\ensuremath{\{~{#1}~\}~{#2}~\{~{#3}~\Downarrow~{#4}~\}}\xspace}
\newcommand{\hrtriple}[3]{\ensuremath{\{~{#1}~\}~{#2}~\{~\Downarrow~{#3}~\}}\xspace}
\newcommand{\nvalid}{\ensuremath{\models_{E}}}
\newcommand{\nprove}{\ensuremath{\vdash_{E}}}
\newcommand{\hprove}{\ensuremath{\vdash}}

\newcommand{\linen}[1]{{\scriptstyle \mathtt{#1}:}\ \ }

\newcommand{\nocol}{\ensuremath{\#\hspace{-.25ex}\mathit{col}}\xspace}
\newcommand{\cpsymbol}{\ensuremath{\mathit{cp}}\xspace}
\newcommand{\cp}[1]{\ensuremath{\cpsymbol[#1]}\xspace}
\newcommand{\cond}[1]{\ensuremath{[#1]}\xspace}

\begin{document}

\mainmatter  

\title{Weakest Precondition Reasoning for\\ Expected Run--Times of Probabilistic Programs\footnote{This work was supported by the Excellence Initiative of the German federal and state government.}}

\titlerunning{\textsf{wp}--Reasoning for Expected Run--Times of Prob.\ Programs}

\author{Benjamin Lucien Kaminski \and Joost-Pieter Katoen \and\\ Christoph Matheja \and Federico Olmedo}
\authorrunning{Kaminski, Katoen, Matheja, Olmedo}

\institute{Software Modeling and Verification Group, RWTH Aachen University\\
Ahornstra\ss{}e 55, 52074 Aachen, Germany\\
\mailsa\\
\mailsc}

\toctitle{Weakest Precondition Reasoning for Expected Run--Times of Probabilistic Programs}
\tocauthor{Kaminski, Katoen, Matheja, Olmedo}

\maketitle

\begin{abstract}
  This paper presents a \textsf{wp}--style calculus for obtaining bounds on the
  expected run--time of probabilistic programs.  Its application includes
  determining the (possibly infinite) expected termination time of a
  probabilistic program and proving positive almost--sure termination---does a
  program terminate with probability one in finite expected time?  We provide
  several proof rules for bounding the run--time of loops, and prove the
  soundness of the approach with respect to a simple operational model.  We show
  that our approach is a conservative extension of Nielson's approach for
  reasoning about the run--time of deterministic programs.  We analyze the
  expected run--time of some example programs including a one--dimensional
  random walk and the coupon collector problem.

\end{abstract}

\section{Introduction}
\label{sec:intro}

Since the early days of computing, randomization has been an important tool for the construction of algorithms.
It is typically used to convert a deterministic program with bad worst--case behavior into an efficient randomized algorithm that yields a correct output with high probability.
The Rabin--Miller primality test, Freivalds' matrix multiplication, and the random pivot selection in Hoare's quicksort algorithm are prime examples.
Randomized algorithms are conveniently described by probabilistic programs.
On top of the usual language constructs, probabilistic programming languages offer the possibility of sampling values from a probability distribution.
Sampling can be used in assignments as well as in Boolean guards.

The interest in probabilistic programs has recently been rapidly growing. 
This is mainly due to their wide applicability~\cite{DBLP:conf/icse/GordonHNR14}. 
Probabilistic programs are for instance used in security to describe cryptographic constructions and security experiments.
In machine learning they are used to describe distribution functions that are analyzed using Bayesian inference.
The sample program
\[ 
C_\mathit{geo} \boldsymbol{\colon} \;\; 
\ASSIGN{b}{1};\; \WHILEDO{b = 1}{\APPASSIGN{b}{\nicefrac 1 2 \cdot \! \langle 0 \rangle +
    \nicefrac 1 2 \cdot \! \langle 1 \rangle}}
\]
for instance flips a fair coin until observing the first heads (i.e. 0). 
It describes a geometric distribution with parameter $\nicefrac 1 2$.

The run--time of probabilistic programs is affected by the outcome of their coin tosses.  Technically speaking, the run--time is a random variable, \ie it is $t_1$ with probability $p_1$, $t_2$ with probability $p_2$ and so on.  An important measure that we consider over probabilistic programs is then their \emph{average} or \emph{expected} run--time (over all inputs).  Reasoning about the expected run--time of probabilistic programs is surprisingly subtle and full of nuances.  In classical sequential programs, a single diverging program run yields the program to have an infinite run--time.  This is not true for probabilistic programs.  They may admit arbitrarily long runs while having a finite expected run--time.  The program $C_\mathit{geo}$, for instance, does admit arbitrarily long runs as for any $n$, the probability of not seeing a heads in the first $n$ trials is always positive.  The expected run--time of $C_\mathit{geo}$ is, however, finite.

In the classical setting, programs with finite run--times can be sequentially composed yielding a new program again with finite run--time.
For probabilistic programs this does not hold in general.
Consider the pair of programs
\begin{align*}
C_1 \boldsymbol{\colon} \;\; 
& \ASSIGN{x}{1};\; \ASSIGN{b}{1};\; \WHILEDO{b = 1}{\APPASSIGN{b}{\nicefrac 1 2 \cdot \! \langle 0 \rangle + \nicefrac 1 2 \cdot \! \langle 1 \rangle};\, \ASSIGN{x}{2x}} \quad \mbox{and} \\
C_2 \boldsymbol{\colon} \;\; 
 & \WHILEDO{x > 0}{\ASSIGN{x}{x-1}}~.
\end{align*}
The loop in $C_1$ terminates on average in two iterations; it thus has a finite expected run--time. 
From any initial state in which $x$ is non--negative, $C_2$ makes $x$ iterations, and thus its expected run--time is finite, too. 
However, the program $C_1 ; C_2$ has an \emph{infinite} expected run--time---even though it almost--surely terminates, i.e.~it terminates with probability one.
Other subtleties can occur as program run--times are very sensitive to variations in the probabilities occurring in the program. 

Bounds on the expected run--time of randomized algorithms are typically obtained using a detailed analysis exploiting classical probability theory (on expectations or martingales)~\cite{Frandsen:1998,DBLP:books/cu/MotwaniR95}.
This paper presents an alternative approach, based on formal program development and verification techniques.
We propose a \textsf{wp}--style calculus \`a la Dijkstra for obtaining bounds on the expected run--time of probabilistic programs. 
The core of our calculus is the transformer $\eetsymbol$, a quantitative variant of Dijkstra's ${\sf wp}$--transformer.
For a program $C$, $\eet{C}{f}(\sigma)$ gives the expected run--time of $C$ started in initial state $\sigma$ under the assumption that $f$ captures the run--time of the computation following $C$. 
In particular, $\eet{\stmt}{\ctert{0}}(\sigma)$ gives the expected run--time of program $\stmt$ on input $\sigma$ (where $\ctert{0}$ is the constantly zero run--time). 
Transformer $\eetsymbol$ is defined inductively on the program structure.
We prove that our transformer conservatively extends Nielson's approach~\cite{Nielson:SCP:87} for reasoning about the run--time of deterministic programs.
In addition we show that $\eet{C}{f}(\sigma)$ corresponds to the expected run--time in a simple operational model for our probabilistic programs based on Markov Decision Processes (MDPs).
The main contribution is a set of proof rules for obtaining (upper and lower) bounds on the expected run--time of loops.
We apply our approach for analyzing the expected run--time of some example programs including a one--dimensional random walk and the coupon collector problem~\cite{Mitzenmacher:2005}.

We finally point out that our technique enables determining the (possibly infinite) expected time until termination of a probabilistic program and proving (universal) \emph{positive almost--sure termination}---does a program terminate with probability one in finite expected time (on all inputs)?
It has been recently shown~\cite{DBLP:conf/mfcs/KaminskiK15} that the universal positive almost--sure termination problem is $\Pi^0_3$--complete, and thus strictly harder to solve than the universal halting problem for deterministic programs.
To the best of our knowledge, the formal verification framework in this paper is the first one that is proved sound and can handle both positive almost--sure termination and infinite expected run--times.


\paragraph{Related work.}
Several works apply \textsf{wp}--style-- or Floyd--Hoare--style reasoning to study quantitative aspects of classical programs.
Nielson~\cite{Nielson:SCP:87,Nielson:UTCS:07} provides a Hoare logic for determining upper bounds on the run--time of deterministic programs.
Our approach applied to such programs yields the tightest upper bound on the run--time that can be derived using Nielson's approach.
Arthan \emph{et al.}~\cite{DBLP:journals/tocl/ArthanMMO09} provide a general framework for sound and complete Hoare--style logics, and show that an instance of their theory can be used to obtain upper bounds on the run--time of while programs.
Hickey and Cohen~\cite{Hickey:1988} automate the average--case analysis of deterministic programs by generating a system of recurrence equations derived from a program whose efficiency is to be analyzed.
They build on top of Kozen's seminal work~\cite{DBLP:journals/jcss/Kozen81} on semantics of probabilistic programs.
Berghammer and M\"uller--Olm~\cite{DBLP:conf/lopstr/BerghammerM03} show how Hoare--style reasoning can be extended to obtain bounds on the closeness of results obtained using approximate algorithms to the optimal solution.
Deriving space and time consumption of deterministic programs has also been considered by Hehner~\cite{Hehner:FAC:1998}.
Formal reasoning about probabilistic programs goes back to Kozen~\cite{DBLP:journals/jcss/Kozen81}, and has been developed further by Hehner~\cite{Hehner:FAC:2011} and McIver and Morgan~\cite{mciver}.
The work by Celiku and McIver~\cite{McIver:FM:2005} is perhaps the closest to our paper.
They provide a \textsf{wp}--calculus for obtaining performance properties of probabilistic programs, including upper bounds on expected run--times.
Their focus is on refinement.
They do neither provide a soundness result of their approach nor consider lower bounds.
We believe that our transformer is simpler to work with in practice, too.
Monniaux~\cite{DBLP:conf/sas/Monniaux01} exploits abstract interpretation to automatically prove the probabilistic termination of programs using exponential bounds on the tail of the distribution.
His analysis can be used to prove the soundness of experimental statistical methods to determine the average run--time of probabilistic programs.
Brazdil \emph{et al.}~\cite{DBLP:journals/jcss/BrazdilKKV15} study the run--time of probabilistic programs with unbounded recursion by considering probabilistic pushdown automata (pPDAs).
They show (using martingale theory) that for every pPDA the probability of performing a long run decreases exponentially (polynomially) in the length of the run, iff the pPDA has a finite (infinite) expected runtime. 
As opposed to our program verification technique, \cite{DBLP:journals/jcss/BrazdilKKV15} considers reasoning at the operational level.
Fioriti and Hermanns~\cite{luis} recently proposed a typing scheme for deciding almost-sure termination.
They showed, amongst others, that if a program is well-typed, then it almost surely terminates.
This result does not cover positive almost-sure-termination.

%
%
%

\paragraph{Organization of the paper.}
Section~\ref{sec:language} defines our probabilistic programming language.
Section~\ref{sec:eet} presents the transformer $\eetsymbol$ and studies its elementary properties such as continuity.
Section~\ref{sec:operational} shows that the $\eetsymbol$ transformer coincides with the expected run--time in an MDP that acts as operational model of our programs.
Section~\ref{sec:loopy-rules} presents two sets of proof rules for obtaining upper and lower bounds on the expected run--time of loops.
In Section~\ref{sec:eet-vs-nielson}, we show that the $\eetsymbol$ transformer is a conservative extension of Nielson's approach for obtaining upper bounds on deterministic programs.
Section~\ref{sec:applications} discusses two case studies in detail.
Section~\ref{sec:conclusion} concludes the paper.
The proofs of the main facts are included in the body of the paper.  All other proofs as well as the detailed calculations for the examples are provided in the appendix.  

\emph{This version of the paper includes a correction of the invariant in the random walk case study from an earlier version of this paper. The resulting run--time remains unchanged.}

\section{A Probabilistic Programming Language}
\label{sec:language}
In this section we present the probabilistic programming language used throughout this paper, together with its run--time model.
To model probabilistic programs we employ a standard imperative language \`a la Dijkstra's Guarded Command Language~\cite{Dijkstra} with two distinguished features: we allow distribution expressions in assignments and guards to be probabilistic. 
For instance, we allow for probabilistic assignments like
\[
\APPASSIGN{y}{\mathtt{Unif}[1 \ldots x]}
\]
which endows variable $y$ with a uniform distribution in the interval $[1\ldots
x]$. We allow also for a program like
\[
\ASSIGN{x}{0}; \; 
\WHILE \: \big(p \cdot \!  \langle \true \rangle + (1{-}p) \!\,\cdot \!
  \langle \false \rangle\big) \: \{\ASSIGN{x}{x+1} \}
\]
which uses a probabilistic loop guard to simulate a geometric distribution with success
probability $p$, i.e. the loop guard evaluates to $\true$ with probability $p$ and to
$\false$ with the remaining probability~$1{-}p$.

Formally, the set of \emph{probabilistic programs} \pProgs is given by the grammar
$$
\begin{array}{rc@{~~}l@{\qquad}l}
\stmt  ~::=~ 
   & & \EMPTY    & \mbox{empty program}\\
   &\mid& \SKIP        & \mbox{effectless operation}\\
   &\mid& \HALT & \mbox{immediate termination}\\
   &\mid& \APPASSIGN{x}{\mu}        & \mbox{probabilistic assignment}\\
   &\mid& \COMPOSE{\stmt}{\stmt} & \mbox{sequential composition}\\
   &\mid& \NDCHOICE{\stmt}{\stmt} & \mbox{non--deterministic choice}\\
   &\mid& \ITE{\pguard}{\stmt}{\stmt} & \mbox{probabilistic conditional}\\
   &\mid& \WHILEDO{\pguard}{\stmt}      & \mbox{probabilistic while loop}
\end{array}
$$
Here $x$ represents a \emph{program variable} in \Vars, $\mu$ a
\emph{distribution expression} in \DExprs, and $\pguard$ a distribution
expression over the truth values, \ie a \emph{probabilistic guard}, in \DExprs.
We assume distribution expressions in \DExprs to represent discrete probability distributions with a (possibly \emph{infinite}) support of total probability mass 1.  
We use $p_1 \cdot
\langle a_1 \rangle + \cdots + p_n \cdot \langle a_n \rangle$ to denote the
distribution expression that assigns probability $p_i$ to $a_i$.  For instance,
the distribution expression $\nicefrac 1 2 \cdot \! \langle \true \rangle +
\nicefrac 1 2 \cdot\! \langle \false \rangle$ represents the toss of a fair
coin.  Deterministic expressions over program variables such as $x-y$ or $x - y
> 8$ are special instances of distribution expressions---they are understood as
Dirac probability distributions\footnote{A Dirac distribution assigns the total
  probability mass, i.e.\ 1, to a single point.}.

To describe the different language constructs we first present some
preliminaries.  A \emph{program state} $\sigma$ is a mapping from program
variables to values in $\Vals$.  Let $\States \triangleq \{\sigma \mid \sigma
\colon \Vars \To \Vals\}$ be the set of program states.  We assume an
interpretation function $\eval{\:\cdot\:}\colon \DExprs \To (\States \To
\Dists(\Vals))$ for distribution expressions, $\Dists(\Vals)$ being the set of
discrete probability distributions over $\Vals$. For $\mu \in \DExprs$,
$\eval{\mu}$ maps each program state to a probability distribution of values.
We use $\probof{\mu}{v}$ as a shorthand for the function mapping each program
state $\sigma$ to the probability that distribution $\eval{\mu}(\sigma)$ assigns
to value $v$, \ie $\probof{\mu}{v}(\sigma) \triangleq
\mathsf{Pr}_{\eval{\mu}(\sigma)}(v)$, where $\mathsf{Pr}$ denotes the
probability operator on distributions over values.

We now present the effects of $\pProgs$ programs and the run--time model that we
adopt for them.  
$\EMPTY$ has no effect and its
execution consumes no time.  $\SKIP$ has also no effect but consumes, in
contrast to $\EMPTY$, one unit of time.  $\HALT$ aborts any further program
execution and consumes no time.  $\APPASSIGN{x}{\mu}$ is a probabilistic
assignment that samples a value from $\eval{\mu}$ and assigns it to variable
$x$; the sampling and assignment consume (altogether) one unit of time.
$\COMPOSE{C_1}{C_2}$ is the sequential composition of programs $C_1$ and $C_2$.
$\NDCHOICE{C_1}{C_2}$ is a non--deterministic choice between programs $C_1$ and $C_2$; we take a demonic view where we assume that out of $C_1$ and $C_2$ we execute the program with the greatest run--time.
$\ITE{\pguard}{C_1}{C_2}$ is a
probabilistic conditional branching: with probability $\probof{\pguard}{\true}$
program $C_1$ is executed, whereas with probability $\probof{\pguard}{\false} =
1{-}\probof{\pguard}{\true}$ program $C_2$ is executed; evaluating (or more
rigorously, sampling a value from) the probabilistic guard requires an
additional unit of time.  $\WHILEDO{\pguard}{C}$ is a probabilistic while loop:
with probability $\probof{\pguard}{\true}$ the loop body $C$ is executed
followed by a recursive execution of the loop, whereas with probability
$\probof{\pguard}{\false}$ the loop terminates; as for conditionals, each
evaluation of the guard consumes one unit of time.
\begin{example}[Race between tortoise and hare]\label{ex:program}
The probabilistic program
\begin{align*}
&\COMPOSE{\COMPOSE{\APPASSIGN{h}{0}}{\APPASSIGN{t}{30}}}{}\\
&\WHILE~(h \leq t)~\{ \\
&\qquad \If~\bigl(\nicefrac 1 2 \cdot\! \langle \true \rangle + \nicefrac 1 2 \cdot\! \langle \false
 \rangle \big)~\{\APPASSIGN{h}{h + \mathtt{Unif}[0 \ldots 10]}\}\\
&\qquad \COMPOSE{\Else~\{\EMPTY\}}{}\\
&\qquad \APPASSIGN{t}{t+1} \\
& \}~,
\end{align*}
adopted from~\cite{Chakarov:CAV:13}, illustrates the use of the programming language.
It models a race between a hare and a tortoise (variables $h$ and $t$ represent their respective positions). 
The tortoise starts with a lead of $30$ and in each step advances one step forward. 
The hare with probability $\nicefrac{1}{2}$ advances a random number of steps between $0$ and $10$ (governed by a uniform distribution) and with the remaining probability remains still. 
The race ends when the hare passes the tortoise. 
\hfill$\triangle$
\end{example}
We conclude this section by fixing some notational conventions. 
To keep our program notation consistent with standard usage, we use the standard symbol $\AssignSymbol$ instead of $\AppAssignSymbol$ for assignments whenever $\mu$ represents a Dirac distribution given by a deterministic expressions over program variables. 
For instance, in the program in \autoref{ex:program} we write $\ASSIGN{t}{t+1}$ instead of $\APPASSIGN{t}{t+1}$. Likewise, when $\pguard$ is a probabilistic guard given as a deterministic Boolean expression over program variables, we use $\ExpToFun{\pguard}$ to denote $\probof{\pguard}{\true}$ and $\ExpToFun{\lnot \pguard}$ to denote $\probof{\pguard}{\false}$. 
For instance, we write $\ExpToFun{b=0}$ instead of $\probof{b=0}{\true}$.

\section{A Calculus of Expected Run--Times}
\label{sec:eet}

Our goal is to associate to any program $\stmt$ a function that maps each state
$\sigma$ to the average or expected run--time of $\stmt$ started in initial
state $\sigma$. We use the functional space of \emph{run--times}
\[
\E ~\triangleq~ \left\{f ~\middle|~ f\colon \States \rightarrow \Rposinf \right\}
\]
to model such functions. Here, $\Rposinf$ represents the set of non--negative
real values extended with $\infty$.  We consider run--times as a mapping from
program states to real numbers (or $\infty$) as the expected run--time of a
program may depend on the initial program state.

We express the run--time of programs using a continuation--passing style by means of the transformer
\[
\eetsymbol[\:\cdot\:] \colon \Stmt \To (\E \To \E)~.
\]
Concretely, $\eet{\stmt}{f}(\sigma)$ gives the expected run--time of program
$\stmt$ from state $\sigma$ assuming that $f$ captures the run--time of the
computation that follows $\stmt$. Function $f$ is usually referred to as
\emph{continuation} and can be thought of as being evaluated in the
final states that are reached upon termination of $C$. 
Observe that, in particular, if we set $f$ to the constantly zero run--time, $\eet{\stmt}{\ctert{0}}(\sigma)$ gives the expected run--time of program $\stmt$ on input $\sigma$.

The transformer $\eetsymbol$ is defined by induction on the structure of $\stmt$ following the rules in \autoref{table:eet-rules}.  
The rules are defined so as to correspond to the run--time model introduced in Section~\ref{sec:language}.  
That is, $\eet{\stmt}{\ctert{0}}$ captures the expected number of assignments, guard evaluations and $\SKIP$ statements.
\begin{table}[t!]
\begin{center}
\renewcommand{\arraystretch}{1.5}
\begin{tabular}{l@{\qquad}l}
\hline\hline
$~\boldsymbol{\quad\qquad C~}$ & $\boldsymbol{\boldeetsymbol\left[C\right]\left(f\right)}~$\\
\hline\hline
$\quad\qquad~\EMPTY$ & $f$\\
$\quad\qquad~\SKIP$ & $\ctert{1} + f$\\
$\quad\qquad~\HALT$ & $\ctert{0}$\\
$\quad\qquad~\APPASSIGN x \mu$ & $\ctert{1} + \lambda \sigma\mydot
\Exp{\pexprdeno{\mu}(\sigma)}{\lambda v.\: f \subst{x}{v}(\sigma)}$\\
$\quad\qquad~\COMPOSE{C_1}{C_2}$ & $\eet{C_1}{\eet{C_2}{f}}$\\
$\quad\qquad~\!\!\NDCHOICE{C_1}{C_2}$ & $\max\{\eet{C_1}{f},\,
\eet{C_2}{f}\}\qquad\quad\,$\\
$\quad\qquad~\ITE{\pguard}{C_1}{C_2}$ & $\ctert{1} + \probof{\pguard}{\true}
\cdot \eet{C_1}{f} + \probof{\pguard}{\false} \cdot \eet{C_2}{f}$\\
$\quad\qquad~\WHILEDO{\pguard}{C'}$ & $\lfp X\mydot \ctert{1} +
\probof{\pguard}{\false} \cdot f + \probof{\pguard}{\true} \cdot \eet{C'}{X}$\\[-1.5ex]
\end{tabular}
\end{center}
\hrule
\vspace{1\baselineskip}
\caption{Rules for defining the expected run--time transformer $\eetsymbol$. 
$\one$ is the constant run--time $\lambda \sigma. 1$.
$\Exp{\eta}{h} \triangleq \sum_{v} \textsf{Pr}_{\eta}(v) \cdot h(v)$ represents the expected value
  of (random variable) $h$ \wrt distribution $\eta$. For $\sigma \in \States$,
  $f\subst{x}{v}(\sigma) \triangleq f(\sigma\subst{x}{v})$, where $\sigma\subst{x}{v}$
  is the state obtained by updating in $\sigma$ the value of $x$ to $v$. $\max \{f_1, f_2\} \triangleq \lambda \sigma.\max
  \{f_1(\sigma), f_2(\sigma)\}$ represents the point--wise lifting of the $\max$
  operator over $\Rposinf$ to the function space of run--times. $\lfp X \protect\mydot  F(X)$ represents the least fixed point of the transformer $F \colon \E \To \E$.}
\label{table:eet-rules}
\vspace{-1\baselineskip}
\end{table}
Most rules in \autoref{table:eet-rules} are self--explanatory.
$\eetsymbol[\EMPTY]$ behaves as the identity since $\EMPTY$ does not modify the
program state and its execution consumes no time.  On the other hand,
$\eetsymbol[\SKIP]$ adds one unit of time since this is the time required by the
execution of $\SKIP$.  $\eetsymbol[\HALT]$ yields always the constant run--time
$\ctert{0}$ since $\HALT$ aborts any subsequent program execution (making their
run--time irrelevant) and consumes no time. The definition of $\eetsymbol$ on
random assignments is more involved: $\eet{\APPASSIGN x \mu}{f}(\sigma) = 1 +
\sum_{v} \textsf{Pr}_{\pexprdeno{\mu}(\sigma) }(v) \cdot
f(\sigma\subst{x}{v})$ is obtained by adding one unit of time (due to the
distribution sampling and assignment of the value sampled) to the sum of the
run--time of each possible subsequent execution, weighted according to their
probabilities.  $\eetsymbol[\COMPOSE{C_1}{C_2}]$ applies $\eetsymbol[C_1]$ to
the expected run--time obtained from the application of $\eetsymbol[C_2]$.
$\eetsymbol[\NDCHOICE{C_1}{C_2}]$ returns the maximum between the run--time of
the two branches. $\eetsymbol[\ITE{\pguard}{C_1}{C_2}]$ adds one unit of time
(on account of the guard evaluation) to the weighted sum of the run--time of the
two branches.  Lastly, the $\eetsymbol$ of loops is given as the least fixed
point of a run--time transformer defined in terms of the run--time of the 
loop body.
\paragraph{Remark.}
We stress that the above run--time model is a design decision for the sake of
concreteness. All our developments can easily be adapted to capture alternative
models.  These include, for instance, the model where only the number of
assignments in a program run or the model where only the number of loop
iterations are of relevance. We can also capture more fine--grained models, where
for instance the run--time of an assignment depends on the \emph{size} of the
distribution expression being sampled.
\begin{example}[Truncated geometric distribution] 
  \label{example:eet}
  To illustrate the effects of the $\eetsymbol$\xspace
  transformer consider the program in \autoref{fig:example:eet}.
%
\begin{figure}
\begin{align*}
C_\mathit{trunc}\boldsymbol{\colon} \;\; 
 &\If~\bigl(\nicefrac 1 2 \cdot\! \langle \true \rangle + \nicefrac 1 2 \cdot\! \langle \false\rangle \bigr)~\{\ASSIGN{\mathit{succ}}{\true}\} ~\Else~ \{\\
&\qquad \If~\bigl(\nicefrac 1 2 \cdot\! \langle \true \rangle + \nicefrac 1 2
\cdot\! \langle \false
 \rangle \bigr)~\{\ASSIGN{\mathit{succ}}{\true}\}\\ 
 & \qquad \Else~\{\ASSIGN{\mathit{succ}}{\false}\}\\
 & \} 
\end{align*}
\vspace*{-20pt}
\caption{Program modeling a truncated geometric distribution}
\label{fig:example:eet}
\end{figure}
%
It can be viewed as modeling a truncated geometric distribution: we repeatedly flip a fair coin until observing the first heads or completing the second unsuccessful trial. 
The calculation of the expected run--time $\eet{\stmt_\mathit{trunc}}{\ctert{0}}$ of program $C_\mathit{trunc}$ goes as follows:
\begin{flalign*}
\MoveEqLeft[3] \eet{\stmt_\mathit{trunc}}{\ctert{0}}\\
=~	&\ctert{1} + \tfrac{1}{2} \cdot \eet{\ASSIGN{\mathit{succ}}{\true}}{\ctert{0}}\\
	& \;\; + \tfrac{1}{2} \cdot
\eet{\ITE{\ldots}{\ASSIGN{\mathit{succ\!}}{\!\true}}{\ASSIGN{\mathit{succ\!}}{\!\false}}}{\ctert{0}}\\
=~	&\ctert{1} + \tfrac{1}{2} \cdot\ctert{1} + \tfrac{1}{2} \cdot \Bigl(\ctert{1} + \tfrac{1}{2} \cdot \eet{\ASSIGN{\mathit{succ}}{\true}}{\ctert{0}} + \tfrac{1}{2} \cdot \eet{\ASSIGN{\mathit{succ}}{\false}}{\ctert{0}} \Bigr) \\
=~	&\ctert{1} + \tfrac{1}{2} \cdot\ctert{1} + \tfrac{1}{2} \cdot
\Bigl(\ctert{1} + \tfrac{1}{2} \cdot \ctert{1} + \tfrac{1}{2} \cdot \ctert{1}
\Bigr) ~=~ \ctert{\tfrac{5}{2}}
\end{flalign*}
Therefore, the execution of $\stmt_\mathit{trunc}$ takes, on average, 2.5 units of time. 
\hfill$\triangle$
\end{example}
Note that the calculation of the expected run--time in the above example is straightforward as the program at hand is loop--free. 
Computing the run--time of loops requires the calculation of least fixed points, which is generally not feasible in practice.
In Section~\ref{sec:loopy-rules}, we present invariant--based proof rules for reasoning about the run--time of loops.

The \eetsymbol\xspace transformer enjoys several algebraic properties.
To formally state these properties we make use of the point--wise order relation ``$\preceq$'' between run--times: given
$f,g \in \E$, $f \preceq g$ iff $f(\sigma) \leq g(\sigma)$ for all states $\sigma \in \States$.
\begin{theorem}[Basic properties of the $\boldeetsymbol$ transformer]\label{thm:eet-prop}
  For any program $\stmt \in \appProgs$, any constant run--time $\ctert{k} =
  \lambda \sigma. k$ for $k \in \Rpos$, any constant $r \in \Rpos$, and any two
  run--times $f, g \in \E$ the following properties hold:

\begin{tabular}{l@{\hspace{3.5em}}l}
	\\[-1.5ex]
	Monotonicity: 			& $f \preceq g ~\implies~ \eet{\stmt}{f} \,\preceq\, \eet{\stmt}{g}$;\\[1ex]
	Propagation of 			& $\eet{\stmt}{\ctert{k} + f} \:=\: \ctert{k} + \eet{\stmt}{f}$\\
	constants: 				&  provided $\stmt$ is $\HALT$--free;\\[1ex]
	Preservation of $\infinity$: 	& $\eet{\stmt}{\infty} \:=\: \infty$\\
								&  provided $\stmt$ is $\HALT$--free;\\[1ex]
	Sub--additivity: 			& $\eet{C}{f + g} \,\preceq\, \eet{C}{f} + \eet{C}{g}$;\\
								&  provided $\stmt$ is fully probabilistic\footnotemark;\\[1ex]
	Scaling: & $\eet{C}{r \cdot f} \,\succeq\, \min\{1,\, r\} \cdot \eet{C}{f}$;\\
& $\eet{C}{r \cdot f} \,\preceq\, \max\{1,\, r\} \cdot \eet{C}{f}$.
\end{tabular}
\footnotetext{A program is  called \emph{fully probabilistic} if it contains no non--deterministic choices.}
\end{theorem}
\begin{proof}
  Monotonicity follows from continuity (see \autoref{thm:eet-cont}
  below). The remaining proofs proceed by induction on the program structure;
  see Appendices \ref{sec:app-eet-const-prop},
  \ref{sec:app-eet-inf}, \ref{sec:app-eet-additivity}, and
  \ref{sec:app-eet-scaling}.  \qed 
\end{proof}
We conclude this section with a technical remark regarding the well--defined\-ness of the \eetsymbol \xspace transformer. 
To guarantee that \eetsymbol \xspace is well--defined, we must show the existence of the least fixed points used to define the run--time of loops. 
To this end, we use a standard denotational semantics argument (see \eg \cite[Ch.~5]{Winskel:1993}): First we endow the set of run--times $\E$ with the structure of an $\omega$--complete partial order ($\omega$--cpo) with bottom element.
Then we use a continuity argument to conclude the existence of such fixed points.  

Recall that $\preceq$ denotes the point--wise comparison between run--times.
It easily follows that $(\E, {\preceq})$ defines an $\omega$--cpo with bottom
element $\zero = \lambda \sigma. 0$ where the supremum of an $\omega$--chain
$f_1 \preceq f_2 \preceq \cdots$ in $ \E$ is also given point--wise, \ie as
$\supn f_n \triangleq \lambda \sigma. \, \supn f_n(\sigma)$; see
Appendix~\ref{sec:proof-thm-cpo} for details. Now we are in a position to
establish the continuity of the \eetsymbol\ transformer:%
\begin{lemma}[Continuity of the $\boldeetsymbol$ transformer]
\label{thm:eet-cont}
For every program $\stmt$ and every $\omega$--chain of run--times $f_1 \preceq f_2 \preceq \cdots$, 
\[ 
\eet{\stmt}{\supn f_n} ~=~ \supn \eet{\stmt}{f_n}~. 
\]
\end{lemma}
\begin{proof}
By induction on the structure of $C$; see Appendix \ref{sec:app-eet-cont}.
\qed
\end{proof}
\autoref{thm:eet-cont} implies that for each program $C \in \pProgs$, guard
$\pguard \in \DExprs$, and run--time $f \in \E$, function $F_f(X) = \ctert{1} +
\probof{\pguard}{\false} \cdot f + \probof{\pguard}{\true} \cdot \eet{C}{X}$ is
also continuous.  The Kleene Fixed Point Theorem then ensures that the least
fixed point $\eet{\WHILEDO{\pguard}{C}}{f}\allowbreak = \allowbreak\lfp F_f$
exists and the expected run--time of loops is thus well-defined.  

Finally, as the aforementioned function $F_f$ is frequently used in the remainder of the
paper, we define:
\begin{definition}[Characteristic functional of a loop]\label{def:F}
Given program $C \in \pProgs$, probabilistic guard $\pguard \in \DExprs$, and run--time $f \in \E$, we call
\[
F_f^{\langle \pguard, C \rangle} \colon \E \To \E,~ X \mapsto \ctert{1} +
\probof{\pguard}{\false} \cdot f + \probof{\pguard}{\true} \cdot \eet{C}{X}
\]
the \emph{characteristic functional} of loop $\WHILEDO{\pguard}{C}$ with
respect to $f$.
\end{definition}
When $C$ and $\pguard$ are understood from the context, we usually omit them and
simply write $F_f$ for the characteristic functional associated to
$\WHILEDO{\pguard}{C}$ with respect to run--time $f$.  Observe that under this
definition, the \eetsymbol\xspace of loops can be recast as
\[
\eet{\WHILEDO{\pguard}{C}}{f} ~=~ \lfp F_f^{\langle \pguard, C \rangle}~.
\]
This concludes our presentation of the \eetsymbol\ transformer. In the next
section we validate the transformer's definition by showing a soundness result
with respect to an operational model of programs.

\section{An Operational Model for Expected Run--Times}
\label{sec:operational}
We prove the soundness of the expected run--time transformer with respect to a
simple operational model for our probabilistic programs.  This model will be
given in terms of a Markov Decision Process (MDP, for short) whose collected
reward corresponds to the run--time.  We first briefly recall all necessary
notions.  A more detailed treatment can be found in~\cite[Ch.\ 10]{katoenbaier}.
A \emph{Markov Decision Process} is a tuple $\rmdp = (\mdpStates,\, \mdpAct,\, \mdpTPF,\, \mdpInit,\, \mdpRew)$ where
$\mdpStates$ is a countable set of states,
$\mdpAct$ is a (finite) set of actions,
$\mdpTPF \colon \mdpStates \times \mdpAct \times \mdpStates \To \ZO$ is the transition probability function such that for all 
        states $s \in \mdpStates$ and actions $\alpha \in \mdpAct$,
        \[ \sum_{s' \in \mdpStates} \mdpTPF(s,\alpha,s') \in \{0,1\}~, \]
$\mdpInit \in \mdpStates$ is the initial state, and
$\mdpRew \colon \mdpStates \To \PReals$ is a reward function.
Instead of $\mdpTPF(s,\alpha,s') = p$, we usually write
$\transb{s}{\alpha}{s'}{p}$. 
%
%
%
%
An MDP $\rmdp$ is a \emph{Markov chain} if no non--deterministic choice is
possible, i.e. for each pair of states $s,s' \in \mdpStates$ there exists exactly one $\alpha \in \mdpAct$ with $\mdpTPF(s,\alpha,s') \neq 0$.
A \emph{scheduler} for $\rmdp$ is a mapping $\mdpSched \colon \mdpStates^{+} \to \mdpAct$, where $\mdpStates^{+}$ denotes the set of non--empty finite sequences of states.
Intuitively, a scheduler resolves the non--determinism of an MDP by selecting an action for each possible sequence of states that has been visited so far.
Hence, a scheduler $\mdpSched$ induces a Markov chain which is denoted by $\rmdpMc$.
In order to define the expected reward of an MDP, we first consider the reward collected along a path.
Let $\mdpPaths{\rmdpMc}$ ($\mdpFPaths{\rmdpMc}$) denote the set of all (finite) paths $\ipath$ ($\fpath$) in $\rmdpMc$.
Analogously, let $\mdpPaths{\rmdpMc}(s)$ and $\mdpFPaths{\rmdpMc}(s)$ denote the set of all infinite and finite 
paths in $\rmdpMc$ starting in state $s \in \mdpStates$, respectively.
For a finite path $\fpath = s_0\ldots s_n$, the \emph{cumulative reward} of $\fpath$ is defined as 
\[ \mdpRew(\fpath) ~\triangleq~ \sum_{k=0}^{n-1} \mdpRew(s_k)~. \]
For an infinite path $\ipath$, the cumulative reward of reaching a non--empty set of target states $T \subseteq \mdpStates$, is defined as
$\mdpRew(\ipath,\mdpEventually T) \triangleq \mdpRew(\ipath(0)\ldots \ipath(n))$ if there exists an $n$ such that $\ipath(n) \in T$ and $\ipath(i) \notin T$ for $0 \leq i < n$
and $\mdpRew(\ipath,\mdpEventually T) \triangleq \infty$ otherwise.
Moreover, we write $\mdpSPaths{s}{T}$ to denote the set of all finite paths $\fpath \in \mdpFPaths{\rmdpMc}(s)$, $s \in \mdpStates$,
with $\fpath(n) \in T$ for some $n \in \Nats$ and $\fpath(i) \notin T$ for $0 \leq i < n$.
The probability of a finite path $\fpath$ is 
\[ \mdpPr{\rmdpMc}{\fpath} ~\triangleq~ \prod_{k=0}^{|\fpath|-1} \mdpTPF(s_k, \mdpSched(s_1,\ldots,s_k), s_{k+1})~. \]
The \emph{expected reward} that an MDP $\rmdp$ eventually reaches a non--empty
set of states $T \subseteq \mdpStates$ from a state $s \in \mdpStates$ is
defined as follows.  If
\[ \inf_{\mathfrak S} \: \mdpPr{\rmdpMc}{s \models \mdpEventually T} ~=~ \inf_{\mathfrak S} \: \sum_{\fpath \in \mdpSPaths{s}{T}} \mdpPr{\rmdpMc}{\fpath} ~<~ 1 \]
then $\ExpRew{\rmdp}{s \models \mdpEventually T} \triangleq \infty$.
Otherwise, 
       \begin{align*}
         \ExpRew{\rmdp}{s \models \mdpEventually T}
         ~\triangleq~ 
         \sup_{\mdpSched} \sum_{\fpath \in \mdpSPaths{s}{T}} \mdpPr{\rmdpMc}{\fpath} \cdot \mdpRew(\fpath)~.
       \end{align*}
%
%
%
%
We are now in a position to define an operational model for our probabilistic programming language.
Let $\Terminated$ denote a special symbol indicating successful termination of a program.
\begin{definition}[The operational MDP of a program]\label{def:operational}
 Given program $\stmt \in \Stmt$, initial program state $\ps_0 \in
 \ProgramStates$, and continuation 
 $\mdpRewSink \in \E$, the \emph{operational MDP} of $C$ is given by
 $\rmdpStmt{\mdpRewSink}{\ps_0}{\stmt} = (\mdpStates,\, \mdpAct,\, \mdpTPF,\,
 \mdpInit,\, \mdpRew)$, where
 \begin{itemize}
  \item $\mdpStates \triangleq ((\Stmt \cup \{ \Terminated \} \cup \{\Terminated;C ~|~ C \in \Stmt\}) \times \ProgramStates) \cup \{ \mdpSinkState \}$, 
  \item $\mdpAct \triangleq \{\actL,\, \actN,\, \actR \}$,
  \item the transition probability function $\mdpTPF$ is given by the rules in \autoref{fig:transition-function},
  \item $\mdpInit \triangleq \mdpState{\stmt}{\ps_0}$, and
  \item $\mdpRew \colon \mdpStates \To \PReals$ is the reward function defined according to \autoref{table:rew-rules}.
 \end{itemize}
\end{definition}
\begin{figure}[tb]
\begin{tabular}{p{\textwidth - 1ex}}
         $\infrule
		{}
		{\transb{\mdpState{\Terminated}{\ps}}{\actN}{\mdpSinkState}{1}}
		~\lrule{\textrm{terminated}}$
        \hfill
	$\infrule
		{}
		{\transb{\mdpSinkState}{\actN}{\mdpSinkState}{1}}
		~\lrule{\textrm{sink}}$
	\\[4ex]
	$\infrule
		{}
		{\trans{\EMPTY}{\ps}{\actN}{\Terminated}{\ps}{1}}
		~\lrule{\textrm{empty}}$
        \hfill
	$\infrule
		{}
		{\trans{\SKIP}{\ps}{\actN}{\Terminated}{\ps}{1}}
		~\lrule{\textrm{skip}}$
	\\[4ex]
	$\infrule
		{}
		{\mdpState{\HALT}{\ps} \xrightarrow{\actN} \mdpSinkState \vdash 1}
		~\lrule{\textrm{halt}}$
        \hfill
	$\infrule
		{\probofInstance{\pexpr}{v}{\ps} = p > 0}
		{\trans{\APPASSIGN{x}{\pexpr}}{\ps}{\actN}{\Terminated}{\ps\subst{x}{v}}{p}}
		~\lrule{\textrm{\textrm{pr--assgn}}}$
	\\[4ex]
	$\infrule
		{\trans{\stmt_1}{\ps}{\alpha}{\stmt_1'}{\ps'}{p},~ \alpha \in
          \mdpAct \quad 0 < p \leq 1}
		{\trans{\stmt_1;\stmt_2}{\ps}{\alpha}{\stmt_1';\stmt_2}{\ps'}{p}}
		~\lrule{\textrm{seq$_1$}}$
        \hfill
	$\infrule
		{}
		{\trans{\Terminated;\stmt_2}{\ps}{\actN}{\stmt_2}{\ps}{1}}
		~\lrule{\textrm{seq$_2$}}$
	\\[4ex]
	$\infrule
		{}
		{\trans{\NDCHOICE{\stmt_1}{\stmt_2}}{\ps}{\actL}{\stmt_1}{\ps}{1}}
		~\lrule{$\Box$--\textrm{\textit{L}}}$
	\hfill
	$\infrule
		{}
		{\trans{\NDCHOICE{\stmt_1}{\stmt_2}}{\ps}{\actR}{\stmt_2}{\ps}{1}}
		~\lrule{$\Box$--\textrm{\textit{R}}}$
	\\[4ex]
	$\infrule
		{\probofInstance{\pguard}{\true}{\ps} = p > 0}
		{\trans{\ITE{\pguard}{\stmt_1}{\stmt_2}}{\ps}{\actN}{\stmt_1}{\ps}{p}}
		~\lrule{\textrm{\textrm{if--true}}}$
         \\[4ex]
	$\infrule
		{\probofInstance{\pguard}{\false}{\ps} = p > 0}
		{\trans{\ITE{\pguard}{\stmt_1}{\stmt_2}}{\ps}{\actN}{\stmt_2}{\ps}{p}}
		~\lrule{\textrm{\textrm{if--false}}}$
	\\[4ex]
	$\infrule
		{}
		{\trans{\WHILEDO{\pguard}{\stmt}}{\ps}{\actN}{\ITE{\pguard}{\stmt;\WHILEDO{\pguard}{\stmt}}{\EMPTY}}{\ps}{1}}
		~\lrule{\textrm{\textrm{while}}}$\\[4ex]
  \end{tabular}
\hrule
\normalsize
  \caption{Rules for the transition probability function of operational MDPs.}
  \label{fig:transition-function}
  \end{figure}
\begin{table}[tb]
\begin{center}
\renewcommand{\arraystretch}{1.5}
\begin{tabular}{@{\qquad}l@{\hspace{2em}}l@{\qquad}}
\hline\hline
$\boldsymbol{s}$ 	& $\boldsymbol{\mdpRew(s)}$\\
\hline\hline
$\mdpState{\Terminated}{\ps}$ & $f(\sigma)$\\
$\mdpState{\SKIP}{\ps}$, $\mdpState{\APPASSIGN{x}{\pexpr}}{\ps}$, $\mdpState{\ITE{\pguard}{\stmt_1}{\stmt_2}}{\ps}$ & $1$	\\
$\mdpSinkState$, $\mdpState{\EMPTY}{\ps}$, $\mdpState{\HALT}{\ps}$, $\mdpState{\COMPOSE{\Terminated}{\stmt_2}}{\ps}$, & $0$\\[-1.15ex]
$\mdpState{\NDCHOICE{\stmt_1}{\stmt_2}}{\ps}$, $\mdpState{\WHILEDO{\pguard}{\stmt}}{\ps}$	&  \\
$\mdpState{\COMPOSE{\stmt_1}{\stmt_2}}{\ps}$ & $\mdpRew(\mdpState{\stmt_1}{\ps})$ \\[.8ex]
\hline
\end{tabular}
\end{center}
\caption{Definition of the reward function $\mdpRew \colon \mdpStates \To \PReals$ of operational MDPs.}
\label{table:rew-rules}
\vspace{-1\baselineskip}
\end{table}
Since the initial state of the MDP $\rmdpStmt{\mdpRewSink}{\ps_0}{\stmt}$ of a program $\stmt$ with initial state $\ps_0$ is uniquely given, instead of $\ExpRew{\rmdpStmt{\mdpRewSink}{\ps_0}{\stmt}}{\mdpState{\stmt}{\ps_0} \models \mdpEventually T}$ we simply write
\[ \ExpRew{\rmdpStmt{\mdpRewSink}{\ps}{\stmt}}{T}~. \]
The rules in \autoref{fig:transition-function} defining the transition
probability function of a program's MDP are self--explanatory.  Since only guard
evaluations, assignments and $\SKIP$ statements are assumed to consume time,
i.e. have a positive reward, we assign a reward of $0$ to all other program
statements.  Moreover, note that all states of the form
$\mdpState{\EMPTY}{\ps}$, $\mdpState{\Terminated}{\ps}$ and $\mdpSinkState$ are
needed, because an operational MDP is defined with respect to a given
continuation $\mdpRewSink \in \E$.  In case of $\mdpState{\EMPTY}{\ps}$, a
reward of $0$ is collected and after that the program successfully terminates,
i.e. enters state $\mdpState{\Terminated}{\ps}$ where the continuation
$\mdpRewSink$ is collected as reward.  In contrast, since no state other than
$\mdpSinkState$ is reachable from the unique sink state $\mdpSinkState$, the
continuation $\mdpRewSink$ is not taken into account if $\mdpSinkState$ is
reached without reaching a state $\mdpState{\Terminated}{\ps}$ first.
Hence the operational MDP directly enters $\mdpSinkState$ from a state of the
form $\mdpState{\HALT}{\ps}$.
\begin{example}[MDP of $C_\mathit{trunc}$] \label{example:rmdp}
 Recall the probabilistic program $C_\mathit{trunc}$ from \autoref{example:eet}.
\autoref{fig:rmdp_example} depicts the MDP $\rmdpStmt{\mdpRewSink}{\ps}{C_\mathit{trunc}}$ for an arbitrary fixed state $\ps \in \ProgramStates$ and an arbitrary continuation $\mdpRewSink \in \E$.
Here labeled edges denote the value of the transition probability function for the respective states, while the reward of each state is provided in gray next to the state.
To improve readability, edge labels are omitted if the probability of a transition is $1$.
Moreover, $\rmdpStmt{\mdpRewSink}{\ps}{C_\mathit{trunc}}$ is a Markov chain, because $C_\mathit{trunc}$ contains no non-deterministic choice.

A brief inspection of \autoref{fig:rmdp_example} reveals that $\rmdpStmt{\mdpRewSink}{\ps}{C_\mathit{trunc}}$ contains three finite paths $\fpath_{\true}$, $\fpath_{\false~\true}$, $\fpath_{\false~\false}$ that eventually reach state $\mdpSinkState$ starting from the initial state $\mdpState{\stmt_\mathit{trunc}}{\ps}$.
These paths correspond to the results of the two probabilistic guards in $\stmt$.
Hence the expected reward of $\rmdpStmt{\mdpRewSink}{\ps}{C}$ to eventually reach $T = \{ \mdpSinkState \}$ is given by
\begin{flalign*}
\MoveEqLeft[2] \ExpRew{\rmdpStmt{\mdpRewSink}{\ps}{\stmt_\mathit{trunc}}}{T} 
\\
   &~=~ \textstyle{\sup_{\mdpSched} \: \sum\nolimits_{\fpath \in \mdpSPaths{s}{T}} \mdpPr{\rmdpMc}{\fpath} \cdot \mdpRew(\fpath)}\\
   & ~=~ \textstyle{\sum\nolimits_{\fpath \in \mdpSPaths{s}{T}} \mdpPr{\rmdp}{\fpath} \cdot \mdpRew(\fpath)} \tag{$\rmdpStmt{\mdpRewSink}{\ps}{C_\mathit{trunc}} = \rmdp$ is a Markov chain} \\
   & ~=~ \mdpPr{\rmdp}{\fpath_{\true}} \cdot \mdpRew(\fpath_{\true})
     ~+~ \mdpPr{\rmdp}{\fpath_{\false~\true}} \cdot \mdpRew(\fpath_{\false~\true}) \\
   & \quad ~+~ \mdpPr{\rmdp}{\fpath_{\false~\false}} \cdot \mdpRew(\fpath_{\false~\false}) \\
   & ~=~ \bigl(\tfrac 1 2 \cdot 1 \cdot 1 \bigr) \cdot (1+1+\mdpRewSink(\ps\subst{\textit{succ}}{\true})) \\
   & \quad ~+~ \bigl(\tfrac 1 2 \cdot \tfrac 1 2 \cdot 1 \cdot 1 \bigr) \cdot (1+1+1+\mdpRewSink(\ps\subst{\textit{succ}}{\true})) \\
   & \quad ~+~ \bigl(\tfrac 1 2 \cdot \tfrac 1 2 \cdot 1 \cdot 1 \bigr) \cdot (1+1+1+\mdpRewSink(\ps\subst{\textit{succ}}{\false})) \\
   & ~=~ 1 + \tfrac 1 2 \cdot \mdpRewSink(\ps\subst{\textit{succ}}{\true})  ~+~  \tfrac 1 4 \cdot (6 + \mdpRewSink(\ps\subst{\textit{succ}}{\true}) \\
   & \quad ~+~ \mdpRewSink(\ps\subst{\textit{succ}}{\false})) \\
   & ~=~ \tfrac 5 2 ~+~ \tfrac 3 4 \cdot \mdpRewSink(\ps\subst{\textit{succ}}{\true}) ~+~ \tfrac 1 4 \cdot \mdpRewSink(\ps\subst{\textit{succ}}{\false}).
\end{flalign*}
Observe that for $\mdpRewSink = \ctert{0}$, the expected reward
$\ExpRew{\rmdpStmt{\mdpRewSink}{\ps}{\stmt_\mathit{trunc}}}{T}$ and the expected
run--time $\eet{\stmt}{f}(\sigma)$ (cf.~\autoref{example:eet}) coincide, both
yielding $\nicefrac 5 2$.  \hfill$\triangle$
\end{example}
\begin{figure}[tb]
 \begin{center}
   \begin{tikzpicture}[->,>=stealth',node distance=2cm,semithick,minimum size=1cm]
\tikzstyle{every state}=[draw=none]
   
   \node [state, initial, initial text=,] (init) {$\mdpState{C}{\ps}$};  
   \node (initLabel) [node distance=0.5cm,gray,above of=init] {$1$};
   
   \node [state,] (psdummy) [below of = init] {};
   
   \node [state,] (pdummy) [left of = psdummy] {};
   \node [state,] (p) [left of = psdummy] {$\mdpState{\ASSIGN{\mathit{succ}}{\true}}{\ps}$};
   \node (pLabel) [node distance=1.5cm,gray,left of=p] {$1$};
   \node (ps) [right of = pdummy, node distance = 0.7cm] {};
   
   
   \node [state,] (rdummy) [below of = p] {};
   \node [state,] (r) [below of = p] {$\mdpState{\Terminated}{\ps\subst{\mathit{succ}}{\true}}$};
   \node (rLabel) [node distance=2.7cm,gray,left of=r] {$\mdpRewSink(\ps\subst{\mathit{succ}}{\true})$};
   
   \node [state,] (sdummy) [right of = psdummy] {};
   \node [state,] (s) [right of = psdummy] {$\mdpState{C'}{\ps}$};
   \node (sLabel) [node distance=1cm,gray,right of=s] {$1$};
   \node (sp) [left of = sdummy, node distance = 0.1cm] {};
   
   
   \node [state,] (s0dummy) [below of = s] {};
   \node [state,] (s0) [below of = s] {$\mdpState{\ASSIGN{\mathit{succ}}{\false}}{\ps}$};
   \node (s0Label) [node distance=1.5cm,gray,right of=s0] {$1$};
   
   \node [state,] (t0dummy) [below of = s0] {};
   \node [state,] (t0) [below of = s0] {$\mdpState{\Terminated}{\ps\subst{\mathit{succ}}{\false}}$};
   \node (rLabel) [node distance=2.7cm,gray,right of=t0] {$\mdpRewSink(\ps\subst{\mathit{succ}}{\false})$};
   \node [state,] (ts) [left of = t0dummy, node distance = 1cm] {};
   
   \node [state,] (sinkdummy) [below of = r] {};
   \node [state,] (sink) [below of = r] {$\mdpSinkState$};
   \node (sinkLabel) [node distance=1cm,gray,left of=sink] {$0$};

  \path
      (init) edge [] node [left] {\scriptsize{$\nicefrac 1 2$}} (pdummy)
      (init) edge [] node [right] {\scriptsize{$\nicefrac 1 2$}} (sdummy)
      (pdummy) edge []  (rdummy)
      (rdummy) edge []  (sinkdummy)
      (ts) edge []  (sink)
      (sinkdummy) edge [loop below] (sinkdummy)
      (sdummy) edge [] node [right] {\scriptsize{$\nicefrac 1 2$}} (s0dummy)
      (sp) edge [] node [above] {\scriptsize{$\nicefrac 1 2$}} (ps)
      (s0dummy)  edge [] (t0dummy)
      
  ;
\end{tikzpicture}
 \end{center}
  \vspace*{-25pt}
 \caption{The operational MDP $\rmdpStmt{\mdpRewSink}{\ps}{C_\mathit{trunc}}$
   corresponding to the program in \autoref{example:rmdp}. $C'$ denotes the
   subprogram $\ITE{\nicefrac 1 2 \cdot \langle \true \rangle + \nicefrac 1 2 \cdot \langle \false
 \rangle}{\ASSIGN{\mathit{succ}}{\true}}{\ASSIGN{\mathit{succ}}{\false}}$.}
 \label{fig:rmdp_example}
\end{figure}
%
%
%
%
%
%
The main result of this section is that $\eetsymbol$ precisely captures the expected reward
of the MDPs associated to our probabilistic programs.
\begin{theorem}[Soundness of the \boldeetsymbol{} transformer]\label{thm:eet-soundness}
Let $\pguard \in \DExprs$, $\stmt \in \Stmt$, and $\mdpRewSink \in \E$. Then, for each $\ps \in \ProgramStates$, we have
  \[ 
      \ExpRew{\rmdpStmt{\mdpRewSink}{\ps}{\stmt}}{\mdpSinkState} ~=~ \eet{\stmt}{\mdpRewSink}(\ps)~.
  \]
\end{theorem}
\begin{proof}
  By induction on the program structure.
  See Appendix~\ref{sec:thm:eet-soundness}.
\qed
\end{proof}
%
%
%
%

\section{Expected Run--Time of Loops}
\label{sec:loopy-rules}
Reasoning about the run--time of loop--free programs consists mostly of
syntactic reasoning.  The run--time of a loop, however, is given in terms of a
least fixed point.  It is thus obtained by fixed point iteration but need not be
reached within a finite number of iterations. To overcome this problem we next
study invariant--based proof rules for approximating the run--time of loops.

We present two families of proof rules which differ in the kind of the
invariants they build on.  In Section~\ref{sec:glob-inv} we present a proof rule
that rests on the presence of an invariant approximating the entire run--time of
a loop in a global manner, while in Section~\ref{sec:inc-inv} we present two
proof rules that rely on a parametrized invariant that approximates the
run--time of a loop in an incremental fashion. Finally in
Section~\ref{sec:ref-bounds} we discuss how to improve the run--time bounds
yielded by these proof rules.

\subsection{Proof Rule Based on Global Invariants}
\label{sec:glob-inv}

The first proof rule that we study allows upper--bounding the expected run--time of
 loops and rests on the notion of \emph{upper invariants}.
\begin{definition}[Upper invariants]
\label{def:right-invariants}
Let $f \in \E$, $C \in \appProgs$ and $\pguard \in \DExprs$. We say that $I \in
\E$ is an \emph{upper invariant} of loop $\WHILEDO{\pguard}{C}$ with respect to
$f$ iff
\begin{align*}
 \ctert{1} + \probof{\pguard}{\false} \cdot f + \probof{\pguard}{\true} \cdot
 \eet{C}{I} ~\preceq I
\end{align*} 
or, equivalently, iff $F_f^{\langle \pguard, C \rangle} (I) \preceq I$, where $F_f^{\langle \pguard, C \rangle}$ is the characteristic functional.
\end{definition}
The presence of an upper invariant of a loop readily establishes an upper bound of the loop's run--time.
\begin{theorem}[Upper bounds from upper invariants]
\label{thm:upper-bounds}
Let $f \in \E$, $C \in \appProgs$ and $\pguard \in \DExprs$. If $I \in \E$ is an
upper invariant of $\WHILEDO{\pguard}{C}$ with respect to $f$ then
\begin{align*}
	\eet{\WHILEDO{\pguard}{C}}{f} ~\preceq~ I~.
\end{align*}
\end{theorem}
%
%
%
\begin{proof}
	The crux of the proof is an application of Park's Theorem\footnote{If $H\colon \mathcal{D} \To \mathcal{D}$ is a continuous function over an $\omega$--cpo $(\mathcal{D},\sqsubseteq)$ with bottom element, then $H(d) \sqsubseteq d$ implies $\lfp H \sqsubseteq d$ for every $d \in \mathcal{D}$.}~\cite{Wechler:MTCS:92} which, given that $F_f^{\langle \pguard, C \rangle}$ is continuous (see \autoref{thm:eet-cont}), states that
	\[
		F_f^{\langle \pguard, C \rangle} (I) \preceq I ~\implies~ \lfp F_f^{\langle \pguard, C \rangle}\preceq I~.
	\]
	The left--hand side of the implication stands for $I$ being an upper invariant, while the right--hand side stands for $\eet{\WHILEDO{\pguard}{C}}{f} \preceq I$.\qed
\end{proof}
Notice that if the loop body $C$ is itself loop--free, it is usually fairly easy
to verify that some $I \in \E$ is an upper invariant, whereas \emph{inferring}
the invariant is---as in standard program verification---one of the most
involved part of the verification effort.
\begin{example}[Geometric distribution]\label{ex:geo-up}
Consider loop 
\[
C_{\mathtt{geo}} \boldsymbol{\colon} \;\;  \WHILEDO{c =
  1}{\APPASSIGN{c}{\nicefrac 1 2 \cdot
    \langle 0 \rangle + \nicefrac 1 2 \cdot  \langle 1 \rangle}}~.
\]
From the calculations below we conclude that $I = \ctert{1} + \eval{c = 1} \cdot
\ctert{4}$ is an upper invariant with respect to $\ctert{0}$:
\begin{flalign*}
\MoveEqLeft[2] \ctert{1} + \eval{c \neq 1} \cdot \ctert{0} + \eval{c = 1}  \cdot
 \eet{\APPASSIGN{c}{\nicefrac 1 2 \cdot
    \langle 0 \rangle + \nicefrac 1 2 \cdot \langle 1 \rangle}}{I}\\
&=~\ctert{1} + \eval{c = 1} \cdot \bigl(\ctert{1} + \tfrac{1}{2} \cdot
I\subst{c}{0} +  \tfrac{1}{2} \cdot I\subst{c}{1} \bigr)  &\\
&=~\ctert{1} + \eval{c = 1} \cdot \bigl(\ctert{1} + \tfrac{1}{2} \cdot (\underbrace{\ctert{1} + \eval{0 = 1} \cdot \ctert{4}}_{{=}\,\ctert{1}}) +  \tfrac{1}{2} \cdot (\underbrace{\ctert{1} + \eval{1 = 1} \cdot \ctert{4}}_{{=}\,\ctert{5}})\bigr)  &\\[-1ex]
&=~\ctert{1} + \eval{c = 1} \cdot \ctert{4} ~=~I ~\preceq~ I
\end{flalign*}
Then applying \autoref{thm:upper-bounds} we obtain
\[
\eet{C_{\mathtt{geo}}}{\ctert{0}} ~\preceq~\ctert{1} + \eval{c = 1} \cdot \ctert{4}~.
\]
In words, the expected run--time of $C_{\mathtt{geo}}$ is at most $5$ from
any initial state where $c=1$ and at most $1$ from the remaining
states. \hfill$\triangle$
\end{example}
The invariant--based technique to reason about the run--time of loops
presented in \autoref{thm:upper-bounds} is complete in the sense that there
always exists an upper invariant that establishes the exact run--time of
the loop at hand. 
\begin{theorem}
\label{thm:upp-inv-compl}
  Let $f \in \E$, $C \in \appProgs$, $\pguard \in \DExprs$. Then there exists an
  upper invariant $I$ of $\WHILEDO{\pguard}{C}$ with respect to $f$ such that
  ${\eet{\WHILEDO{\pguard}{C}}{f}=I}$.
\end{theorem}
\begin{proof} The result follows from showing that
  $\eet{\WHILEDO{\pguard}{C}}{f}$ is itself an upper invariant. Since
  $\eet{\WHILEDO{\pguard}{C}}{f} = \lfp F_f^{\langle \pguard, C \rangle}$ this
  amounts to showing that
  \[
  F_f^{\langle \pguard, C \rangle}\bigl(\, \lfp
  F_f^{\langle \pguard, C \rangle} \,\bigr) ~\preceq~ \lfp F_f^{\langle \pguard, C
    \rangle}~,
  \]
  which holds by definition of $\lfp$.\qed
\end{proof}
Intuitively, the proof of this theorem shows that
$\eet{\WHILEDO{\pguard}{C}}{f}$ itself is the tightest upper invariant that the
loop  admits.
\subsection{Proof Rules Based on Incremental Invariants}
\label{sec:inc-inv}

We now study a second family of proof rules which builds on the notion of
$\omega$--invariants to establish \emph{both} upper and lower bounds for the run--time of loops.
\begin{definition}[$\omega$--invariants]
\label{def:omega-invariant}
Let $f \in \E$, $C \in \appProgs$ and $\pguard \in \DExprs$. Moreover let $I_n
\in \E$ be a run--time parametrized by $n \in \Nats$.  We say that
$I_n$ is a \emph{lower $\omega$--invariant} of loop $\WHILEDO{\pguard}{C}$ with
respect to $f$ iff
\[
 F_f^{\langle \pguard, C \rangle} (\ctert{0}) \succeq I_0%
 \qquad\text{and}\qquad
 F_f^{\langle \pguard, C \rangle} (I_n) \succeq I_{n+1} \quad\text{for all $n
   \geq 0$}~.
\]
Dually, we say that $I_n$ is an \emph{upper $\omega$--invariant} iff
\[
 F_f^{\langle \pguard, C \rangle} (\ctert{0}) \preceq I_0%
 \qquad\text{and}\qquad
 F_f^{\langle \pguard, C \rangle} (I_n) \preceq I_{n+1} \quad\text{for all $n
   \geq 0$}~.
\]
\end{definition}
Intuitively, a lower (resp.\ upper) $\omega$--invariant $I_n$ represents a lower
(resp.\ upper) bound for the expected run--time of those program runs that finish within
${n+1}$ iterations, weighted according to their probabilities. Therefore we can
use the asymptotic behavior of $I_n$ to approximate from below (resp.\ above)
the expected run--time of the entire loop.
%
\begin{theorem}[Bounds from $\omega$--invariants]
\label{thm:omega-inv}
Let $f \in \E$, $C \in \appProgs$, ${\pguard \in \DExprs}$.
\begin{enumerate}
\item \label{thm:omega-low-inv} If $I_n$ is a lower $\omega$--invariant of
  $\WHILEDO{\pguard}{C}$ with respect to $f$ and $\lim\limits_{n \To \infty} I_n$
  exists\footnote{Limit $\lim_{n \To \infty} I_n$ is to be understood pointwise,
    on $\Rposinf$, \ie $\lim_{n \To \infty} I_n = \lambda \sigma. \lim_{n \To
      \infty} I_n(\sigma)$ and $\lim_{n \To \infty} I_n(\sigma) = \infty$ is
    considered a valid value.}, then
\[
\eet{\WHILEDO{\pguard}{C}}{f} \succeq \lim_{n \To \infty} I_n~.
\]
\item \label{thm:omega-upp-inv} If $I_n$ is an upper $\omega$--invariant of
  $\WHILEDO{\pguard}{C}$ with respect to $f$ and $\lim\limits_{n \To \infty} I_n$
  exists, then
\[
\eet{\WHILEDO{\pguard}{C}}{f} \preceq \lim_{n \To \infty} I_n~.
\]
\end{enumerate}
\end{theorem}
\begin{proof}
  We prove only the case of lower $\omega$--invariants since the other case
  follows by a dual argument. Let $F_f$ be the characteristic functional of the
  loop with respect to $f$. 
  Let $F_f^0 = \ctert{0}$ and $F_f^{n+1} = F_f(F_f^{n})$. 
  By the Kleene Fixed Point Theorem,
  $\eet{\WHILEDO{\pguard}{C}}{f} = \sup_n F_f^n$ and since $F_f^0 \preceq F_f^1
  \preceq \ldots$ forms an $\omega$--chain, by the Monotone Sequence
  Theorem\footnote{If $\langle a_n \rangle_{n \in \Nats}$ is an increasing
    sequence in $\Rposinf$, then $\lim_{n \To \infty} a_n$
    coincides with supremum $\supn a_n$.}, $\sup_n F_f^n = \lim_{n \To \infty}
  F_f^n$. Then the proof follows from showing that $F_f^{n+1} \succeq I_n$. We
  prove this by induction on $n$. The base case $F_f^{1} \succeq I_0$ holds
  because $I_n$ is a lower $\omega$--invariant. For
  the inductive case we reason as follows:
\[
F_f^{n+2} =  F_f \bigl(F_f^{n+1}\bigr) \succeq F_f (I_n) \succeq I_{n+1}~.
\]
Here the first inequality follows by I.H.\ and the monotonicity of $F_f$ (recall
that $\eetsymbol[C]$ is monotonic by
\autoref{thm:eet-prop}), while the second inequality holds
because $I_n$ is a lower $\omega$--invariant.\qed
\end{proof}
\begin{example}[Lower bounds for $C_{\mathtt{geo}}$] Reconsider loop $C_{\mathtt{geo}}$ from \autoref{ex:geo-up}. 
Now we use \autoref{thm:omega-inv}.\ref{thm:omega-low-inv} to show that $\ctert{1} +
  \eval{c = 1} \cdot \ctert{4}$ is also a lower bound of its run--time. To
  this end we first show that $I_n = \ctert{1} + \eval{c = 1} \cdot
  \left(\ctert{4}-\nicefrac{\ctert{3}}{\ctert{2}^n}\right)$ is a lower
  $\omega$--invariant of the loop with respect to $\ctert{0}$:
\begin{flalign*}
F_{\ctert{0}}(\ctert{0})
&=~\ctert{1} + \eval{c \neq 1} \cdot \ctert{0} + \eval{c = 1}  \cdot
 \eet{\APPASSIGN{c}{\nicefrac 1 2
    \langle 0 \rangle + \nicefrac 1 2 \langle 1 \rangle}}{\ctert{0}}\\
&=~\ctert{1} + \eval{c = 1} \cdot \bigl(\ctert{1} + \tfrac{1}{2} \cdot
\ctert{0}\subst{c}{0} +  \tfrac{1}{2} \cdot \ctert{0}\subst{c}{1} \bigr)  \\
&=~\ctert{1} + \eval{c = 1} \cdot \ctert{1}  ~=~  \ctert{1} + \eval{c = 1} \cdot \left(\ctert{4}-\nicefrac{\ctert{3}}{\ctert{2}^0}\right) ~=~ I_0 ~\succeq~ I_0\\[2ex]
F_{\ctert{0}}(I_n)
&=~\ctert{1} + \eval{c \neq 1} \cdot \ctert{0} + \eval{c = 1}  \cdot
 \eet{\APPASSIGN{c}{\nicefrac 1 2
    \langle 0 \rangle + \nicefrac 1 2 \langle 1 \rangle}}{I_n}\\
&=~\ctert{1} + \eval{c = 1} \cdot \bigl(\ctert{1} + \tfrac{1}{2} \cdot
I_n\subst{c}{0} +  \tfrac{1}{2} \cdot I_n\subst{c}{1} \bigr)  \\
&=~\ctert{1} + \eval{c = 1} \cdot \Bigl( \ctert{1} + \tfrac{1}{2} \cdot (\ctert{1}+\ctert{0}) + \tfrac{1}{2} \cdot \bigl(\ctert{1} + \left(\ctert{4}-\tfrac{\ctert{3}}{\ctert{2}^n}\right) \bigr)  \Bigr) \\
&=~\ctert{1} + \eval{c = 1} \cdot
\left(\ctert{4}-\tfrac{\ctert{3}}{\ctert{2}^{n+1}}\right) ~=~I_{n+1} ~\succeq~ I_{n+1}
\end{flalign*}
Then from  \autoref{thm:omega-inv}.\ref{thm:omega-low-inv} we obtain
\[
\eet{C_{\mathtt{geo}}}{\ctert{0}} ~\succeq~ \lim_{n \To \infty}  \, \Bigl( \ctert{1} + \eval{c = 1} \cdot
  \bigl(\ctert{4}-\tfrac{\ctert{3}}{\ctert{2}^n}\bigr) \Bigr) ~=~ \ctert{1} + \eval{c = 1} \cdot \ctert{4}~.
\]
Combining this result with the upper bound $\eet{C_{\mathtt{geo}}}{\ctert{0}}
\preceq \ctert{1} + \eval{c = 1} \cdot \ctert{4}$ established in
\autoref{ex:geo-up} we conclude that $\ctert{1} + \eval{c = 1} \cdot \ctert{4}$
is the exact run--time of $C_{\mathtt{geo}}$.  Observe, however, that the above
calculations show that $I_n$ is both a lower and an upper $\omega$--invariant
(exact equalities $F_0(\ctert{0})=I_0$ and $F_0(I_n)=I_{n+1}$ hold). Then we can
apply
\autoref{thm:omega-inv}.\ref{thm:omega-low-inv}~and~\ref{thm:omega-inv}.\ref{thm:omega-upp-inv}
simultaneously to derive the exact run--time without having to resort to the
result from \autoref{ex:geo-up}. 

\paragraph{Invariant Synthesis.} In order to synthesize invariant
$I_n = \ctert{1} + \eval{c = 1} \cdot
\left(\ctert{4}-\nicefrac{\ctert{3}}{\ctert{2}^n}\right)$, 
we proposed template $I_n = \ctert{1} + \eval{c = 1} \cdot a_n$ and observed
that under this template the definition of lower $\omega$--invariant reduces to
$a_0 \leq 1$, $a_{n+1} \leq 2 + \tfrac{1}{2} a_n$, which is satisfied by
$a_n= 4 - \nicefrac{3}{2^n}$. \hfill$\triangle$
\end{example}
Now we apply \autoref{thm:omega-inv}.\ref{thm:omega-low-inv} to a program with infinite expected run--time. 
\begin{example}[Almost--sure termination at infinite expected run--time]
\label{ex:npast}
Recall the program from the introduction:
\begin{align*}
C \boldsymbol{\colon} \;\; 
& \linen{1}  \COMPOSE{\COMPOSE{\ASSIGN{x}{1}}{\ASSIGN{b}{1}}}{}\\
& \linen{2} \COMPOSE{\WHILEDO{b = 1}{\COMPOSE{\APPASSIGN{b}{\nicefrac 1 2 \langle 0 \rangle + \nicefrac 1 2 \langle 1 \rangle}}{\ASSIGN{x}{2x}}}}{}\\
& \linen{3} \WHILEDO{x > 0}{\ASSIGN{x}{x-1}}
\end{align*}
Let $C_i$ denote the $i$-th line of $C$.  We show that $\eet{C}{\ctert{0}}
\succeq \infty$.\footnote{Note that while this program terminates with
  probability one, the expected run--time to achieve termination is infinite.} Since
\[
\eet{C}{\ctert{0}} ~=~ \eet{C_1}{\eet{C_2}{\eet{C_3}{\ctert{0}}}}
\]
we start by showing that 
\[
\eet{C_3}{\ctert{0}} ~\succeq~ \ctert{1} + \eval{x > 0} \cdot 2x
\]
using lower $\omega$--invariant $J_n = \ctert{1} + \eval{n > x > 0}
\cdot 2x + \eval{x \geq n} \cdot (2n-1)$. 
We omit here the details of verifying that $J_n$ is a lower $\omega$--invariant. 
Next we show that
\begin{samepage}
\begin{align*}
	\eet{C_2}{\ctert{1} + \eval{x > 0} \cdot 2x} ~\succeq~ 	\ctert{1}	&{}+ \eval{b \neq 1} \cdot \bigl( \ctert{1} + \eval{x > 0} \cdot 2x \bigr)\\
	 												&{}+  \eval{b  = 1} \cdot \big(\ctert{7} + \eval{x > 0} \cdot \infinity\big)
\end{align*}
\end{samepage}
by means of the lower $\omega$--invariant
\[
I_n ~=~ \ctert{1} + \eval{b \neq 1} \cdot \bigl( \ctert{1}
+ \eval{x > 0} \cdot 2x \bigr) + \eval{b = 1} \cdot \Bigl( \ctert{7} -
\ctert{\tfrac{5}{2^n}} + n \cdot \eval{x > 0} \cdot 2x \Bigr)~.
\]
Let $F$ be the characteristic functional of loop $C_2$ with respect to $\ctert{1} + \eval{x > 0} \cdot 2x$.
The calculations to establish that $I_n$ is a lower $\omega$--invariant now go as follows:
\begin{samepage}
\begin{flalign*}
F(\ctert{0})
&=~\ctert{1} + \ExpToFun{b \neq 1} \cdot \bigl( \ctert{1} + \ExpToFun{x > 0} \cdot 2x \bigr)  \\
&\qquad\:  + \ExpToFun{b = 1} \cdot  \Bigl( \ctert{1} + \tfrac{1}{2} \cdot
(\ctert{1} + \ctert{0}\subst{x,b}{2x,0}) + \tfrac{1}{2} \cdot (\ctert{1} +
\ctert{0}\subst{x,b}{2x,1}) \Bigr)\\
&=~\ctert{1} + \ExpToFun{b \neq 1} \cdot \bigl( \ctert{1} + \ExpToFun{x > 0} \cdot 2x \bigr)  + \ExpToFun{b = 1} \cdot  \bigl( \ctert{1} + \tfrac{1}{2} \cdot
\ctert{1} + \tfrac{1}{2} \cdot \ctert{1} \bigr)\\
&=~\ctert{1} + \ExpToFun{b \neq 1} \cdot \bigl( \ctert{1} + \ExpToFun{x > 0}
\cdot 2x \bigr)  + \ExpToFun{b = 1} \cdot  \ctert{2} ~=~I_0 ~\succeq~ I_0\\[2ex]
F(I_n)
&=~\ctert{1} + \ExpToFun{b \neq 1} \cdot \bigl( \ctert{1} + \ExpToFun{x > 0} \cdot 2x \bigr)  \\
&\;\;  + \ExpToFun{b = 1} \!\cdot\!  \Bigl( \ctert{1} + \tfrac{1}{2} \cdot
(\ctert{1} + I_n\subst{x,b}{2x,0}) + \tfrac{1}{2} \cdot (\ctert{1} +
I_n\subst{x,b}{2x,1}) \Bigr)\\
&=~\ctert{1} + \ExpToFun{b \neq 1} \cdot \bigl( \ctert{1} + \ExpToFun{x > 0} \cdot 2x \bigr)  \\
&\;\;  + \ExpToFun{b = 1} \! \cdot\!  \Bigl( \ctert{1} + \tfrac{1}{2} \cdot
(\ctert{3} + \eval{2x > 0} \cdot 4x) + \tfrac{1}{2} \!\cdot\!  \Bigl( \ctert{9} -
\ctert{\tfrac{5}{2^n}} + n \cdot \eval{2x > 0} \cdot 4x \Bigr) \Bigr)\\
&=~\ctert{1} + \ExpToFun{b \neq 1} \cdot \bigl( \ctert{1} + \ExpToFun{x > 0} \cdot 2x \bigr)  \\
&\;\;  + \ExpToFun{b = 1} \!\cdot\!  \Bigl(\ctert{7} -
\ctert{\tfrac{5}{2^{n{+}1}}} + (n{+}1) \cdot \eval{x > 0} \cdot 2x \Bigr)\\
&=~I_{n+1} ~\succeq~ I_{n+1}
\end{flalign*}
\end{samepage}
Now we can complete the run--time analysis of program $C$:
\begin{flalign*}
\MoveEqLeft[3] \eet{C}{\ctert{0}} \\
=~& \eet{C_1}{\eet{C_2}{\eet{C_3}{\ctert{0}}}}\\
\succeq~& \eet{C_1}{\ctert{1} + \eval{b \neq 1} \cdot \bigl( \ctert{1} + \eval{x > 0} \cdot 2x \bigr) +  \eval{b  = 1} \cdot \big( \ctert{7} + \eval{x > 0} \cdot \boldsymbol{\infty}\big)}\\
=~& \eetsymbol[\ASSIGN{x}{1}]\Big(\eetsymbol[\ASSIGN{b}{1}]\Big(\ctert{1} + \eval{b \neq 1} \cdot \bigl( \ctert{1} + \eval{x > 0} \cdot 2x \bigr)\\[-1ex]
	&\qquad\qquad\qquad\qquad\quad\quad\:\: +  \eval{b  = 1} \cdot \big( \ctert{7} + \eval{x > 0} \cdot \boldsymbol{\infty}\big)\Big)\Big)\\
=~& \eet{ \ASSIGN{x}{1}}{\ctert{8} + \eval{x > 0} \cdot \boldsymbol{\infty}} ~=~\ctert 8 + \infinity ~=~ \boldsymbol{\infty}
\end{flalign*}
Overall, we obtain that the expected run--time of the program $C$ is infinite even though it terminates with probability one.
Notice furthermore that sub--programs $\WHILEDO{b = 1}{\COMPOSE{\APPASSIGN{b}{\nicefrac 1 2 \langle 0 \rangle + \nicefrac 1 2 \langle 1 \rangle}}{\ASSIGN{x}{2x}}}$ and $\WHILEDO{x > 0}{\ASSIGN{x}{x-1}}$ have expected run--time $\ctert{1} + \eval{b}\cdot \ctert{4}$ and $\ctert{1} + \eval{x > 0} \cdot 2 x$, respectively, i.e.\ both have a finite expected run--time.

\paragraph{Invariant synthesis.} In order to synthesize the $\omega$--invariant
$I_n$ of loop $C_2$ we propose the template $I_n =  \ctert{1} + \ExpToFun{b
  \neq 1} \cdot \bigl( \ctert{1} + \ExpToFun{x > 0} \cdot 2x \bigr)  +
\ExpToFun{b = 1} \cdot \bigl( a_n + b_n \cdot \ExpToFun{x > 0} \cdot 2x
\bigr)$ and from the definition of lower $\omega$--invariants we obtain $a_0 \leq
2$, $a_{n+1} \leq \nicefrac{7}{2} + \nicefrac{1}{2} \cdot a_n$ and $b_0 \leq 0$, $b_{n+1} \leq 1 + b_n$. These recurrences admit  solutions $a_n=7 -
\nicefrac{5}{2^n}$ and $b_n=n$. 
\hfill$\triangle$
\end{example}
As the proof rule based on upper invariants, the proof rules based on
$\omega$-invariants are also complete: Given loop $\WHILEDO{\pguard}{C}$ and
run--time $f$, it is enough to consider the $\omega$-invariant $I_n =
F_f^{n+1}$, where $F_f^{n}$ is defined as in the proof of
\autoref{thm:omega-inv} to yield the exact run--time
$\eet{\WHILEDO{\pguard}{C}}{f}$ from an application of \autoref{thm:omega-inv}.
We formally capture this result by means of the following theorem:
\begin{theorem}
\label{thm:w-inv-compl}
  Let $f \in \E$, $C \in \appProgs$ and $\pguard \in \DExprs$. Then there exists
  a (both lower and upper) $\omega$--invariant $I_n$ of $\WHILEDO{\pguard}{C}$
  with respect to $f$ such that ${\eet{\WHILEDO{\pguard}{C}}{f}= \lim_{n \To
      \infty} I_n}$.
\end{theorem}
\autoref{thm:w-inv-compl} together with \autoref{thm:upp-inv-compl} shows that
the set of invariant--based proof rules presented in this section are
complete. Next we study how to refine invariants to make the bounds that these
proof rules yield more precise.

\subsection{Refinement of Bounds}
\label{sec:ref-bounds}

An important property of both upper and lower bounds of the run--time of loops is that they can be easily refined by repeated application of the characteristic functional.
\begin{theorem}[Refinement of bounds]
\label{thm:upp-inv-ref}
Let $f \in \E$, $C \in \appProgs$ and $\pguard \in \DExprs$. If $I$ is an upper
(resp.\ lower) bound of $\eet{\WHILEDO{\pguard}{C}}{f}$ and $F_f^{\langle
  \pguard, C \rangle}(I) \preceq I$ (resp.\ $F_f^{\langle \pguard, C \rangle}(I)
\succeq I$), then $F_f^{\langle \pguard, C \rangle}(I)$ is also an upper (resp.\
lower) bound, at least as precise as $I$.
\end{theorem}
\begin{proof}
If $I$ is an upper bound of $\eet{\WHILEDO{\pguard}{C}}{f}$ we have $\lfp F_f^{\langle \pguard, C \rangle} \preceq I$. 
	Then from the monotonicity of $F_f^{\langle \pguard, C \rangle}$ (recall that \eetsymbol\xspace is monotonic by \autoref{thm:eet-prop}) and from $F_f^{\langle \pguard, C \rangle}(I) \preceq I$ we obtain
	\[
		\eet{\WHILEDO{\pguard}{C}}{f} ~=~ \lfp F_f^{\langle \pguard, C \rangle} ~=~ F_f^{\langle \pguard, C \rangle}(\lfp F_f^{\langle \pguard, C \rangle}) ~\preceq~ F_f^{\langle \pguard, C \rangle}(I) ~\preceq~I~,
	\]
	which means that $F_f^{\langle \pguard, C \rangle}(I)$ is also an upper
        bound, possibly tighter than $I$.  The case for lower bounds is
        completely analogous.\qed
\end{proof}
Notice that if $I$ is an upper invariant of $\WHILEDO{\pguard}{C}$ then $I$ fulfills all necessary conditions of \autoref{thm:upp-inv-ref}.
In practice, \autoref{thm:upp-inv-ref} provides a means of iteratively improving the precision of bounds yielded by Theorems \ref{thm:upper-bounds} and \ref{thm:omega-inv}, as for instance for upper bounds we have
\begin{align*}
	\eet{\WHILEDO{\pguard}{C}}{f} ~\preceq~ \cdots ~\preceq~ F_f^{\langle \pguard, C \rangle}\left(F_f^{\langle \pguard, C \rangle}(I)\right) ~ \preceq~ F_f^{\langle \pguard, C \rangle}(I) ~\preceq I~.
\end{align*}
If $I_n$ is an upper (resp.\ lower) $\omega$-invariant, applying
\autoref{thm:upp-inv-ref} requires checking that $F_f^{\langle \pguard, C
  \rangle}(L) \preceq L$ (resp.\ $F_f^{\langle \pguard, C \rangle}(L) \succeq
L$), where $L=\lim_{n \To \infty} I_n$. This proof obligation can be discharged
by showing that $I_n$ forms an $\omega$-chain, \ie that $I_n \preceq I_{n+1}$
for all $n \in \Nats$.

\section{Run--Time of Deterministic Programs}
\label{sec:eet-vs-nielson}
The notion of expected run--times as defined by $\eetsymbol$ is clearly
applicable to deterministic programs, \ie programs containing neither probabilistic guards nor probabilistic assignments nor non--deterministic choice operators.
We show that the \eetsymbol{} of deterministic programs coincides with the tightest upper bound
on the run--time that can be derived in an extension of Hoare logic~\cite{DBLP:journals/cacm/Hoare69} due to Nielson~\cite{Nielson:SCP:87,Nielson:UTCS:07}.

In order to compare our notion of \eetsymbol{} to the aforementioned calculus we restrict our programming language to the language 
$\DetStmt$ of deterministic programs considered in \cite{Nielson:UTCS:07} which is given by the following grammar:
\begin{align*}
	C ~~::=~~ 	& \SKIP ~\;|\;~ \ASSIGN{x}{E} ~\;|\;~ \COMPOSE{C}{C}
 				 ~\;|\;~ \ITE{\pguard}{C}{C} ~\;|\;~ \WHILEDO{\pguard}{C}~,
\end{align*}
where $E$ is a \emph{deterministic} expression and $\pguard$ is a \emph{deterministic} guard, i.e. 
$\eval{E}(\ps)$ and $\eval{\pguard}(\ps)$ are Dirac distributions for each $\ps \in \ProgramStates$.
For simplicity, we slightly abuse notation and
write $\eval{E}(\ps)$ to denote the unique value $v \in \Vals$ such that $\probofInstance{E}{v}{\ps} = 1$.

%
%
%
%
For deterministic programs, the MDP $\rmdpStmt{\ctert{0}}{\ps}{\stmt}$ of a program $\stmt \in \DetStmt$ and a program state $\ps \in \ProgramStates$
is a labeled transition system. 
In particular, if a terminal state of the form $\mdpState{\Terminated}{\ps'}$ is reachable from the initial state of $\rmdpStmt{\ctert{0}}{\ps}{\stmt}$,
it is unique. 
Hence we may capture the effect of a deterministic program by a partial function
$\opSem{\:\cdot\:}{\:\cdot\:}\colon \DetStmt \times \ProgramStates \pto \ProgramStates$ 
mapping each $\stmt \in \DetStmt$ and  $\ps \in \ProgramStates$ to a program state
$\ps' \in \ProgramStates$ if and only if there exists a state $\mdpState{\Terminated}{\ps'}$ that is reachable in the MDP $\rmdpStmt{\ctert{0}}{\ps}{\stmt}$ from the initial state 
$\mdpState{\stmt}{\ps}$. Otherwise, $\opSem{\stmt}{\ps}$ is undefined.
Nielson~\cite{Nielson:SCP:87,Nielson:UTCS:07} developed an extension of the
classical Hoare calculus for total correctness of programs in order to
establish additionally upper bounds on the run--time of programs.  Formally, a
\emph{correctness} \emph{property} is of the form
\[ \htriple{P}{\stmt}{E}{Q}~, \]
where $\stmt \in \DetStmt$, $E$ is a deterministic expression over the program variables, and $P,Q$ are (first--order) assertions.
Intuitively, $\htriple{P}{\stmt}{E}{Q}$ is valid, written $\nvalid \htriple{P}{\stmt}{E}{Q}$, if and only if there exists a natural number $k$
such that for each state $\ps$ satisfying the precondition $P$, the program $\stmt$ terminates after at most $k \cdot \eval{E}(\ps)$ steps in a state satisfying
postcondition $Q$.
In particular, it should be noted that $E$ is evaluated in the \emph{initial} state $\ps$.

Figure~\ref{fig:nielson-inference-rules} is taken verbatim from \cite{Nielson:UTCS:07} except for minor changes to match our notation.
Most of the inference rules are self--explanatory extensions of the standard Hoare calculus for total correctness
of deterministic programs~\cite{DBLP:journals/cacm/Hoare69} which is obtained by omitting the gray parts.

The run--time of $\SKIP$ and $\ASSIGN{x}{E}$ is one time unit.
Since guard evaluations are assumed to consume no time in this calculus, 
any upper bound on the run--time of both branches of a conditional is also an upper bound on the run--time of the conditional itself (cf. rule $[\textrm{if}]$).
The rule of consequence allows to increase an already proven upper bound on the run--time by an arbitrary constant factor.
Furthermore, the run--time of two sequentially composed programs $\stmt_1$ and $\stmt_2$ is, intuitively, the sum of their run--times $E_1$ and $E_2$.
However, run--times are expressions which are evaluated in the initial state.
Thus, the run--time of $\stmt_2$ has to be expressed in the initial state of $\stmt_1;\stmt_2$. 
Technically, this is achieved by adding a fresh (and hence universally quantified) variable $u$ that is an upper bound on $E_2$ 
and at the same time is equal to a new expression $E_2'$ in the precondition of $\stmt_1;\stmt_2$. 
Then, the run--time of $\stmt_1;\stmt_2$ is given by the sum $E_1 + E_2'$.

The same principle is applied to each loop iteration. 
Here, the run--time of the loop body is given by $E_1$ and the run--time of the remaining $z$
loop iterations, $E'$, is expressed in the initial state by adding a fresh variable $u$.
Then, any upper bound of $E \geq E_1 + E'$ is an upper bound on the run--time of $z$ loop iterations.

We denote provability of a correctness property $\htriple{P}{\stmt}{E}{Q}$ and a total correctness property $\hrtriple{P}{\stmt}{Q}$ in the 
standard Hoare calculus by $\nprove \htriple{P}{\stmt}{E}{Q}$ and $\hprove \hrtriple{P}{\stmt}{Q}$, respectively.
 \begin{figure}[tb]
  \vspace{-1.5ex}
  \begin{align*}
   &
    \infrule
      {}
      {\htriple{P}{\SKIP}{{\color{gray}1}}{P}}
    ~\lrule{\SKIP}
    \quad 
    \infrule
      {}
      {\htriple{Q\subst{x}{\eval{E}}}{\ASSIGN{x}{E}}{{\color{gray}1}}{Q}}
    ~\lrule{\textrm{\textrm{Assgn}}}
    \\[1.5ex]
    &
    \infrule
      {
        \htriple{P \wedge {\color{gray} E_2' = u}}{\stmt_1}{{\color{gray}E_1}}{Q \wedge {\color{gray} E_2 \leq u}}
        \quad 
        \htriple{Q}{\stmt_2}{{\color{gray}E_2}}{R}
      }
      {\htriple{P}{\stmt_1;\stmt_2}{{\color{gray}E_1+E_2'}}{R}}
    ~\lrule{\textrm{Seq}}
    \\
    & \text{\color{gray} where $u$ is a fresh logical variable}
    \\[1.5ex]
    &
    \infrule
      {
        \htriple{P \wedge \pguard}{\stmt_1}{{\color{gray}E}}{Q}
        \quad 
        \htriple{P \wedge \neg \pguard}{\stmt_2}{{\color{gray}E}}{Q}
      }
      {\htriple{P}{\ITE{\pguard}{\stmt_1}{\stmt_2}}{{\color{gray}E}}{Q}}
    ~\lrule{\textrm{if}}
    \\[1.5ex]
    &
    \infrule
    {\htriple{P(z+1) \wedge {\color{gray} E'=u}}{\stmt}{{\color{gray}E_1}}{P(z) {\color{gray} ~\wedge~ E \leq u}}}
    {\htriple{\exists z \mydot P(z)}{\WHILEDO{\pguard}{\stmt}}{{\color{gray}E}}{P(0)}}
    ~\lrule{\textrm{\textrm{\WHILE}}}
    \\
    & 
    \text{where } z \in \Nats,~ P(z+1) \Rightarrow \pguard {~\wedge~ \color{gray} E \geq E_1 + E'},~
    P(0) \Rightarrow \neg \pguard {~\wedge~ \color{gray} E \geq 1}
    \\[-0.5ex]
    & \text{\color{gray} and $u$ is a fresh logical variable}
    \\[1.5ex]
    &
    \infrule
    {\htriple{P'}{\stmt}{{\color{gray}E'}}{Q'}}
    {\htriple{P}{\stmt}{{\color{gray}E}}{Q}}
    ~\lrule{\textrm{cons}}
    \\
    & \text{where } P \Rightarrow P' {\color{gray} ~\wedge~ E' \leq k \cdot E \text{ for some $k \in \Nats$}} \text{ and } Q' \Rightarrow Q
  \end{align*}
\hrule
  \caption{Inference system for order of magnitude of run--time of deterministic programs according to Nielson~\cite{Nielson:SCP:87}}
  \label{fig:nielson-inference-rules}
 \end{figure}
\begin{theorem}[Soundness of \boldeetsymbol{} for deterministic programs]\label{thm:eet-axiomatic-soundness}
 For all $\stmt \in \DetStmt$ and assertions $P,Q$, we have
 \[
  \hprove \hrtriple{P}{\stmt}{Q} \text{ implies } \nprove \htriple{P}{\stmt}{\eet{\stmt}{\ctert{0}}}{Q}.
 \]
\end{theorem}
\begin{proof}
 By induction on the structure of the program.
 See Appendix~\ref{sec:proof-thm-eet-axiomatic-soundness} for a formal proof.
\end{proof}
Intuitively, this theorem means that for every terminating deterministic program, the \eetsymbol{} is an upper bound on the 
run--time, i.e. \eetsymbol{} is sound with respect to the inference system shown in \autoref{fig:nielson-inference-rules}.
The next theorem states that no tighter bound can be derived in this calculus.
We cannot get a more precise relationship, since we assume guard evaluations to consume time.
\begin{theorem}[Completeness of \boldeetsymbol{} w.r.t.~Nielson] \label{thm:det-eet-completeness}
  For all $\stmt \in \DetStmt$, assertions $P,Q$ and deterministic expressions $E$,
  $\nprove \htriple{P}{\stmt}{E}{Q}$ implies that there exists a natural number $k$ such that for all $\ps \in \ProgramStates$ satisfying $P$, we have 
  \[\eet{\stmt}{\ctert{0}}(\ps) ~\leq~ k \cdot (\eval{E}(\ps))~. \]
\end{theorem}
\begin{proof}
 By induction on $C$'s structure;
see Appendix~\ref{sec:proof-thm-eet-axiomatic-soundness} for a detailed proof.
 \qed
\end{proof}
\autoref{thm:eet-axiomatic-soundness} together with \autoref{thm:det-eet-completeness} shows that our notion of $\eetsymbol$
is a conservative extension of Nielson's approach for reasoning about the run--time of deterministic programs.
In particular, given a correctness proof of a deterministic program $\stmt$ in Hoare logic, it suffices to compute $\eet{C}{\zero}$ in order to obtain 
a corresponding proof in Nielson's proof system.

\section{Case Studies}
\label{sec:applications}
In this section we use our $\eetsymbol$--calculus to analyze the run--time of
three well--known randomized algorithms: a fair \emph{One--Dimensional (Symmetric)
  Random Walk}, the \emph{Coupon Collector's Problem}, and the \emph{randomized Quicksort} algorithm.

\subsection{One--Dimensional Random Walk}
Consider program 
\begin{align*}
C_\mathit{rw} \boldsymbol{\colon}\;\; &\ASSIGN{x}{10};\\
&\WHILE \,  (x > 0) \, \{\\
&\qquad\APPASSIGN{x}{\nicefrac 1 2 \cdot \langle x{-}1
          \rangle + \nicefrac 1 2 \cdot \langle x{+}1 \rangle}\\
& \}~,
\end{align*}
which models a one--dimensional walk of a particle which starts at position ${x =
10}$ and moves with equal probability to the left or to the right in each turn.
The random walk terminates when the particle reaches position $x = 0$. It can be shown
that the program terminates with probability one~\cite{Hurd:TPHOL:02}
but requires, on average, an infinite time to do so. We now apply our
$\eetsymbol$--calculus to formally derive this run--time assertion.\footnote{The invariant proposed in this technical report differs from the originally published invariant which unfortunately contained a mistake.}

The expected run--time of $C_\mathit{rw}$ is given by 
\[
\eet{C_\mathit{rw}}{\zero} ~=~\eet{\ASSIGN{x}{10}}{\eet{\WHILEDO{x>0}{C}}{\zero}}~,
\]
where $C$ stands for the probabilistic assignment in the loop body. Thus, we
need to first determine run--time $\eet{\WHILEDO{x > 0}{C}}{\zero}$. To do so,
we claim that 
\[
I_n ~=~ \one + \sum_{k=0}^{n} \,\eval{x \,{>}\, k} \cdot a_{n,k} \quad
\text{with} \quad  a_{n,k} = \frac{1}{2^{n}} \left[- \binom{n}{\lfloor\!\frac{n-k}{2}\!\rfloor} + 2 \,
\sum_{i=0}^{n-k}
2^i \binom{n-i}{\lfloor\!\frac{n-i-k}{2}\!\rfloor} \right]
\]
is a lower $\omega$--invariant of loop $\WHILEDO{x > 0}{C}$ with respect to continuation $\ctert{0}$. 
To verify this claim, we proceed in two steps: First we show that for 
positive constants $a_{n,k}$,
$I_n = \one + \sum_{k=0}^{n} \,\eval{x \,{>}\, k} \cdot a_{n,k}$ is a lower
$\omega$--invariant of the loop with respect to continuation $\ctert{0}$ whenever the $a_{n,k}$ satisfy the recurrence relation
$$
\begin{array}{r@{\:\:}c@{\:\:}l}
 a_{0,0} &=& 1~,\\[0.5ex]
 a_{n+1,0} &=& 2 + \tfrac{1}{2} \cdot (a_{n,0} + a_{n,1})~,\\[1ex]
a_{n+1,k} &=& \tfrac{1}{2} \cdot (a_{n,k-1} + a_{n,k+1}) \text{ for all } 1 \,{\leq}\,
  k \,{\leq}\, n{+}1~,\\[1ex]
a_{n,k} &=& 0 \text{ for all } k > n~,
\end{array}
$$
for $n \geq 0$.
Then we show that our proposed $a_{n,k}$ satisfies this recurrence relation.\footnote{We
assume that the binomial coefficient $\binom{n}{m}$ is $0$ whenever $m < 0$.} Required calculations for both steps can be found in Appendix
\ref{sec:app-random-walk}. Now \autoref{thm:omega-inv} and the
fact that $\lim_{n \rightarrow \infty} a_{n,0} = \infty$ (see Appendix
\ref{sec:app-random-walk}) give
\[
\eet{\WHILEDO{x \,{>}\, 0}{C}}{\zero} \succeq \lim_{n \To \infty} I_n \succeq \lim_{n
  \To \infty} \one + \eval{x \,{>}\, 0} \cdot a_{n,0} = \one + \eval{x \,{>}\, 0} \cdot \boldsymbol{\infinity}~.
\]
Altogether we have
\begin{align*}
\eet{C_\mathit{rw}}{\zero} 	
&~=~\eet{\ASSIGN{x}{10}}{\eet{\WHILEDO{x>0}{C}}{\zero}}\\
&~\succeq ~\eet{\ASSIGN{x}{10}}{\one + \eval{x \,{>}\, 0} \cdot \boldsymbol{\infinity}}\\
&~=~\ctert{1} + (\ctert{1} + \eval{x \,{>}\, 0} \cdot \boldsymbol{\infinity})\subst{x}{10}\\
&~=~\ctert{1} + (\ctert{1} + 1 \cdot \boldsymbol{\infinity}) ~=~ \boldsymbol{\infinity}~,
\end{align*}
which says that $\eet{C_\mathit{rw}}{\zero} \succeq
\boldsymbol{\infinity}$. Since the reverse inequality holds trivially, we
conclude that $\eet{C_\mathit{rw}}{\zero} = \boldsymbol{\infinity}$.

\subsection{The Coupon Collector's Problem} \label{subsec:coupon}

Now we apply our \eetsymbol--calculus to solve the Coupon Collector's
  Problem. This problem arises from the following scenario\footnote{The
  problem formulation presented here is taken
  from~\cite{Mitzenmacher:2005}.}: Suppose each box of cereal contains one of
$N$ different coupons and once a consumer has collected a coupon of each type,
he can trade them for a prize. The aim of the problem is determining the average
number of cereal boxes the consumer should buy to collect all coupon types,
assuming that each coupon type occurs with the same probability in the cereal
boxes.

The problem can be modeled by program $C_\mathit{cp}$ below:
%
\begin{align*}	
&			\COMPOSE{\ASSIGN{\mathit{cp}}{[0,\, \ldots,\, 0]}}{\COMPOSE{\ASSIGN{i}{1}}{\ASSIGN{x}{N}}}\\
	&			\WHILE \:(x > 0)\: \{\\
	&\qquad				\WHILE \:(\cp{i} \neq 0)\: \{ \\
	&\qquad \qquad				\APPASSIGN{i}{\mathtt{Unif}[1 \ldots N]} \\
	&\qquad				\}; \\
	&\qquad			\COMPOSE{\ASSIGN{\cp{i}}{1}}{\ASSIGN{x}{x-1}}\\
	&\}
\end{align*}
Array $\mathit{cp}$ is initialized to $0$ and whenever we obtain the first
coupon of type $i$, we set $\mathit{cp}[i]$ to $1$. The outer loop is iterated
$N$ times and in each iteration we collect a new---unseen---coupon type. The
collection of the new coupon type is performed by the inner loop.

We start the run--time analysis of $C_\mathit{cp}$ introducing some
notation. Let $\stmt_{\textnormal{in}}$ and $\stmt_{\textnormal{out}}$,
respectively, denote the inner and the outer loop of $C_\mathit{cp}$.
Furthermore, let $\nocol \triangleq \sum_{i=1}^{N} \cond{\cp{i} \neq 0}$ denote
the number of coupons that have already been collected.

\paragraph{Analysis of the inner loop.}
For analyzing the run--time of the outer loop we need to refer to the run--time of
its body, with respect to an arbitrary continuation $g \in \E$. Therefore, we
first analyze the run--time of the inner loop $C_\mathit{in}$.  We propose the
following lower and upper $\omega$--invariant for the inner loop
$C_\mathit{in}$:
%
%

\begin{align*}
{J}_\textnormal{n}^g 	~=~ 	\ctert{1} ~+~ & \cond{\cp{i} = 0} \cdot g \\
        ~+~ & \cond{\cp{i} \neq 0} \cdot \sum_{k =
          0}^{n}\left(\frac{\nocol}{N}\right)^k \biggl(\ctert{2} + \frac{1}{N}
          \sum_{j=1}^{N} \: \eval{\cp{j}=0} \cdot g[i/j]\biggr)~.
\end{align*}
%
Moreover, we write $J^g$ for the same invariant where $n$ is replaced by
$\infty$.  A detailed verification that $J_\textnormal{n}^g$ is indeed a lower
and upper $\omega$--invariant is provided in Appendix
\ref{sec:app-coupon-collector-inner}. \autoref{thm:omega-inv} now yields
\begin{align*}
 J^g = \lim_{n \to \infty} J_\textnormal{n}^g ~\preceq~ \eet{\stmt_{\textnormal{in}}}{g} ~\preceq~ \lim_{n \to \infty} J_\textnormal{n}^g = J^g. \tag{$\star$}
\end{align*}
Since the run--time of a deterministic assignment $\ASSIGN{x}{E}$ is
\begin{equation}
\eet{\ASSIGN{x}{E}}{f} ~=~ \ctert{1} + f\subst{x}{E} \tag{$\maltese$}~,
\end{equation}
the expected run--time of the body of the outer loop reduces to
\begin{align*}
 \MoveEqLeft[2] \eet{\COMPOSE{\stmt_{\textnormal{in}}}{\COMPOSE{\ASSIGN{\cp{i}}{1}}{\ASSIGN{x}{x-1}}}}{g} \tag{$\dag$} \\
 &=~ \ctert{2} + \eet{\stmt_{\textnormal{in}}}{g[x/x-1,\, \cp{i}/1]} \tag{by $\maltese$} \\
 &=~ \ctert{2} + J^{g[x/x-1,\, \cp{i}/1]} \tag{by $\star$} \\
 &=~ \ctert{2} + J^{g}[x/x-1,\, \cp{i}/1]~.
\end{align*}

\paragraph{Analysis of the outer loop.}

Since program $C_\mathit{cp}$ terminates right after the execution of the outer
loop $C_\mathit{out}$, we analyze the run--time of the outer loop
$C_\mathit{out}$ with respect to continuation $\zero$, \ie
$\eet{C_\mathit{out}}{\ctert{0}}$. To this end we propose
\begin{align*}
 I_{n} 
 & ~=~ \ctert{1} + \sum_{\ell = 0}^{n} \cond{x > \ell} \cdot \biggl( \ctert{3} + \cond{n \neq 0} + 2 \cdot \sum_{k=0}^{\infty} \left( \frac{\nocol+\ell}{N} \right)^{\! \!k} \,\biggr) \\
 & \qquad - 2 \cdot \cond{\cp{i}=0} \cdot \cond{x > 0} \cdot \sum_{k=0}^{\infty} \left( \frac{\nocol}{N} \right)^{k}                      
\end{align*}
as both an upper and lower $\omega$--invariant of $C_\mathit{out}$ with respect
to $\ctert{0}$. A detailed verification that $I_n$ is an $\omega$-invariant is found in
Appendix~\ref{sec:app-coupon-collector-outer}. Now \autoref{thm:omega-inv} yields
\begin{align*}
 I ~=~ \lim_{n \to \infty} I_{n} ~\preceq~ \eet{\stmt_{\textnormal{out}}}{\zero} ~\preceq~ \lim_{n \to \infty} I_{n} = I \tag{$\ddag$}~,
\end{align*}
where $I$ denotes the same invariant as $I_n$ with $n$ replaced by $\infty$.

\paragraph{Analysis of the overall program.}
To obtain the overall expected run--time of program $C_\mathit{cp}$ we have to
account for the initialization instructions before the outer loop. The
calculations go as follows:
\begin{samepage}
\begin{align*} 
\MoveEqLeft[2] \eet{C_{\mathit{cp}}}{\zero} \\
 &=~ \eet{\COMPOSE{\COMPOSE{\COMPOSE{\ASSIGN{\cpsymbol}{[0,\ldots,0]}}{\ASSIGN{i}{1}}}{\ASSIGN{x}{N}}}{\stmt_{\textnormal{out}}}}{\zero}&\\[0.3ex]
 &=~ \ctert{3}
 +\eet{\stmt_{\textnormal{out}}}{\zero}[x/N,i/1,\cp{1}/0,\ldots,\cp{N}/0] &
 \text{(by $\maltese$)} \\[0.5ex]
 &=~ \ctert{3} + I[x/N,i/1,\cp{1}/0,\ldots,\cp{N}/0] & \text{(by $\ddag$)} \\
 &=~ \textstyle{\ctert{4} + \cond{N>0} \cdot \left( 4N + 2 \: \sum_{
       \ell=1}^{N-1} \left( \sum_{k=0}^{\infty} \left( \frac{\ell}{N} \right)^{k} \right) \right)} &\\
 &=~ \textstyle{\ctert{4} + \cond{N>0} \cdot \left( 4N + 2 \: \sum_{\ell = 1}^{N-1} 
     \frac{N}{\ell} \right)} &
 \text{$\Bigl($\begin{tabular}{l}geom.\ series and\\[-0.5ex] sum reordering\end{tabular}$\Bigr)$} \\
 &=~ \ctert{4} + \cond{N>0} \cdot 2N \cdot (\ctert{2} + \mathcal{H}_{N-1})~, 
\end{align*}
\end{samepage}
where $\mathcal{H}_{N-1} \triangleq 0 + \nicefrac 1 1 + \nicefrac 1 2 +
\nicefrac 1 3 + \cdots + \nicefrac 1 {N-1}$ denotes the $(N{-}1)$-th harmonic
number. Since the harmonic numbers approach asymptotically to the natural
logarithm, we conclude that the coupon collector algorithm
$C_{\mathit{cp}}$ runs in expected time $\Theta(N \cdot \log(N))$.


\section{Conclusion}
\label{sec:conclusion}
We  have studied a \textsf{wp}--style calculus for reasoning about the expected run--time and positive almost--sure termination of probabilistic programs. 
Our main contribution consists of several sound and complete proof rules for obtaining upper as well as lower bounds on the expected run--time of loops.
We applied these rules to analyze the expected run--time of a variety of example programs including the well--known coupon collector problem.
While finding invariants is, in general, a challenging task, we were able to guess correct invariants by considering a few loop unrollings most of the time.
Hence, we believe that our proof rules are natural and widely applicable.

Moreover, we proved that our approach is a conservative extension of Nielson's approach for reasoning about the run--time of deterministic programs
and that our calculus is sound with respect to a simple operational model.

\paragraph{Acknowledgement.}
We thank Gilles Barthe for bringing to our attention the coupon collector problem as a particularly intricate case study for formal verification of expected run--times Thomas Noll for bringing to our attention Nielson's Hoare logic, and Johannes H\"{o}lzl for pointing out an error in the random walk case study in an earlier version of this paper.

\bibliographystyle{splncs03}
\bibliography{literature}

\begin{thebibliography}{10}
\providecommand{\url}[1]{\texttt{#1}}
\providecommand{\urlprefix}{URL }

\bibitem{DBLP:journals/tocl/ArthanMMO09}
Arthan, R., Martin, U., Mathiesen, E.A., Oliva, P.: A general framework for
  sound and complete {F}loyd-{H}oare logics. {ACM} Trans. Comput. Log.  11(1)
  (2009)

\bibitem{katoenbaier}
Baier, C., Katoen, J.: Principles of Model Checking. {MIT} Press (2008)

\bibitem{DBLP:conf/lopstr/BerghammerM03}
Berghammer, R., M{\"{u}}ller{-}Olm, M.: Formal development and verification of
  approximation algorithms using auxiliary variables. In: Logic Based Program
  Synthesis and Transformation (LOPSTR). LNCS, vol. 3018, pp. 59--74. Springer
  (2004)

\bibitem{DBLP:journals/jcss/BrazdilKKV15}
Br{\'{a}}zdil, T., Kiefer, S., Kucera, A., Varekov{\'{a}}, I.H.: Runtime
  analysis of probabilistic programs with unbounded recursion. J. Comput. Syst.
  Sci.  81(1),  288--310 (2015)

\bibitem{McIver:FM:2005}
Celiku, O., McIver, A.: Compositional specification and analysis of cost-based
  properties in probabilistic programs. In: Formal Methods (FM). LNCS, vol.
  3582, pp. 107--122. Springer (2005)

\bibitem{Chakarov:CAV:13}
Chakarov, A., Sankaranarayanan, S.: Probabilistic program analysis with
  martingales. In: Computer Aided Verification (CAV). LNCS, vol. 8044, pp.
  511--526. Springer Berlin Heidelberg (2013)

\bibitem{Dijkstra}
Dijkstra, E.W.: {A Discipline of Programming}. {Prentice Hall} (1976)

\bibitem{luis}
Fioriti, L.M.F., Hermanns, H.: Probabilistic termination: Soundness,
  completeness, and compositionality. In: Principles of Programming Languages
  (POPL). pp. 489--501. ACM (2015)

\bibitem{Frandsen:1998}
Frandsen, G.S.: Randomised algorithms (1998), {L}ecture Notes, University of
  Aarhus, Denmark

\bibitem{DBLP:conf/icse/GordonHNR14}
Gordon, A.D., Henzinger, T.A., Nori, A.V., Rajamani, S.K.: Probabilistic
  programming. In: Future of Software Engineering (FOSE). pp. 167--181. {ACM}
  (2014)

\bibitem{Hehner:FAC:1998}
Hehner, E.C.R.: Formalization of time and space. Formal Aspects of Computing
  10(3),  290--306 (1998)

\bibitem{Hehner:FAC:2011}
Hehner, E.C.R.: A probability perspective. Formal Aspects of Computing  23(4),
  391--419 (2011)

\bibitem{Hickey:1988}
Hickey, T., Cohen, J.: Automating program analysis. J. ACM  35(1),  185--220
  (1988)

\bibitem{DBLP:journals/cacm/Hoare69}
Hoare, C.A.R.: An axiomatic basis for computer programming. Commun. {ACM}
  12(10),  576--580 (1969)

\bibitem{Hurd:TPHOL:02}
Hurd, J.: A formal approach to probabilistic termination. In: Theorem Proving
  in Higher Order Logics (TPHOL), LNCS, vol. 2410, pp. 230--245. Springer
  Berlin Heidelberg (2002)

\bibitem{DBLP:conf/mfcs/KaminskiK15}
Kaminski, B.L., Katoen, J.: On the hardness of almost-sure termination. In:
  Mathematical Foundations of Computer Science (MFCS), Part {I}. LNCS, vol.
  9234, pp. 307--318. Springer (2015)

\bibitem{DBLP:journals/jcss/Kozen81}
Kozen, D.: {S}emantics of {P}robabilistic {P}rograms. J. Comput. Syst. Sci.
  22(3),  328--350 (1981)

\bibitem{mciver}
McIver, A., Morgan, C.: {A}bstraction, {R}efinement and {P}roof for
  {P}robabilistic {S}ystems. Springer (2004)

\bibitem{Mitzenmacher:2005}
Mitzenmacher, M., Upfal, E.: Probability and Computing: Randomized Algorithms
  and Probabilistic Analysis. Cambridge University Press (2005)

\bibitem{DBLP:conf/sas/Monniaux01}
Monniaux, D.: An abstract analysis of the probabilistic termination of
  programs. In: Symposium on Static Analysis (SAS). Lecture Notes in Computer
  Science, vol. 2126, pp. 111--126. Springer (2001)

\bibitem{DBLP:books/cu/MotwaniR95}
Motwani, R., Raghavan, P.: Randomized Algorithms. Cambridge University Press
  (1995)

\bibitem{Nielson:SCP:87}
Nielson, H.R.: A {H}oare-like proof system for analysing the computation time
  of programs. Sci. Comput. Program.  9(2),  107--136 (1987)

\bibitem{Nielson:UTCS:07}
Nielson, H.R., Nielson, F.: Semantics with Applications: An Appetizer.
  Undergraduate Topics in Computer Science, Springer (2007)

\bibitem{schechter:1996}
Schechter, E.: Handbook of Analysis and Its Foundations. Elsevier Science
  (1996)

\bibitem{Wechler:MTCS:92}
Wechler, W.: Universal Algebra for Computer Scientists, {EATCS} Monographs on
  Theoretical Computer Science, vol.~25. Springer (1992)

\bibitem{Winskel:1993}
Winskel, G.: The Formal Semantics of Programming Languages: An Introduction.
  MIT Press (1993)

\end{thebibliography}

\clearpage

\appendix
\section{Omitted Proofs}
\subsection{Propagation of Constants for \textbf{\textsf{ert}}}
\label{sec:app-eet-const-prop}

For a program $C$, we prove 
\[ \eet{C}{\ctert{k} + f} ~=~ \ctert{k} + \eet{C}{f} \]
by induction on the structure of $C$.
As the induction base we have the atomic programs:

\paragraph{$\EMPTY$:} We have:
\begin{align*}
\eet{\EMPTY}{\ctert{k} + f} ~=~ 	&\ctert{k} + f\tag{\autoref{table:eet-rules}}\\
					~=~	&\ctert{k}  + \eet{\EMPTY}{f}\tag{\autoref{table:eet-rules}}
\end{align*}

\paragraph{$\SKIP$:} We have:
\begin{align*}
	\eet{\SKIP}{\ctert{k} + f}  	~=~&\ctert{1} + \ctert{k} + f \tag{\autoref{table:eet-rules}}\\
						~=~&\ctert{k} + \eet{\SKIP}{f}\tag{\autoref{table:eet-rules}}
\end{align*}

\paragraph{$\APPASSIGN{x}{\mu}$:} 
The proof relies on the fact that for any distribution $\nu$, $\Exp{\nu}{\ctert{k} + f} = \ctert{k}
+ \Exp{\nu}{f}$ and that our distribution expressions denote
distributions of total mass $1$. 
We have:
\begin{align*}
	\eet{\APPASSIGN{x}{\mu}}{\ctert{k} + f} 	~=~ & \ctert{1} +
        \lambda \sigma\mydot \Exp{\pexprdeno{\mu}(\sigma)}{\lambda v.\:
          (\ctert{k} + f)\subst{x}{v}(\sigma)}\tag{\autoref{table:eet-rules}}\\
									~=~ & \ctert{1} + \lambda \sigma\mydot \Exp{\pexprdeno{\mu}(\sigma)}{\lambda v.\: k + f\subst{x}{v}(\sigma)}\tag{\text{$\ctert{k}\subst{x}{v}=\ctert{k}$}} \\
									~=~ & \ctert{1} + \ctert{k} + \lambda \sigma\mydot \Exp{\pexprdeno{\mu}(\sigma)}{\lambda v.\: f\subst{x}{v}(\sigma)} & \\
									~=~ & \ctert{k} + \eet{\APPASSIGN{x}{\mu}}{f} \tag{\autoref{table:eet-rules}}
\end{align*}
As the induction hypothesis we now assume that for arbitrary but fixed $C_1, C_2 \in \appProgs$ it holds that both
\begin{align*}
	\eet{C_1}{\ctert{k} + f} ~=~ \ctert{k} + \eet{C_1}{f} 
\end{align*}
and
\begin{align*}
	\eet{C_2}{\ctert{k} + f} ~=~ \ctert{k} + \eet{C_2}{f}~,
\end{align*}
for any $f\in\E$.

\paragraph{$\COMPOSE{C_1}{C_2}$:} We have:%
\begin{align*}
	\eet{\COMPOSE{C_1}{C_2}}{\ctert{k} + f} 	~=~& \eet{C_1}{\eet{C_2}{\ctert{k} + f}} \tag{\autoref{table:eet-rules}}\\
									~=~& \eet{C_1}{\ctert{k} + \eet{C_2}{f}} \tag{I.H.~on $C_2$}\\
									~=~& \ctert{k} + \eet{C_1}{\eet{C_2}{f}} \tag{I.H.~on $C_1$}\\
									~=~& \ctert{k} + \eet{\COMPOSE{C_1}{C_2}}{f} \tag{\autoref{table:eet-rules}}
\end{align*}

\paragraph{$\NDCHOICE{C_1}{C_2}$:} We have: 
\begin{align*}
		&\hspace{-2em}\eet{\NDCHOICE{C_1}{C_2}}{\ctert{k} + f} \\
	~=~	&\max\, \bigl\{ \eet{C_1}{\ctert{k} + f},\, \eet{C_2}{\ctert{k} + f} \bigr\} \tag{\autoref{table:eet-rules}}\\
	~=~	&\max\, \bigl\{\ctert{k} + \eet{C_1}{f},\, \ctert{k} + \eet{C_2}{f} \bigr\} & \tag{I.H.~on $C_1$ and $C_2$}\\
	~=~	&\ctert{k} + \max\, \bigl\{\eet{C_1}{f},\, \eet{C_2}{f} \bigr\} & \\
	~=~	&\ctert{k} + \eet{\NDCHOICE{C_1}{C_2}}{f} \tag{\autoref{table:eet-rules}}
\end{align*}

\paragraph{$\ITE{\pguard}{C_1}{C_2}$:}  We have:
\begin{align*}
		&\hspace{-2em}\eet{\ITE{\pguard}{C_1}{C_2}}{\ctert{k} + f} \\
	~=~	& \ctert{1} + \eval{\pguard} \cdot \eet{C_1}{\ctert{k} + f} + \eval{\neg \pguard} \cdot \eet{C_2}{\ctert{k} + f} \tag{\autoref{table:eet-rules}}\\
	~=~	& \ctert{1} + \eval{\pguard} \cdot \bigl(\ctert{k} + \eet{C_1}{f}\bigr) + \eval{\neg \pguard} \cdot \bigl(\ctert{k} + \eet{C_2}{f}\bigr) \tag{I.H.~on $C_1$ and $C_2$}\\
	~=~	& \ctert{1} + \ctert{k} + \eval{\pguard} \cdot \eet{C_1}{f} + \eval{\neg \pguard} \cdot \ctert{k} + \eet{C_2}{f}\\
	~=~	& \ctert{k} + \eet{\ITE{\pguard}{C_1}{C_2}}{f} \tag{\autoref{table:eet-rules}}
\end{align*}

\paragraph{$\WHILEDO{\pguard}{C'}$:} Let 
\[
F_f (X) ~=~  \ctert{1} + \eval{\neg \pguard} \cdot f + \eval{\pguard} \cdot \eet{C'}{X}~
\]
be the characteristic functional associated to loop $\WHILEDO{\pguard}{C'}$. The
proof boils down to showing that
\[
\lfp F_{\ctert{k} +f} ~=~ \ctert{k} + \lfp F_{f}~,
\]
which is equivalent to the pair of inequalities $\lfp F_{\ctert{k}
  +f} \leq \ctert{k} + \lfp F_{f}$ and $\lfp F_{f} \leq \lfp F_{\ctert{k} +f} -
\ctert{k}$.  These inequalities follow, in turn, from equalities
\[
F_{\ctert{k} +f} ( \ctert{k} + \lfp F_{f}) = \ctert{k} + \lfp F_{f} \qquad \text{and} \qquad 
F_{f}(\lfp F_{\ctert{k} +f} - \ctert{k}) = \lfp F_{\ctert{k} +f} - \ctert{k}~.
\]
(This is because $\lfp$ gives the \emph{least} fixed point of a transformer and
then $F(x) = x$ implies $\lfp F \leq x$.) Let us now discharge each of the above
equalities:

\begin{align*}
\MoveEqLeft[1] F_{\ctert{k} +f} ( \ctert{k} + \lfp F_{f}) \\
	&~=~\ctert{1} + \eval{\neg \pguard} \cdot (\ctert{k} + f) + \eval{\pguard} \cdot \eet{C'}{\ctert{k} + \lfp F_{f}} \tag{Definition of $F_{\ctert{k} +f}$}\\
	&~=~\ctert{1} + \eval{\neg \pguard} \cdot (\ctert{k} + f) + \eval{\pguard} \cdot \bigl( \ctert{k} +\eet{C'}{\lfp F_{f}} \bigr) \tag{I.H.~on $C'$}\\
	&~=~\ctert{1} + \ctert{k} + \eval{\neg \pguard} \cdot f + \eval{\pguard} \cdot \eet{C'}{\lfp F_{f}} \\
	&~=~\ctert{k} + F_{f} (\lfp F_{f}) \tag{Definition of $F_{f}$}\\
	&~=~\ctert{k} + \lfp F_{f} \tag{Definition of $\lfp$}\\
\\[-1ex]
	\MoveEqLeft[1] F_{f} (\lfp F_{\ctert{k}  + f} - \ctert{k}) \\
	&~=~\ctert{1} + \eval{\neg \pguard} \cdot f + \eval{\pguard} \cdot \eet{C'}{\lfp F_{\ctert{k}  + f} - \ctert{k}} \tag{Definition of $F_{f}$}\\
	&~=~\ctert{1} + \eval{\neg \pguard} \cdot f + \eval{\pguard} \cdot \bigl(\eet{C'}{\lfp F_{\ctert{k}  + f}} - \ctert{k} \bigr) \tag{I.H.~on $C'$}\\
	&~=~\ctert{1} + \eval{\neg \pguard} \cdot (\ctert{k} + f) + \eval{\pguard} \cdot \eet{C'}{\lfp F_{\ctert{k}  + f}} - \ctert{k}  \\
	&~=~F_{\ctert{k}  + f} (\lfp F_{\ctert{k}  + f}) - \ctert{k} \tag{Definition of $F_{\ctert{k} +f}$}\\
	&~=~\lfp F_{\ctert{k}  + f} - \ctert{k} \tag{Definition of $\lfp$}
\end{align*}

\subsection{Preservation of $\boldsymbol{\infty}$ for \textbf{\textsf{ert}}}
\label{sec:app-eet-inf}

For a program $C$, we prove 
\[ \eet{C}{\boldsymbol{\infty}} ~=~ \boldsymbol{\infty} \]
by induction on the structure of $C$.
As the induction base we have the atomic programs:

\paragraph{$\EMPTY$:} We have:
\begin{align*}
\eet{\EMPTY}{\boldsymbol{\infty}} ~=~ 	& \boldsymbol{\infty} \tag{\autoref{table:eet-rules}}\end{align*}

\paragraph{$\SKIP$:} We have:
\begin{align*}
	\eet{\SKIP}{\boldsymbol{\infty}}  	~=~&\ctert{1} +
        \boldsymbol{\infty} \tag{\autoref{table:eet-rules}}\\
						~=~& \boldsymbol{\infty}
\end{align*}

\paragraph{$\APPASSIGN{x}{\mu}$:} The proof relies on the fact that for any
distribution $\nu$ of total mass $1$ and any constant $k \in \Rposinf$,
$\Exp{\nu}{\ctert{k}} = \ctert{k}$. We have:
\begin{align*}
	\eet{\APPASSIGN{x}{\mu}}{\boldsymbol{\infty}} 	~=~ & \ctert{1} + \lambda \sigma\mydot \Exp{\pexprdeno{\mu}(\sigma)}{\lambda v.\: \boldsymbol{\infty}\subst{x}{v}(\sigma)}\tag{\autoref{table:eet-rules}}\\
									~=~ & \ctert{1} + \lambda \sigma\mydot \Exp{\pexprdeno{\mu}(\sigma)}{\lambda v.\: \infty}\tag{\text{$\boldsymbol{\infty}\subst{x}{v}=\boldsymbol{\infty}$}} \\
									~=~ & \ctert{1} + \boldsymbol{\infty} & \\
									~=~ &
                                                                        \boldsymbol{\infty} 
\end{align*}
As the induction hypothesis we now assume that for arbitrary but fixed $C_1, C_2 \in \appProgs$ it holds that both
\begin{align*}
	\eet{C_1}{\boldsymbol{\infty}} ~=~ \boldsymbol{\infty}
\end{align*}
and
\begin{align*}
	\eet{C_2}{\boldsymbol{\infty}} ~=~ \boldsymbol{\infty}~.
\end{align*}

\paragraph{$\COMPOSE{C_1}{C_2}$:} We have:%
\begin{align*}
	\eet{\COMPOSE{C_1}{C_2}}{\boldsymbol{\infty}} 	~=~& \eet{C_1}{\eet{C_2}{\boldsymbol{\infty}}} \tag{\autoref{table:eet-rules}}\\
									~=~& \eet{C_1}{\boldsymbol{\infty}} \tag{I.H.~on $C_2$}\\
									~=~& \boldsymbol{\infty} \tag{I.H.~on $C_1$}
\end{align*}

\paragraph{$\NDCHOICE{C_1}{C_2}$:} We have: 
\begin{align*}
		&\hspace{-2em}\eet{\NDCHOICE{C_1}{C_2}}{\boldsymbol{\infty}} \\
	~=~	&\max\, \bigl\{ \eet{C_1}{\boldsymbol{\infty}},\, \eet{C_2}{\boldsymbol{\infty}} \bigr\} \tag{\autoref{table:eet-rules}}\\
	~=~	&\max\, \bigl\{\boldsymbol{\infty},\, \boldsymbol{\infty} \bigr\} & \tag{I.H.~on $C_1$ and $C_2$}\\
	~=~	&\boldsymbol{\infty}
\end{align*}

\paragraph{$\ITE{\pguard}{C_1}{C_2}$:}  We have:
\begin{align*}
		&\hspace{-2em}\eet{\ITE{\pguard}{C_1}{C_2}}{\boldsymbol{\infty}} \\
	~=~	& \ctert{1} + \eval{\pguard} \cdot \eet{C_1}{\boldsymbol{\infty}} + \eval{\neg \pguard} \cdot \eet{C_2}{\boldsymbol{\infty}} \tag{\autoref{table:eet-rules}}\\
	~=~	& \ctert{1} + \eval{\pguard} \cdot \boldsymbol{\infty} +
        \eval{\neg \pguard} \cdot \boldsymbol{\infty} \tag{I.H.~on $C_1$ and
          $C_2$}\\
	~=~	& \ctert{1} + (\eval{\pguard} + \eval{\neg \pguard}) \cdot \boldsymbol{\infty} \\
	~=~	& \ctert{1} + 1 \cdot \boldsymbol{\infty} \\
	~=~	& \boldsymbol{\infty}
\end{align*}

\paragraph{$\WHILEDO{\pguard}{C'}$:} Let $F$ be the characteristic functional of
the loop with respect to run--time $\boldsymbol{\infty}$. Since \eetsymbol\ is
continuous, $F$ is continuous and the Kleene Fixed Point Theorem states that 
\[
\eet{\WHILEDO{\pguard}{C}}{\boldsymbol{\infty}} ~=~ \supn F^n~,
\]
where $F^0 = \ctert{0}$ and $F^{n+1} = F(F^n)$. Then the result follows from
showing that $F^n \,\succeq\, \ctert{n}$. We do this by induction on
$n$. The base case is immediate. For the inductive case take $\sigma \in \States$; then
\begin{align*}
F^{n+1}(\sigma)
&~=~ 1 + \eval{\neg \pguard}(\sigma) \cdot \infty + \eval{\pguard}(\sigma)
\cdot \eet{C'}{F^n}(\sigma)~.
\end{align*}
Now we distinguish two cases. If $\eval{\neg \pguard}(\sigma) > 0$ we have
\[
F^{n+1}(\sigma) ~\geq~ 1 + \eval{\neg \pguard}(\sigma) \cdot \infty ~=~ \infty ~\geq~ n+1~.
\]
If, on the contrary, $\eval{\neg \pguard}(\sigma) = 0$, we have $\eval{\pguard}(\sigma)
= 1$ and 
\begin{align*}
F^{n+1}(\sigma)
&~=~ 1 + \eet{C'}{F^n}(\sigma)\\
&~\geq~ 1 + \eet{C'}{\ctert{n}}(\sigma) \tag{I.H. and monot.\ of $\eetsymbol[C']$} \\
&~=~ 1 + n + \eet{C'}{\ctert{0}}(\sigma) \tag{Prop.\ of constants of  $\eetsymbol[C']$}\\
&~\geq~ 1 + n
\end{align*}

\subsection{Sub--Additivity of \boldeetsymbol{}}
\label{sec:app-eet-additivity}

We prove 
\begin{align*}
	\eet{C}{f + g} ~\preceq~ \eet{C}{f} + \eet{C}{g}
\end{align*}
for fully probabilistic programs by induction on the program structure.
As the induction base we have the atomic programs:

\paragraph{$\EMPTY$~:} 

We have:
\begin{align*}
					&f + g ~\preceq~ f + g\\
	\Longleftrightarrow~	&\eet{\EMPTY}{f + g} ~\preceq~ \eet{\EMPTY}{f} + \eet{\EMPTY}{g}\tag{\autoref{table:eet-rules}}
\end{align*}

\paragraph{$\SKIP$~:} 

We have:
\begin{align*}
					&\one + f + g ~\preceq~ \one + f + \one + g\\
	\Longleftrightarrow~	&\eet{\SKIP}{f + g} ~\preceq~ \eet{\SKIP}{f} + \eet{\SKIP}{g}\tag{\autoref{table:eet-rules}}
\end{align*}

\paragraph{$\HALT$~:} 

We have:
\begin{align*}
					&\zero ~\preceq~ \zero + \zero\\
	\Longleftrightarrow~	&\eet{\HALT}{f + g} ~\preceq~ \eet{\HALT}{f} + \eet{\HALT}{g}\tag{\autoref{table:eet-rules}}
\end{align*}

%

\paragraph{$\APPASSIGN x \mu$~:}

We have:
\begin{align*}
					&\one + \lambda \sigma.\: \Exp{\llbracket \mu \rrbracket (s)}{\lambda v.\: f\subst{x}{v}} + \lambda \sigma.\: \Exp{\llbracket \mu \rrbracket (s)}{\lambda v.\: g\subst{x}{v}}\\
					&\qquad\preceq~ \one + \lambda \sigma.\: \Exp{\llbracket \mu \rrbracket (s)}{\lambda v.\: f\subst{x}{v}} + \one + \lambda \sigma.\: \Exp{\llbracket \mu \rrbracket (s)}{\lambda v.\: g\subst{x}{v}}\\
	\Longleftrightarrow~	&\one + \lambda \sigma.\: \Exp{\llbracket \mu \rrbracket (s)}{\lambda v.\: f\subst{x}{v}} + \lambda \sigma.\: \Exp{\llbracket \mu \rrbracket (s)}{\lambda v.\: g\subst{x}{v}}\tag{\autoref{table:eet-rules}}\\
	&\qquad\preceq~ \eet{\APPASSIGN x \mu}{f} + \eet{\APPASSIGN x \mu}{g}\\
	\Longleftrightarrow~	&\one + \lambda \sigma.\: \Exp{\llbracket \mu \rrbracket (s)}{\lambda v.\: f\subst{x}{v} + \lambda v.\: g\subst{x}{v}} \\
	&\qquad\preceq~ \eet{\APPASSIGN x \mu}{f} + \eet{\APPASSIGN x \mu}{g}\tag{linearity of \textsf{E}}\\
	\Longleftrightarrow~	&\one + \lambda \sigma.\: \Exp{\llbracket \mu \rrbracket (s)}{\lambda v.\: (f + g)\subst{x}{v}} \\
	&\qquad\preceq~ \eet{\APPASSIGN x \mu}{f} + \eet{\APPASSIGN x \mu}{g}\\
	\Longleftrightarrow~&\eet{\APPASSIGN x \mu}{f + g}\\
	&\qquad\preceq~ \eet{\APPASSIGN x \mu}{f} + \eet{\APPASSIGN x \mu}{g}\tag{\autoref{table:eet-rules}}
\end{align*}
As the induction hypothesis we now assume that for arbitrary but fixed $C_1, C_2 \in \appProgs$ it holds that both
\begin{align*}
	\eet{C_1}{f + g} ~\preceq~ \eet{C_1}{f} + \eet{C_1}{g}
\end{align*}
and
\begin{align*}
	\eet{C_2}{f + g} ~\preceq~ \eet{C_2}{f} + \eet{C_2}{g}~,
\end{align*}
for any $f,g\in\E$.

We now do the induction step by considering the composed programs:

\paragraph{$\COMPOSE{C_1}{C_2}$~:}

By the induction hypothesis on $C_2$ we have:
\begin{align*}
					&\eet{C_2}{f + g} ~\preceq~ \eet{C_2}{f} + \eet{C_2}{g}\\
	\implies\;			&\eet{C_1}{\eet{C_2}{f + g}} ~\preceq~ \eet{C_1}{\eet{C_2}{f} + \eet{C_2}{g}}\tag{\autoref{thm:eet-prop}, \textit{Monotonicity}}\\
	\Longleftrightarrow~	&\eet{\COMPOSE{C_1}{C_2}}{f + g} ~\preceq~ \eet{C_1}{\eet{C_2}{f} + \eet{C_2}{g}}\tag{\autoref{table:eet-rules}}\\
	\implies\;			&\eet{\COMPOSE{C_1}{C_2}}{f + g} ~\preceq~ \eet{C_1}{\eet{C_2}{f} + \eet{C_2}{g}}\tag{I.H.~on $C_1$}\\
					&\qquad ~\preceq~  \eet{C_1}{\eet{C_2}{f}} +  \eet{C_1}{\eet{C_2}{g}}\\
	\implies\;			&\eet{\COMPOSE{C_1}{C_2}}{f + g} ~\preceq~  \eet{C_1}{\eet{C_2}{f}} +  \eet{C_1}{\eet{C_2}{g}}\\
	\implies\;			&\eet{\COMPOSE{C_1}{C_2}}{f + g} ~\preceq~  \eet{\COMPOSE{C_1}{C_2}}{f} +  \eet{\COMPOSE{C_1}{C_2}}{g}\tag{\autoref{table:eet-rules}}
\end{align*}

\paragraph{$\ITE{\pguard}{C_1}{C_2}$~:}

We have:
\begin{align*}
			&\eet{\ITE{\pguard}{C_1}{C_2}}{f + g}\\
	~=~ 		&\one + \eval{\pguard} \cdot \eet{C_1}{f + g} + \eval{\neg \pguard} \cdot \eet{C_2}{f + g}\tag{\autoref{table:eet-rules}}\\
	~\preceq~	&\one + \eval{\pguard} \cdot \left(\eet{C_1}{f} + \eet{C_1}{g}\right)\tag{I.H.~on $C_1$ and $C_2$}\\
			&\qquad {} + \eval{\neg \pguard} \cdot \left(\eet{C_2}{f} + \eet{C_2}{g}\right)\\
	~=~		&\one + \eval{\pguard} \cdot \eet{C_1}{f} + \eval{\neg \pguard} \cdot \eet{C_2}{f}\\
			&\qquad {} + \eval{\pguard} \cdot \eet{C_1}{g} + \eval{\neg \pguard} \cdot \eet{C_2}{g}\\
	~\preceq~	&\one + \eval{\pguard} \cdot \eet{C_1}{f} + \eval{\neg \pguard} \cdot \eet{C_2}{f}\\
			&\qquad {} + \one + \eval{\pguard} \cdot \eet{C_1}{g} + \eval{\neg \pguard} \cdot \eet{C_2}{g}\\
	~=~		&\eet{\ITE{\pguard}{C_1}{C_2}}{f}\tag{\autoref{table:eet-rules}}\\
			&\qquad{} + \eet{\ITE{\pguard}{C_1}{C_2}}{g}
\end{align*}

\paragraph{$\WHILEDO{\pguard}{C'}$:} Let 
\begin{align*}
	F_{h} (X) ~=~  \ctert{1} + \eval{\neg \pguard} \cdot h + \eval{\pguard} \cdot \eet{C'}{X}~.
\end{align*}
be the characteristic functional associated to loop $\WHILEDO{\pguard}{C'}$. The
proof boils down to showing that
\begin{align*}
	\lfp F_{f + g} ~\preceq~ \lfp F_{f} + \lfp F_{g}~.
\end{align*}
This inequality follows from the inequality
\begin{align*}
	F_{f + g} (X) ~\preceq~ F_{f}(X) + F_{g}(X)
\end{align*}
by fixed point induction.
We discharge this proof obligation as follows:
\begin{align*}
	F_{f + g} (X) 	~=~		&\ctert{1} + \eval{\neg \pguard} \cdot (f + g) + \eval{\pguard} \cdot \eet{C'}{X} \tag{Definition of $F_{f+g}$}\\
				~\preceq~	&\ctert{1} + \eval{\neg \pguard} \cdot f + \eval{\neg \pguard} \cdot g + \eval{\pguard} \cdot \eet{C'}{X} + \eval{\pguard} \cdot \eet{C'}{X}\\
				~=~		&\ctert{1} + \eval{\neg \pguard} \cdot f + \eval{\pguard} \cdot \eet{C'}{X} + \eval{\neg \pguard} \cdot g + \eval{\pguard} \cdot \eet{C'}{X}\\
				~\preceq~	&\ctert{1} + \eval{\neg \pguard} \cdot f + \eval{\pguard} \cdot \eet{C'}{X} + \one + \eval{\neg \pguard} \cdot g + \eval{\pguard} \cdot \eet{C'}{X}\\
				~=~		&F_{f} (X) + F_{g} (X) \tag{Definition of $F_{f}$ and $F_{g}$}
\end{align*}

\subsection{Scaling of \textbf{\textsf{ert}}}
\label{sec:app-eet-scaling}

We prove 
\begin{align*}
	\min\{1,\, r\} \cdot \eet{C}{f} ~\preceq~  \eet{C}{r \cdot f} ~\preceq~ \max\{1,\, r\} \cdot \eet{C}{f}
\end{align*}
by induction on the program structure.
As the induction base we have the atomic programs:

\paragraph{$\EMPTY$~:} 

We have:
\begin{align*}
	&\min\{r,\, 1\} \cdot f ~\preceq~ r\cdot f ~\preceq~ \max\{r,\, 1\} \cdot f\\
	\Longleftrightarrow~&\min\{r,\, 1\} \cdot \eet{\EMPTY}{f} ~\preceq~ \eet{\EMPTY}{r \cdot f}\\
	&\qquad\preceq~ \max\{r,\, 1\} \cdot \eet{\EMPTY}{f}\tag{\autoref{table:eet-rules}}
\end{align*}

\paragraph{$\SKIP$~:} 

We have:
\begin{align*}
	&\min\{r,\, 1\} + \min\{r,\, 1\} \cdot f ~\preceq~ 1 + r\cdot f ~\preceq~ \max\{r,\, 1\} + \max\{r,\, 1\} \cdot f\\
	\Longleftrightarrow~&\min\{r,\, 1\} \cdot (1 + f) ~\preceq~ 1 + r\cdot f ~\preceq~ \max\{r,\, 1\} \cdot (1 + f)\\
	\Longleftrightarrow~&\min\{r,\, 1\} \cdot \eet{\SKIP}{f} ~\preceq~ \eet{\SKIP}{r \cdot f}\\
	&\qquad\preceq~ \max\{r,\, 1\} \cdot \eet{\SKIP}{f}\tag{\autoref{table:eet-rules}}
\end{align*}

%

\paragraph{$\APPASSIGN x \mu$~:}

We have:
\begin{align*}
	&\min\{r,\, 1\} + \min\{r,\, 1\} \cdot \lambda \sigma.\: \Exp{\llbracket \mu \rrbracket (s)}{\lambda v.\: f\subst{x}{v}}\\
	&\qquad\preceq~ 1 + r \cdot \lambda \sigma.\: \Exp{\llbracket \mu \rrbracket (s)}{\lambda v.\: f\subst{x}{v}}\\
	&\qquad\preceq~ \max\{r,\, 1\} + \max\{r,\, 1\} \cdot \lambda \sigma.\: \Exp{\llbracket \mu \rrbracket (s)}{\lambda v.\: f\subst{x}{v}}\\
	\Longleftrightarrow~&\min\{r,\, 1\} \cdot \big(1 + \lambda \sigma.\: \Exp{\llbracket \mu \rrbracket (s)}{\lambda v.\: f\subst{x}{v}}\big)\\
	&\qquad\preceq~ 1 + r \cdot \lambda \sigma.\: \Exp{\llbracket \mu \rrbracket (s)}{\lambda v.\: f\subst{x}{v}}\\
	&\qquad\preceq~ \max\{r,\, 1\} \cdot \big(1 + \lambda \sigma.\: \Exp{\llbracket \mu \rrbracket (s)}{\lambda v.\: f\subst{x}{v}}\big)\\
	\Longleftrightarrow~&\min\{r,\, 1\} \cdot \eet{\APPASSIGN x \mu}{f} ~\preceq~ 1 + r \cdot \lambda \sigma.\: \Exp{\llbracket \mu \rrbracket (s)}{\lambda v.\: f\subst{x}{v}}\\
	&\qquad\preceq~ \max\{r,\, 1\} \cdot \eet{\APPASSIGN x \mu}{f}\\
	\Longleftrightarrow~&\min\{r,\, 1\} \cdot \eet{\APPASSIGN x \mu}{f} ~\preceq~ 1 + \lambda \sigma.\: r \cdot \Exp{\llbracket \mu \rrbracket (s)}{\lambda v.\: f\subst{x}{v}}\\
	&\qquad\preceq~ \max\{r,\, 1\} \cdot \eet{\APPASSIGN x \mu}{f}\\
	\Longleftrightarrow~&\min\{r,\, 1\} \cdot \eet{\APPASSIGN x \mu}{f} ~\preceq~ 1 + \lambda \sigma.\: \Exp{\llbracket \mu \rrbracket (s)}{ r \cdot \lambda v.\: f\subst{x}{v}}\\
	&\qquad\preceq~ \max\{r,\, 1\} \cdot \eet{\APPASSIGN x \mu}{f}\tag{linearity of \textsf{E}}\\
	\Longleftrightarrow~&\min\{r,\, 1\} \cdot \eet{\APPASSIGN x \mu}{f} ~\preceq~ 1 + \lambda \sigma.\: \Exp{\llbracket \mu \rrbracket (s)}{\lambda v.\: (r \cdot f)\subst{x}{v}}\\
	&\qquad\preceq~ \max\{r,\, 1\} \cdot \eet{\APPASSIGN x \mu}{f}\\
	\Longleftrightarrow~&\min\{r,\, 1\} \cdot \eet{\APPASSIGN x \mu}{f} ~\preceq~ \eet{\APPASSIGN x \mu}{r \cdot f}\\
	&\qquad\preceq~ \max\{r,\, 1\} \cdot \eet{\APPASSIGN x \mu}{f}\tag{\autoref{table:eet-rules}}
\end{align*}
As the induction hypothesis we now assume that for arbitrary but fixed $C_1, C_2 \in \appProgs$ holds
\begin{align*}
	\min\{1,\, r\} \cdot \eet{C_1}{f} ~\preceq~  \eet{C_1}{r \cdot f} ~\preceq~ \max\{1,\, r\} \cdot \eet{C_1}{f}
\end{align*}
and
\begin{align*}
	\min\{1,\, r\} \cdot \eet{C_2}{f} ~\preceq~  \eet{C_2}{r \cdot f} ~\preceq~ \max\{1,\, r\} \cdot \eet{C_2}{f}~,
\end{align*}
for any $f\in\E$ and any $r \in \Rpos$.

We now do the induction step by considering the composed programs:

\paragraph{$\COMPOSE{C_1}{C_2}$~:}

By the induction hypothesis on $C_2$ it holds that:
\begin{align*}
	&\min\{1,\, r\} \cdot \eet{C_2}{f} ~\preceq~ \eet{C_2}{r\cdot f} ~\preceq~ \max\{1,\, r\} \cdot \eet{C_2}{f}\\
	\implies\;&\eet{C_1}{\min\{1,\, r\} \cdot \eet{C_2}{f}} ~\preceq~ \eet{C_1}{\eet{C_2}{r \cdot f}}\\
	& \qquad\preceq~ \eet{C_1}{\max\{1,\, r\} \cdot \eet{C_2}{f}}\tag{\autoref{thm:eet-prop}, \textit{Monotonicity}}\\
	\Longleftrightarrow~&\eet{C_1}{\min\{1,\, r\} \cdot \eet{C_2}{f}} ~\preceq~ \eet{\COMPOSE{C_1}{C_2}}{r\cdot f}\\
	& \quad\preceq~ \eet{C_1}{\max\{1,\, r\} \cdot \eet{C_2}{f}}\tag{\autoref{table:eet-rules}}\\
	\implies\;&\min\big\{1,\, \min\{1,\, r\}\big\} \cdot \eet{C_1}{\eet{C_2}{f}} ~\preceq~ \eet{\COMPOSE{C_1}{C_2}}{r\cdot f}\\
	& \quad\preceq~ \max\big\{1,\, \max\{1,\, r\}\big\} \cdot \eet{C_1}{\eet{C_2}{f}}\tag{I.H.~on $C_1$}\\
	\Longleftrightarrow~&\min\big\{1,\, \min\{1,\, r\}\big\} \cdot \eet{\COMPOSE{C_1}{C_2}}{f} ~\preceq~ \eet{\COMPOSE{C_1}{C_2}}{r\cdot f}\\
	& \quad\preceq~ \max\big\{1,\, \max\{1,\, r\}\big\} \cdot \eet{\COMPOSE{C_1}{C_2}}{f}\tag{\autoref{table:eet-rules}}\\
	\Longleftrightarrow~&\min\{1,\, r\} \cdot \eet{\COMPOSE{C_1}{C_2}}{f} ~\preceq~ \eet{\COMPOSE{C_1}{C_2}}{r\cdot f}\\
	& \quad\preceq~ \max\{1,\, r\} \cdot \eet{\COMPOSE{C_1}{C_2}}{f}
\end{align*}

\paragraph{$\NDCHOICE{C_1}{C_2}$~:}
We have:
\begin{align*}
	\eet{\NDCHOICE{C_1}{C_2}}{r \cdot f} ~=~ \max\big\{\eet{C_1}{r \cdot f},\, \eet{C_2}{r \cdot f}\big\}\tag{\autoref{table:eet-rules}}\\
\end{align*}
By the induction hypothesis on $C_1$ and $C_2$ we then obtain:
\begin{align*}
					&\max\big\{\min\{1,\, r\} \cdot \eet{C_1}{f},\, \min\{1,\, r\} \cdot \eet{C_2}{f}\big\} \\
					&\qquad ~\preceq~ \max\big\{\eet{C_1}{r \cdot f},\, \eet{C_2}{r \cdot f}\big\}\\
					&\qquad ~\preceq~ \max\big\{\max\{1,\, r\} \cdot \eet{C_1}{f},\, \max\{1,\, r\} \cdot \eet{C_2}{f}\big\} \\
	\Longleftrightarrow~	&\min\{1,\, r\} \cdot \max\big\{\eet{C_1}{f},\, \eet{C_2}{f}\big\} \\
					&\qquad ~\preceq~ \max\big\{\eet{C_1}{r \cdot f},\, \eet{C_2}{r \cdot f}\big\}\\
					&\qquad ~\preceq~ \max\{1,\, r\} \cdot \max\big\{\eet{C_1}{f},\, \eet{C_2}{f}\big\}  \\
	\Longleftrightarrow~	&\min\{1,\, r\} \cdot \eet{\NDCHOICE{C_1}{C_2}}{f} \tag{\autoref{table:eet-rules}}\\
					&\qquad ~\preceq~ \eet{\NDCHOICE{C_1}{C_2}}{r \cdot f}\\
					&\qquad ~\preceq~ \max\{1,\, r\} \cdot \eet{\NDCHOICE{C_1}{C_2}}{f} 
\end{align*}

\paragraph{$\ITE{B}{C_1}{C_2}$~:}

We have
\begin{align*}
	&\eet{\ITE{B}{C_1}{C_2}}{r \cdot f}\\
	~=~ &1 + \ind{B} \cdot \eet{C_1}{r \cdot f} + \ind{\neg B} \cdot \eet{C_2}{r \cdot f}
\end{align*}
and
\begin{align*}
	\eet{\ITE{B}{C_1}{C_2}}{f} ~=~ 1 + \ind{B} \cdot \eet{C_1}{f} + \ind{\neg B} \cdot \eet{C_2}{f}~.
\end{align*}
By the induction hypothesis on $C_1$ we have
\begin{align*}
	&\min\{1,\, r\} \cdot \eet{C_1}{f} ~\preceq~ \eet{C_1}{r\cdot f} ~\preceq~\max\{1,\, r\} \cdot \eet{C_1}{f}~,
\end{align*}
and by the induction hypothesis on $C_2$ we have
\begin{align*}
	&\min\{1,\, r\} \cdot \eet{C_2}{f} ~\preceq~ \eet{C_2}{r\cdot f} ~\preceq~\max\{1,\, r\} \cdot \eet{C_2}{f}~.
\end{align*}
From that we can conclude that both
\begin{align*}
	&\min\{1,\, r\} \cdot \ind{B} \cdot \eet{C_1}{f} ~\preceq~ \ind{B} \cdot \eet{C_1}{r\cdot f}\\
	&\qquad\preceq~\max\{1,\, r\} \cdot \ind{B} \cdot \eet{C_1}{f}
\end{align*}
and
\begin{align*}
	&\min\{1,\, r\} \cdot \ind{\neg B} \cdot \eet{C_2}{f} ~\preceq~ \ind{\neg B} \cdot \eet{C_2}{r\cdot f}\\
	&\qquad\preceq~\max\{1,\, r\} \cdot \ind{\neg B} \cdot \eet{C_2}{f}
\end{align*}
hold.
Adding these inequations up gives
\begin{align*}
	&\min\{1,\, r\}  \cdot \big(\ind{B} \cdot \eet{C_1}{f} + \ind{\neg B} \cdot \eet{C_2}{f}\big)\\
	&\qquad\preceq~ \ind{B} \cdot\eet{C_1}{r\cdot f} + \ind{\neg B} \cdot \eet{C_2}{r\cdot f}\\
	&\qquad\preceq~ \max\{1,\, r\} \cdot \big(\ind{B} \cdot\eet{C_1}{f} + \ind{\neg B} \cdot \eet{C_2}{f} \big)\\
	\implies\;&\min\{1,\, r\} + \min\{1,\, r\}  \cdot \big(\ind{B} \cdot \eet{C_1}{f} + \ind{\neg B} \cdot \eet{C_2}{f}\big)\\
	&\qquad\preceq~ 1 + \ind{B} \cdot\eet{C_1}{r\cdot f} + \ind{\neg B} \cdot \eet{C_2}{r\cdot f}\\
	&\qquad\preceq~ \max\{1,\, r\} + \max\{1,\, r\} \cdot \big(\ind{B} \cdot\eet{C_1}{f} + \ind{\neg B} \cdot \eet{C_2}{f} \big)\\
	\Longleftrightarrow~&\min\{1,\, r\}  \cdot \big(1 + \ind{B} \cdot \eet{C_1}{f} + \ind{\neg B} \cdot \eet{C_2}{f}\big)\\
	&\qquad\preceq~ 1 + \ind{B} \cdot\eet{C_1}{r\cdot f} + \ind{\neg B} \cdot \eet{C_2}{r\cdot f}\\
	&\qquad\preceq~ \max\{1,\, r\} \cdot \big(1 + \ind{B} \cdot\eet{C_1}{f} + \ind{\neg B} \cdot \eet{C_2}{f} \big)\\
	\Longleftrightarrow~& \min\{1,\, r\}  \cdot \eet{\ITE{B}{C_1}{C_2}}{f}\\
	&\qquad\preceq~ \eet{\ITE{B}{C_1}{C_2}}{r \cdot f}\\
	&\qquad\preceq~ \max\{1,\, r\} \cdot \eet{\ITE{B}{C_1}{C_2}}{f}\tag{\autoref{table:eet-rules}}
\end{align*}

\paragraph{$\WHILEDO{B}{C_1}$:}
We have
\begin{align*}
\eet{\WHILEDO{B}{C_1}}{f} ~=~ &\mu X.\: 1 + \ind{B} \cdot \eet{C_1}{X} + \ind{\neg B} \cdot f\\
~=~ &\mu F_f
\end{align*}
for $F_f(X) = 1 + \ind{B} \cdot \eet{C_1}{X} + \ind{\neg B} \cdot f$.
We first prove that
\begin{align*}
	\min\{1,\, r\} \cdot F_f^n(\zero) ~\preceq~ F_{r \cdot f}^n(\zero) ~\preceq~ \max\{1,\, r\} \cdot F_f^n(\zero)
\end{align*}
holds for any $f \in \E$, $r\in\Rpos$, and $n\in\mathbb N$.
We prove this by induction on $n$.

As the induction base, we have $n=0$, so
\begin{align*}
	&\min\{1,\, r\} \cdot F_f^0(\zero) ~\preceq~ F_{r \cdot f}^0(\zero) ~\preceq~ \max\{1,\, r\} \cdot F_f^0(\zero)\\
	\Longleftrightarrow~&\min\{1,\, r\} \cdot \zero ~\preceq~ \zero ~\preceq~ \max\{1,\, r\} \cdot \zero\\
	\Longleftrightarrow~&\zero ~\preceq~ \zero ~\preceq~ \zero~,
\end{align*}
which trivially holds.

As the induction hypothesis we now assume that 
\begin{align*}
	\min\{1,\, r\} \cdot F_f^n(\zero) ~\preceq~ F_{r \cdot f}^n(\zero) ~\preceq~ \max\{1,\, r\} \cdot F_f^n(\zero)
\end{align*}
holds for any $f \in \E$, $r\in\Rpos$, and some arbitrary but fixed $n$.

For the induction step we now show that 
\begin{align*}
	\min\{1,\, r\} \cdot F_f^{n+1}(\zero) ~\preceq~ F_{r \cdot f}^{n+1}(\zero) ~\preceq~ \max\{1,\, r\} \cdot F_f^{n+1}(\zero)
\end{align*}
also holds.
For that, consider that by the induction hypothesis we have:
\begin{align*}
	&\min\{1,\, r\} \cdot F_f^n(\zero) ~\preceq~ F_{r \cdot f}^n(\zero) ~\preceq~ \max\{1,\, r\} \cdot F_f^n(\zero)\\
	\implies\;&\eet{C_1}{\min\{1,\, r\} \cdot F_f^n(\zero)} ~\preceq~ \eet{C_1}{F_{r \cdot f}^n(\zero)}\\
	&\quad\preceq~ \eet{C_1}{\max\{1,\, r\} \cdot F_f^n(\zero)}\tag{\autoref{thm:eet-prop}, \textit{Monotonicity}}\\
	\implies\;&\min\{1,\, r\} \cdot \eet{C_1}{F_f^n(\zero)} ~\preceq~ \eet{C_1}{\min\{1,\, r\} \cdot F_f^n(\zero)}\\
	&\quad\preceq~ \eet{C_1}{F_{r \cdot f}^n(\zero)} ~\preceq~ \eet{C_1}{\max\{1,\, r\} \cdot F_f^n(\zero)}\\
	&\quad\preceq~ \max\{1,\, r\} \cdot \eet{C_1}{F_f^n(\zero)}\tag{I.H.~on $C_1$}\\
	\implies\;&\min\{1,\, r\} \cdot \eet{C_1}{F_f^n(\zero)} ~\preceq~ \eet{C_1}{F_{r \cdot f}^n(\zero)}\\
	&\quad\preceq~ \max\{1,\, r\} \cdot \eet{C_1}{F_f^n(\zero)}\\
	\Longleftrightarrow~&\min\{1,\, r\} \cdot \ind{B} \cdot \eet{C_1}{F_f^n(\zero)} ~\preceq~ \ind{B} \cdot \eet{C_1}{F_{r \cdot f}^n(\zero)}\\
	&\quad\preceq~ \max\{1,\, r\} \cdot \ind{B} \cdot \eet{C_1}{F_f^n(\zero)}\\
	\implies\;&\ind{\neg B} \cdot \min\{1,\, r\} \cdot f + \min\{1,\, r\} \cdot \ind{B} \cdot \eet{C_1}{F_f^n(\zero)}\\
	&\quad\preceq~ \ind{\neg B} \cdot r \cdot f + \ind{B} \cdot \eet{C_1}{F_{r \cdot f}^n(\zero)}\\
	&\quad\preceq~ \ind{\neg B} \cdot \max\{1,\, r\} \cdot f + \max\{1,\, r\} \cdot \ind{B} \cdot \eet{C_1}{F_f^n(\zero)}\\
	\implies\;&\min\{1,\, r\} + \ind{\neg B} \cdot \min\{1,\, r\} \cdot f + \min\{1,\, r\} \cdot \ind{B} \cdot \eet{C_1}{F_f^n(\zero)}\\
	&\quad\preceq~ 1 + \ind{\neg B} \cdot r \cdot f + \ind{B} \cdot \eet{C_1}{F_{r \cdot f}^n(\zero)}\\
	&\quad\preceq~ \max\{1,\, r\} + \ind{\neg B} \cdot \max\{1,\, r\} \cdot f + \max\{1,\, r\} \cdot \ind{B} \cdot \eet{C_1}{F_f^n(\zero)}\\
	\Longleftrightarrow~&\min\{1,\, r\} \cdot \big(1 + \ind{\neg B} \cdot f + \ind{B} \cdot \eet{C_1}{F_f^n(\zero)}\big)\\
	&\quad\preceq~ 1 + \ind{\neg B} \cdot r \cdot f + \ind{B} \cdot \eet{C_1}{F_{r \cdot f}^n(\zero)}\\
	&\quad\preceq~ \max\{1,\, r\} \cdot \big(1 + \ind{\neg B} \cdot f + \ind{B} \cdot \eet{C_1}{F_f^n(\zero)}\big)\\
	\Longleftrightarrow~&\min\{1,\, r\} \cdot F_f^{n+1}(\zero) ~\preceq~ F_{r \cdot f}^{n+1}(\zero) ~\preceq~ \max\{1,\, r\} \cdot F_f^{n+1}(\zero)
\end{align*}
So we have shown by induction that 
\begin{align*}
	\min\{1,\, r\} \cdot F_f^n(\zero) ~\preceq~ F_{r \cdot f}^n(\zero) ~\preceq~ \max\{1,\, r\} \cdot F_f^n(\zero)
\end{align*}
holds for any $f \in \E$, $r\in\Rpos$, and $n\in\mathbb N$.
Then also:
\belowdisplayskip=0pt
\begin{align*}
	&\min\{1,\, r\} \cdot \sup_{n\in\mathbb N}F_f^n(\zero) ~\preceq~ \sup_{n\in\mathbb N} F_{r \cdot f}^n(\zero) ~\preceq~ \max\{1,\, r\} \cdot \sup_{n\in\mathbb N}  F_f^n(\zero)\\
	\Longleftrightarrow~&\min\{1,\, r\} \cdot \mu F_f ~\preceq~ \mu F_{r \cdot f} ~\preceq~ \max\{1,\, r\} \cdot \mu F_f\\
	\Longleftrightarrow~&\min\{1,\, r\} \cdot \eet{\WHILEDO{B}{C_1}}{f} ~\preceq~ \eet{\WHILEDO{B}{C_1}}{r \cdot f}\\
	&\qquad\preceq~ \max\{1,\, r\} \cdot \eet{\WHILEDO{B}{C_1}}{f}\tag{\autoref{table:eet-rules}}
\end{align*}
\qed

\subsection{The $\boldsymbol{\omega}$--Complete Partial Order $\boldsymbol{(\E,\, {\preceq})}$}
\label{sec:proof-thm-cpo}

We prove that $(\E,\, {\preceq})$ is an $\omega$--cpo with bottom element $\zero\colon \sigma \mapsto 0$ and top element $\infinity\colon \sigma \mapsto \infty$.
First we prove that $(\E,\, {\preceq})$ is a partial order and for that we first prove that $\preceq$ is reflexive:
\begin{align*}
						&\forall \sigma \in \States\mydot f(\sigma) = f(\sigma)\\
	~\implies\;				&\forall \sigma \in \States\mydot f(\sigma) \leq f(\sigma)\\
	~\Longleftrightarrow~	&f \preceq f 										\tag{Definition of $\preceq$}
\end{align*}
Next, we prove that $\preceq$ is transitive:
\begin{align*}
						&f \preceq g \textnormal{ and } g \preceq h\\
	~\Longleftrightarrow~	&\forall \sigma \in \States\mydot f(\sigma) \geq g(\sigma) \textnormal{ and } \forall \sigma \in \States\mydot g(\sigma) \geq h(\sigma) \tag{Definition of $\preceq$}\\
	~\Longleftrightarrow~	&\forall \sigma \in \States\mydot f(\sigma) \geq g(\sigma) \geq h(\sigma)\\
	~\implies\;				&\forall \sigma \in \States\mydot f(\sigma) \geq h(\sigma)\\
	~\Longleftrightarrow~	&f \preceq h\tag{Definition of $\preceq$}									
\end{align*}
Finally, we prove that $\preceq$ is antisymmetric:
\begin{align*}
						&f \preceq g \textnormal{ and } g \preceq f\\
	~\Longleftrightarrow~	&\forall \sigma \in \States\mydot f(\sigma) \geq g(\sigma) \textnormal{ and } \forall \sigma \in \States\mydot g(\sigma) \geq f(\sigma) \tag{Definition of $\preceq$}\\
	~\Longleftrightarrow~	&\forall \sigma \in \States\mydot f(\sigma) \geq g(\sigma) \geq f(\sigma)\\
	~\implies\;				&\forall \sigma \in \States\mydot f(\sigma) = g(\sigma)\\
	~\Longleftrightarrow~	&f = g
\end{align*}
At last, we have to prove that every $\omega$--chain $S \subseteq \E$ has a supremum.
Such a supremum can be constructed by taking the pointwise supremum
\begin{align*}
	\sup S ~=~ \lambda \sigma\mydot \sup \big\{f(\sigma) ~\big|~ f\in S\big\}~,
\end{align*}
which always exists as every subset of $\PReals^\infty$ has a supremum.
\qed

\subsection{Continuity of  \textbf{\textsf{ert}}}
\label{sec:app-eet-cont}

Let $\langle f_n \rangle$ be an $\omega$-chain of run--times. We prove
\[ \eet{C}{\supn  f_n} ~=~ \supn \eet{C}{f_n} \]
by induction on the structure of $C$.

\vspace{1.5ex} \noindent \textbf{Case $\boldsymbol{C=\EMPTY}$:}%
\vspace{-1.5ex}
\begin{flalign*}
\MoveEqLeft[2] \eet{\EMPTY}{\supn  f_n} \\
&=~\supn  f_n & (\text{def.~$\eetsymbol$})\\
&=~\supn \eet{\EMPTY}{f_n}& (\text{def.~$\eetsymbol$})
\end{flalign*}

\vspace{1.5ex}\noindent \textbf{Case $\boldsymbol{C=\SKIP}$:}%
\vspace{-1.5ex}
\begin{flalign*}
\MoveEqLeft[2] \eet{\SKIP}{\supn  f_n} \\
&=~\ctert{1} + \supn  f_n & (\text{def.~$\eetsymbol$})\\
&=~\supn \ctert{1} + f_n & (\text{$\ctert{1}$ is const.~for $n$}) \\
&=~\supn \eet{\SKIP}{f_n}& (\text{def.~$\eetsymbol$})
\end{flalign*}

\vspace{1.5ex}\noindent \textbf{Case $\boldsymbol{C=\APPASSIGN{x}{\mu}}$:} The
proof relies on the Lebesgue's Monotone Convergence Theorem (LMCT); see
\eg~\cite[p.~567]{schechter:1996}.
\vspace{-1.5ex}
\begin{flalign*}
\MoveEqLeft[2] \eet{\APPASSIGN{x}{\mu}}{\supn  f_n} \\
&=~ \ctert{1} + \lambda \sigma\mydot
\Exp{\pexprdeno{\mu}(\sigma)}{\lambda v.\: (\supn  f_n)\subst{x}{v}(\sigma)} &  (\text{def.~$\eetsymbol$}) \\
&=~ \ctert{1} + \lambda \sigma\mydot
\Exp{\pexprdeno{\mu}(\sigma)}{\supn \lambda v.\: f_n\subst{x}{v}(\sigma)} & \\
&=~ \ctert{1} + \lambda \sigma\mydot \supn
\Exp{\pexprdeno{\mu}(\sigma)}{\lambda v.\: f_n\subst{x}{v}(\sigma)} & (\text{LMCT}) \\
&=~ \supn \ctert{1} + \lambda \sigma\mydot
\Exp{\pexprdeno{\mu}(\sigma)}{\lambda v.\: f_n\subst{x}{v}(\sigma)} & (\text{$\ctert{1}$ is const.~for $n$}) \\
&=~\supn \eet{\APPASSIGN{x}{\mu}}{f_n} & (\text{def.~$\eetsymbol$})
\end{flalign*}

\vspace{1.5ex} \noindent \textbf{Case $\boldsymbol{C=\COMPOSE{C_1}{C_2}}$:}%
\vspace{-1.5ex} 
\begin{flalign*}
\MoveEqLeft[2] \eet{\COMPOSE{C_1}{C_2}}{\supn  f_n} \\
&=~\eet{C_1}{\eet{C_2}{\supn  f_n}} & (\text{def.~$\eetsymbol$})\\
&=~\eet{C_1}{\supn \eet{C_2}{f_n}} & (\text{IH on $C_2$})\\
&=~\supn \eet{C_1}{\eet{C_2}{f_n}} & (\text{IH on $C_1$})\\
&=~\supn \eet{\COMPOSE{C_1}{C_2}}{f_n} &  (\text{def.~$\eetsymbol$})
\end{flalign*}

\vspace{1.5ex} \noindent \textbf{Case $\boldsymbol{C=\ITE{\pguard}{C_1}{C_2}}$:}
The proof relies on a Monotone Sequence Theorem that says that if $\langle a_n
\rangle$ is a monotonic sequence in $\Rposinf$ then the supremum $\supn a_n$
coincides with $\lim_{n \To \infty} a_n$.
\vspace{-1.5ex}\begin{flalign*}
\MoveEqLeft[2] \eet{\ITE{\pguard}{C_1}{C_2}}{\supn  f_n} \\
&=~\ctert{1} + \eval{\pguard}
\cdot \eet{C_1}{\supn f_n} + \eval{\neg \pguard} \cdot \eet{C_2}{\supn f_n} & (\text{def.~$\eetsymbol$})\\
&=~\ctert{1} + \eval{\pguard}
\cdot \supn \eet{C_1}{f_n} + \eval{\neg \pguard} \cdot \supn \eet{C_2}{f_n} &
(\text{IH on $C_1$, $C_2$})\\
&=~\ctert{1} + \eval{\pguard}
\cdot \lim_{n \To \infty} \eet{C_1}{f_n} + \eval{\neg \pguard} \cdot \lim_{n \To \infty} \eet{C_2}{f_n} &
(\text{MCT})\\
&=~\lim_{n \To \infty} \ctert{1} + \eval{\pguard}
\cdot \eet{C_1}{f_n} + \eval{\neg \pguard} \cdot \eet{C_2}{f_n} &
\\
&=~ \supn \ctert{1} + \eval{\pguard}
\cdot \eet{C_1}{f_n} + \eval{\neg \pguard} \cdot \eet{C_2}{f_n} & (\text{MCT})
\\
&=~\supn \eet{\ITE{\pguard}{C_1}{C_2}}{f_n} &  (\text{def.~$\eetsymbol$})
\end{flalign*}

\vspace{1.5ex} \noindent \textbf{Case $\boldsymbol{C=\NDCHOICE{C_1}{C_2}}$:}
\vspace{-1.5ex}\begin{flalign*}
\MoveEqLeft[2] \eet{\NDCHOICE{C_1}{C_2}}{\supn  f_n} \\
&=~\max\, \bigl\{ \eet{C_1}{\supn  f_n},\, \eet{C_2}{\supn  f_n} \bigr\} & (\text{def.~$\eetsymbol$})\\
&=~\max\, \bigl\{\supn \eet{C_1}{f_n},\, \supn \eet{C_2}{f_n} \bigr\}&
(\text{IH on $C_1$, $C_2$})\\
&\leq~\supn \max\, \bigl\{ \eet{C_1}{f_n},\, \eet{C_2}{f_n} \bigr\} & \\
&=~\supn \eet{\NDCHOICE{C_1}{C_2}}{f_n} &  (\text{def.~$\eetsymbol$})
\end{flalign*}
Let $A = \max\, \bigl\{\supn \eet{C_1}{f_n},\, \supn \eet{C_2}{f_n} \bigr\}$. We
prove the remaining inequality $A \geq \supn \max\, \bigl\{ \eet{C_1}{f_n},\,
\eet{C_2}{f_n} \bigr\} $ by contradiction. Assume that
\begin{align*}
A &< \supn \max\, \bigl\{ \eet{C_1}{f_n},\, \eet{C_2}{f_n} \bigr\}%
\intertext{holds. Then there must exist a natural number $m$ such that}%
A &< \max\, \bigl\{ \eet{C_1}{f_{m}},\, \eet{C_2}{f_{m}} \bigr\}
\end{align*}
since if $\max\, \bigl\{ \eet{C_1}{f_n},\, \eet{C_2}{f_n} \bigr\}$ were bounded
from above by $A$ for every $n$, the supremum $\supn \max\, \bigl\{
\eet{C_1}{f_n},\, \eet{C_2}{f_n} \bigr\}$ would also be bounded from above by
$A$. From this fact we can derive the contradiction $A < A$ as follows:
\begin{flalign*}
A &< \max\, \bigl\{\eet{C_1}{f_m},\,  \eet{C_2}{f_m} \bigr\} &\\
& \leq \max\, \bigl\{\supn \eet{C_1}{f_n},\, \supn \eet{C_2}{f_n} \bigr\} &\\
& = A &
\end{flalign*}

\vspace{1.5ex} \noindent \textbf{Case $\boldsymbol{C=\WHILEDO{\pguard}{C'}}$:} Let 
\[
F_f (X) ~=~  \ctert{1} + \eval{\neg \pguard} \cdot f + \eval{\pguard} \cdot \eet{C'}{X}~
\]
be the characteristic functional associated to loop $\WHILEDO{\pguard}{C'}$. The
proof relies on two properties of $F_f$. Fact 1 says that $F_{\supn f_n} ~=~
\supn F_{f_n}$ and follows from a straightforward reasoning. Fact 2 says that
$\supn F_{f_n}$ is a continuous transformer (in $\E \To \E$) and follows from
the fact that $\langle F_{f_n} \rangle$ forms an $\omega$-chain of continuous
transformers (since by IH $\eet{C'}{\cdot}$ is continuous) and continuous
functions are closed under lubs. Finally, we will make use of a continuity
result of the $\lfp$ operator: Fact 3 says that $\lfp \colon [\E \To \E] \To \E$
is itself continuous when restricted to the set of continuous transformers in
$\E \To \E$, denoted $[\E \To \E]$.
\begin{flalign*}
\MoveEqLeft[2] \eet{\WHILEDO{\pguard}{C'}}{\supn  f_n} \\
&=~\lfp (F_{\supn f_n}) & (\text{def.~$\eetsymbol$})\\
&=~\lfp \Bigl(\,\supn F_{f_n} \Bigr) & (\text{Fact 1})\\
&=~\supn \lfp \bigl( F_{f_n} \bigr) & (\text{Facts 2 and 3})\\
&=~\supn \eet{\WHILEDO{\pguard}{C'}}{f_n}& (\text{def.~$\eetsymbol$})
\end{flalign*}

\subsection{Proof of \autoref{thm:eet-soundness}}
\label{sec:thm:eet-soundness}

Before we prove soundness of \eetsymbol{} with respect to the simple operational model of our probabilistic programming language
introduced in Section~\ref{sec:operational}, some preparation is needed.
In particular, we have to consider bounded $\WHILE$ loops that are obtained by successively unrolling it up to a finite number of executions of the loop body.

\begin{lemma}\label{thm:eet-unroll-loop}
Let $\pguard \in \DExprs$, $\stmt \in \Stmt$, and $\mdpRewSink \in \E$. Then
 \[
  \eet{\WHILEDO{\pguard}{\stmt}}{\mdpRewSink} = \eet{\ITE{\pguard}{\stmt;\WHILEDO{\pguard}{\stmt}}{\EMPTY}}{\mdpRewSink}.
 \]
\end{lemma}
%
%
\begin{proof}
 Let $\eetF(X)$ be the characteristic functional corresponding to $\WHILEDO{\pguard}{\stmt}$ as defined in \autoref{def:F}.
 Then,
 \begin{align*}
        & \eet{\WHILEDO{\pguard}{\stmt}}{\mdpRewSink} \\
    ~=~ & \lfp \eetF \tag{\autoref{table:eet-rules}} \\ 
    ~=~ & \eetF(\lfp \eetF) \tag{Def. \lfp} \\
    ~=~ & 1 + \probof{\pguard}{\true} \cdot \eet{\stmt}{\lfp \eetF} + \probof{\pguard}{\false} \cdot \mdpRewSink \tag{\autoref{def:F}} \\
    ~=~ & 1 + \probof{\pguard}{\true} \cdot \eet{\stmt}{\eet{\WHILEDO{\pguard}{\stmt}}{\mdpRewSink}} \tag{\autoref{table:eet-rules}} \\
        &   + \probof{\pguard}{\false} \cdot \mdpRewSink \\
    ~=~ & 1 + \probof{\pguard}{\true} \cdot \eet{\stmt}{\eet{\WHILEDO{\pguard}{\stmt}}{\mdpRewSink}} \tag{\eet{\EMPTY}{\mdpRewSink} = \mdpRewSink} \\
        &   + \probof{\pguard}{\false} \cdot \eet{\EMPTY}{\mdpRewSink} \\
    ~=~ & 1 + \probof{\pguard}{\true} \cdot \eet{\stmt;\WHILEDO{\pguard}{c}}{\mdpRewSink} \tag{\autoref{table:eet-rules}} \\
        &   + \probof{\pguard}{\false} \cdot \eet{\EMPTY}{\mdpRewSink} \\
    ~=~ & \eet{\ITE{\pguard}{\stmt;\WHILEDO{\pguard}{c}}{\EMPTY}}{\mdpRewSink} \tag{\autoref{table:eet-rules}}.
 \end{align*}
 \qed
\end{proof}
\begin{definition}[Bounded $\WHILE$ loops]
Let $\pguard \in \DExprs$, $\stmt \in \Stmt$, $\mdpRewSink \in \E$, and $k \in \Nats$.
Then
\begin{align*}
	\BOUNDEDWHILE{0}{\pguard}{\stmt} ~\eqdef~ 	&\HALT, \textnormal{ and }\\
	\BOUNDEDWHILE{k+1}{\pguard}{\stmt} ~\eqdef~ 	&\ITE{\pguard}{\stmt;\BOUNDEDWHILE{k}{\pguard}{\stmt}}{\EMPTY}~.
\end{align*}
\end{definition}
%
%

To improve readability, let $\stmt' \eqdef \WHILEDO{\pguard}{\stmt}$ and
  $\stmt_{k} \eqdef \BOUNDEDWHILE{k}{\pguard}{\stmt}$ for the remainder of this section.

The following lemma states that the $\eetsymbol$ of a while loop can be expressed in terms of the $\eetsymbol$ of bounded while loops.
\begin{lemma}\label{thm:eet-bounded-loops}
Let $\pguard \in \DExprs$, $\stmt \in \Stmt$, and $\mdpRewSink \in \E$.
Then
 \[ \sup_{k \in \Nats} \eet{\BOUNDEDWHILE{k}{\pguard}{\stmt}}{\mdpRewSink} = \eet{\WHILEDO{\pguard}{\stmt}}{\mdpRewSink}. \]
\end{lemma}
\begin{proof}
  Let $\eetF(X)$ be the characteristic functional corresponding to $\WHILEDO{\pguard}{\stmt}$ as defined in \autoref{def:F}.
  Assume, for the moment, that for each $k \in \Nats$, we have
  $\eet{\stmt_k}{\mdpRewSink} = \eetF^{k}(\ctert{0})$.
  Then, using Kleene's Fixed Point Theorem, we can establish that
  \begin{align*}
    \sup_{k \in \Nats} \eet{\stmt_k}{\mdpRewSink} 
    ~=~  \sup_{k \in \Nats} \eetF^{k}(\zero)
    ~=~  \lfp X . \eetF(X)
    ~=~  \eet{\stmt'}{\mdpRewSink}.
  \end{align*}

  Hence, it suffices to show that 
  $\eet{\stmt_k}{\mdpRewSink} = \eetF^{k}(\ctert{0})$
  for each $k \in \Nats$ by induction over $k$.
  
  \paragraph{Induction Base} $k = 0$.
  \begin{align*}
   \eet{\stmt_k}{\mdpRewSink}
   ~=~ \eet{\HALT}{\mdpRewSink}
   ~=~ \ctert{0}
   ~=~ \eetF^{0}(\ctert{0}).
  \end{align*} 
  
  \paragraph{Induction Hypothesis}
  Assume that $\eet{\stmt_k}{\mdpRewSink} = \eetF^{k}(\ctert{0})$ holds for an arbitrary, fixed $k \in \Nats$.

  \paragraph{Induction Step} $k \mapsto k+1$.
  \begin{align*}
    & \eet{\stmt_{k+1}}{\mdpRewSink} \\
    ~=~ & \eet{\ITE{\pguard}{\stmt;\stmt_k}{\EMPTY}}{\mdpRewSink} \tag{Def. $C_{k+1}$} \\
    ~=~ & 1 + \probof{\pguard}{\true} \cdot \eet{\stmt}{ \eet{\stmt_k}{\mdpRewSink} } 
            + \probof{\pguard}{\false} \cdot \eet{\EMPTY}{\mdpRewSink} \tag{\autoref{table:eet-rules}{}}\\
    ~=~ & 1 + \probof{\pguard}{\true} \cdot \eet{\stmt}{\eetF^{k}(\ctert{0})} 
            + \probof{\pguard}{\false} \cdot \eet{\EMPTY}{\mdpRewSink} \tag{I.H.} \\
    ~=~ & \eetF^{k+1}(\ctert{0}) \tag{\autoref{def:F}}.
  \end{align*}
  \qed
\end{proof}
Furthermore, we will use the following decomposition lemma.
\begin{lemma} \label{thm:exprew-seq}
  Let $C_1,C_2 \in \Stmt$, $\mdpRewSink \in \E$, and $\ps \in \ProgramStates$.
  Then
  \begin{align*}
      & \ExpRew{\rmdpStmt{\mdpRewSink}{\ps}{\stmt_1;\stmt_2}}{\mdpSinkState} 
        ~=~ \ExpRew{\rmdpStmt{g(\stmt_2,f)}{\ps}{\stmt_1}}{\mdpSinkState},
  \end{align*}
  where   
  \begin{align*}
     g(\stmt_2,f) ~\triangleq~ \ExpRew{\lambda \rho . \rmdpStmt{\mdpRewSink}{\rho}{\stmt_2}}{\mdpSinkState}.
  \end{align*}
\end{lemma}
\begin{proof}
 The MDP $\rmdpStmt{\mdpRewSink}{\ps}{\stmt_1;\stmt_2}$ is of the following form:
  
  \begin{tikzpicture}[->,>=stealth',shorten >=1pt,node distance=2.7cm,semithick,minimum size=1cm]
\tikzstyle{every state}=[draw=none]

   \node [state, initial, initial text=,] (init) {$\mdpState{\stmt_1;\stmt_2}{\ps}$};  
   \node (initLabel) [node distance=0.5cm,gray,below of=init] {$\mdpRew(\mdpState{\stmt_1}{\ps})$};
   
   \node [state,] (termp) [right of = init] {$\mdpState{\Terminated;\stmt_2}{\ps'}$};
   \node (termpLabel) [node distance=0.5cm,gray,below of=termp] {$0$};
   
   \node [state,] (termp') [below of = termp] {$\mdpState{\Terminated;\stmt_2}{\ps'}$};
   \node (termp'Label) [node distance=0.5cm,gray,below of=termp'] {$0$};
   
   \node [state,] (p) [right of = termp] {$\mdpState{\stmt_2}{\ps'}$};
   \node (pLabel) [node distance=0.5cm,gray,below of=p] {$\mdpRew(\mdpState{\stmt_2}{\ps'})$};
   
   \node [state,] (q) [right of = termp'] {$\mdpState{\stmt_2}{\ps'}$};
   \node (qLabel) [node distance=0.5cm,gray,below of=q] {$\mdpRew(\mdpState{\stmt_2}{\ps'})$};
   
   \node [] (dummy) [below= 0.8 cm of init] {$\vdots$};
   
    \node [] (dummy1) [on grid, right =2cm of p] {$\ldots$};
    \node [] (dummy2) [on grid, right =2cm of q] {$\ldots$};
  \path [] 
      (init) edge [decorate,decoration={snake, post length=2mm}] (termp)
      (init) edge [decorate,decoration={snake, post length=2mm}] (dummy)
      (init) edge [decorate,decoration={snake, post length=2mm}] (termp')
      (termp) edge [] (p)
      (termp') edge [] (q)
      (p) edge [decorate,decoration={snake, post length=2mm}] (dummy1)
      (q) edge [decorate,decoration={snake, post length=2mm}] (dummy2)
  ;
\end{tikzpicture}
  
  Hence, every path starting in $\mdpState{\stmt_1;\stmt_2}{\ps}$ either eventually reaches 
  $\mdpState{\Terminated;\stmt_2}{\ps'}$ for some $\ps' \in \ProgramStates$ and then immediately $\mdpState{\stmt_2}{\ps'}$ or diverges, i.e. never reaches $\mdpSinkState$.
  Since $\mdpState{\stmt_2}{\ps'}$ is the initial state of the MDP $\rmdpStmt{\mdpRewSink}{\ps'}{\stmt_2}$,
  we can transform $\rmdpStmt{\mdpRewSink}{\ps}{\stmt_1;\stmt_2}$ into an MDP $\rmdpStmt{\mdpRewSink'}{\ps}{\stmt_1}$ with the same expected reward by setting
  \begin{align*}
   \mdpRewSink' = \ExpRew{\lambda \rho . \rmdpStmt{\mdpRewSink}{\rho}{\stmt_2}}{\mdpSinkState} = g(\stmt_2,f).
  \end{align*}
  \qed
\end{proof}
The next two lemmas state that the supremum 
\[ sup_{k \in \Nats} \eet{\BOUNDEDWHILE{k}{\pguard}{\stmt}}{\mdpRewSink} \]
is both an upper and a lower bound of the expected reward
\[ \ExpRew{\rmdpStmt{\mdpRewSink}{\ps}{\WHILEDO{\pguard}{\stmt}}}{\mdpSinkState}. \]
\begin{lemma}\label{thm:exprew-bounded-loops-leq}
Let $\pguard \in \DExprs$, $\stmt \in \Stmt$, $\mdpRewSink \in \E$, and $\ps \in \ProgramStates$.
Then
  \[ 
     \sup_{k \in \Nats} \ExpRew{\rmdpStmt{\mdpRewSink}{\ps}{\BOUNDEDWHILE{k}{\pguard}{\stmt}}}{\mdpSinkState}   
     \leq
     \ExpRew{\rmdpStmt{\mdpRewSink}{\ps}{\WHILEDO{\pguard}{\stmt}}}{\mdpSinkState}.
  \] 
\end{lemma}

\begin{proof}
  We prove that
  \[ 
   \ExpRew{\rmdpStmt{\mdpRewSink}{\ps}{\stmt_k}}{\mdpSinkState}
   \leq
   \ExpRew{\rmdpStmt{\mdpRewSink}{\ps}{\stmt'}}{\mdpSinkState}   
  \]
  for each $k \in \Nats$ by induction over $k$.
  \paragraph{Induction Base} $k=0$.
  \begin{align*}
   \ExpRew{\rmdpStmt{\mdpRewSink}{\ps}{\HALT}}{\mdpSinkState}
   ~=~ 0 \leq \ExpRew{\rmdpStmt{\mdpRewSink}{\ps}{\stmt'}}{\mdpSinkState}.
  \end{align*}
  
  \paragraph{Induction Hypothesis}
  Assume that
    \[ 
   \ExpRew{\rmdpStmt{\mdpRewSink}{\ps}{\stmt_k}}{\mdpSinkState}
   \leq
   \ExpRew{\rmdpStmt{\mdpRewSink}{\ps}{\stmt'}}{\mdpSinkState}   
  \]
  holds for an arbitrary, fixed, $k \in \Nats$.
  
  \paragraph{Induction Step} $k \mapsto k+1$.
    \begin{align*}
   & \ExpRew{\rmdpStmt{\mdpRewSink}{\ps}{\stmt_{k+1}}}{\mdpSinkState} \\
   & ~=~ \ExpRew{\rmdpStmt{\mdpRewSink}{\ps}{\ITE{\pguard}{\stmt;\stmt_{k}}{\EMPTY}}}{\mdpSinkState} \tag{Def. $C_{k+1}$} \\
   & ~=~ 1 + \probofInstance{\pguard}{\true}{\ps} \cdot \ExpRew{\rmdpStmt{\mdpRewSink}{\ps}{\stmt;\stmt_k}}{\mdpSinkState} \tag{\autoref{def:operational}} \\
   & \qquad + \probofInstance{\pguard}{\false}{\ps} \cdot \ExpRew{\rmdpStmt{\mdpRewSink}{\ps}{\EMPTY}}{\mdpSinkState}\\
   & ~=~ 1 + \probofInstance{\pguard}{\true}{\ps} \cdot 
   \ExpRew{\rmdpStmt{g(\stmt_k,\mdpRewSink)}{\ps}{\stmt}}{\mdpSinkState} \tag{\autoref{thm:exprew-seq}} \\ 
   & \qquad \qquad
           + \probofInstance{\pguard}{\false}{\ps} \cdot \ExpRew{\rmdpStmt{\mdpRewSink}{\ps}{\EMPTY}}{\mdpSinkState}\\
   & ~\leq~ 1 + \probofInstance{\pguard}{\true}{\ps} \cdot 
   \ExpRew{\rmdpStmt{\ExpRew{\lambda \rho . \rmdpStmt{\mdpRewSink}{\rho}{\stmt'}}{\mdpSinkState} }{\ps}{\stmt}}{\mdpSinkState} \tag{I.H.}
   \\ & \qquad \qquad
           + \probofInstance{\pguard}{\false}{\ps} \cdot \ExpRew{\rmdpStmt{\mdpRewSink}{\ps}{\EMPTY}}{\mdpSinkState}  \\
   & ~=~ 1 + \probofInstance{\pguard}{\true}{\ps} \cdot \ExpRew{\rmdpStmt{\mdpRewSink}{\ps}{\stmt;\stmt'}}{\mdpSinkState} \tag{\autoref{thm:exprew-seq}}
   \\ & \qquad \qquad
           + \probofInstance{\pguard}{\false}{\ps} \cdot \ExpRew{\rmdpStmt{\mdpRewSink}{\ps}{\EMPTY}}{\mdpSinkState}  \\
   & ~=~ \ExpRew{\rmdpStmt{\mdpRewSink}{\ps}{\ITE{\pguard}{\stmt;\stmt'}{\EMPTY}}}{\mdpSinkState} \tag{\autoref{def:operational}} \\
   & ~=~ \ExpRew{\rmdpStmt{\mdpRewSink}{\ps}{\stmt'}}{\mdpSinkState}
  \end{align*}
  where the last step is immediate, because in the MDP $\rmdpStmt{\mdpRewSink}{\ps}{\stmt'}$, state
  $\mdpState{\stmt'}{\ps}$ has 0 reward and immediately reaches 
  $\mdpState{\ITE{\pguard}{\stmt;\stmt'}{\EMPTY}}{\ps}$.
  \qed
\end{proof}
\begin{lemma}\label{thm:exprew-bounded-loops-geq}
Let $\pguard \in \DExprs$, $\stmt \in \Stmt$, $\mdpRewSink \in \E$, and $\ps \in \ProgramStates$.
Then
  \[ 
     \sup_{k \in \Nats} \ExpRew{\rmdpStmt{\mdpRewSink}{\ps}{\BOUNDEDWHILE{k}{\pguard}{\stmt}}}{\mdpSinkState}   
     \geq
     \ExpRew{\rmdpStmt{\mdpRewSink}{\ps}{\WHILEDO{\pguard}{\stmt}}}{\mdpSinkState}.
  \] 
\end{lemma}
\begin{proof}
  Let $\ipath$ be a path in the MDP $\rmdpStmt{\mdpRewSink}{\ps}{\stmt'}$ with $\mdpRew(\ipath) > 0$ 
  starting in $\mdpState{\stmt'}{\ps}$.
  Then, there exists a finite prefix $\fpath$ of $\ipath$ reaching a state $\mdpState{\Terminated}{\ps'}$
  for some $\ps' \in \ProgramStates$ with $\mdpRew(\ipath) = \mdpRew(\fpath)$.
  Since $\fpath$ is finite, only finitely many states with first component $\stmt'$, say $k$, 
  are visited. 
  We show that a corresponding path $\fpath'$ with 
  $\mdpRew(\fpath) = \mdpRew(\fpath')$
  exists in the MDP
  $\rmdpStmt{\mdpRewSink}{\ps}{\stmt_k}$ by induction over $k \geq 1$.
  
  \paragraph{Induction Base} $k = 1$.
  The only finite path $\fpath$ in the MDP $\rmdpStmt{\mdpRewSink}{\ps}{\stmt'}$ reaching a state with first component $\Terminated$ is
  \[ \fpath = \mdpState{\stmt'}{\ps} \mdpState{\ITE{\pguard}{\stmt;\stmt'}{\EMPTY}}{\ps}
              \mdpState{\EMPTY}{\ps} \mdpState{\Terminated}{\ps}
  \]
  The corresponding path in the MDP $\rmdpStmt{\mdpRewSink}{\ps}{\stmt_k}$ is
  \[
     \fpath' = \mdpState{\stmt_k}{\ps} \mdpState{\EMPTY}{\ps} \mdpState{\Terminated}{\ps}
  \]
  with $\mdpRew(\fpath') = \mdpRew(\fpath) = 1 + \mdpRewSink(\ps)$.
  
  \paragraph{Induction Hypothesis}
  For every path $\fpath$ reaching a state with first component $\Terminated$ and positive reward in the MDP $\rmdpStmt{\mdpRewSink}{\ps}{\stmt'}$
  visiting $k > 1$ (for an arbitrary, but fixed $k$) states with first component $\stmt'$, there exists a path
  $\fpath'$ reaching a state with first component $\Terminated$ in the MDP $\rmdpStmt{\mdpRewSink}{\ps}{\stmt_k}$ with
  $\mdpRew(\fpath') = \mdpRew(\fpath)$.
  
  \paragraph{Induction Step} $k \mapsto k+1$.
  Every finite path $\fpath$ in the MDP $\rmdpStmt{\mdpRewSink}{\ps}{\stmt'}$ as described above
  is of the form 
  \[ \fpath = \mdpState{\stmt'}{\ps} \mdpState{\ITE{\pguard}{\stmt;\stmt'}{\EMPTY}}{\ps}
              \mdpState{\stmt;\stmt'}{\ps} ... \mdpState{\stmt'}{\ps'} ... \mdpState{\Terminated}{\ps''}
  \]
  such that $k$ states with first component $\stmt'$ are visited starting from state
  $\mdpState{\stmt'}{\ps'}$.
  Let $\fpath_2$ be the path starting in this state.
  By I.H. there exists a corresponding path $\fpath_2'$ in the MDP 
  $\rmdpStmt{\mdpRewSink}{\ps'}{\stmt_k}$ with 
  $\mdpRew(\fpath_2) = \mdpRew(\fpath_2')$.
  Now,
  \[ \fpath' = 
     \mdpState{\stmt_{k+1}}{\ps} 
     \ldots
     \mdpState{\Terminated;\stmt_k}{\ps'}
     \fpath_2'
  \]
  is a path in $\rmdpStmt{\mdpRewSink}{\ps'}{\stmt_{k+1}}$ with 
  $\mdpRew(\fpath') = \mdpRew(\fpath)$.

  Hence, for every path $\fpath$ with positive reward in the MDP $\rmdpStmt{\mdpRewSink}{\ps}{\stmt'}$,
  there exists a corresponding path $\fpath'$ in the MDP $\rmdpStmt{\mdpRewSink}{\ps}{\stmt_k}$ for some $k \in \Nats$.
  Thus, we include all paths with positive reward in the MDP $\rmdpStmt{\mdpRewSink}{\ps}{\stmt'}$ by taking the supremum
  of the expected reward of the MDPs $\rmdpStmt{\mdpRewSink}{\ps}{\stmt_k}$ over $k \in \Nats$. 
  \qed
\end{proof}
We are now in a position to show soundness of $\eet{C}{f}$ with respect to $\rmdpStmt{\mdpRewSink}{\ps_0}{\stmt}$ by induction on the structure of $C$, 
where the proof for the case of the while loop reduces basically to application of \autoref{thm:exprew-bounded-loops-leq} and \autoref{thm:exprew-bounded-loops-geq}.
\begingroup
\def\thetheorem{\ref{thm:eet-soundness}}
\begin{theorem}[Soundness of \eetsymbol]
Let $\pguard \in \DExprs$, $\stmt \in \Stmt$, and $\mdpRewSink \in \E$. Then, for each $\ps \in \ProgramStates$, we have
  \[ 
      \ExpRew{\rmdpStmt{\mdpRewSink}{\ps}{\stmt}}{\mdpSinkState} ~=~ \eet{\stmt}{\mdpRewSink}(\ps)~.
  \]
\end{theorem}
\addtocounter{theorem}{-1}
\endgroup
\begin{proof}
  We have to show that 
  \[ 
      \ExpRew{\rmdpStmt{\mdpRewSink}{\ps}{\stmt}}{\mdpSinkState} = \eet{\stmt}{\mdpRewSink}(\ps)
  \]
  holds for each $\stmt \in \Stmt$, $\ps \in \ProgramStates$ and $\mdpRewSink \in \E$.
  We prove the claim by induction on the structure of program statements $\stmt \in \Stmt$.
  The base cases are $\stmt = \EMPTY$, $\stmt = \SKIP$ $\stmt = \ASSIGN{x}{\pexpr}$ and $\stmt = \HALT$.
  
  \paragraph{The empty program} $\stmt = \EMPTY$.
   
  The MDP $\rmdpStmt{\mdpRewSink}{\ps}{\EMPTY}$ contains exactly one infinite
    path $\ipath$ which has the following form:
    
    \begin{tikzpicture}[->,>=stealth',shorten >=1pt,node distance=2.5cm,semithick,minimum size=2cm]
\tikzstyle{every state}=[draw=none]
  \draw[white, use as bounding box] (-1.6,-0.85) rectangle (6.5,.55);
  
   \node [state, initial, initial text=,] (init) {$\mdpState{\EMPTY}{\ps}$};  
   \node (initLabel) [node distance=0.5cm,gray,below of=init] {$0$};
   
   \node [state] (exit) [right of=init] {$\mdpState{\Terminated}{\ps}$};
   \node (exitLabel) [node distance=0.5cm,gray,below of=exit] {$\mdpRewSink(\ps)$};
   
   \node [state] (sink) [right of=exit] {$\mdpSinkState$};
   \node (sinkLabel) [node distance=0.5cm,gray,below of=sink] {$0$};

  \path [] 
      (init) edge [] node [above=-0.8cm] {\scriptsize{$1$}} (exit)
      (exit) edge [] node [above=-0.8cm] {\scriptsize{$1$}} (sink)
      (sink) edge [loop right] node [right=-0.8cm] {\scriptsize{$1$}} (sink)
  ;
\end{tikzpicture}
    
    
   Moreover, $\mdpSPaths{\mdpState{\EMPTY}{\ps}}{\mdpSinkState} = \{ \mdpState{\EMPTY}{\ps} \mdpState{\Terminated}{\ps} \mdpSinkState \}$.
   Thus, we have
   \begin{align*}
    & \ExpRew{\rmdpStmt{\mdpRewSink}{\ps}{\EMPTY}}{\mdpSinkState} \\
    & ~=~ \sup_{\mdpSched} \sum_{\fpath \in \mdpSPaths{\mdpState{\EMPTY}{\ps}}{\mdpSinkState}} \mdpPr{}{\fpath} \cdot \mdpRew(\fpath) \\
    & ~=~ \sup_{\mdpSched} (1 \cdot (\mdpRew(\mdpState{\EMPTY}{\ps} \mdpState{\Terminated}{\ps} \mdpSinkState)))  \\
    & ~=~ 1 \cdot (0+\mdpRewSink(\ps)+0) = \mdpRewSink(\ps) \\
    & ~=~ \eet{\EMPTY}{\mdpRewSink}(\ps).
   \end{align*}
  
  \paragraph{The effectless time-consuming program} $\stmt = \SKIP$.
  
  The MDP $\rmdpStmt{\mdpRewSink}{\ps}{\SKIP}$ contains exactly one infinite $\ipath$ which has the following form:
  
    \begin{tikzpicture}[->,>=stealth',shorten >=1pt,node distance=2.5cm,semithick,minimum size=2cm]
\tikzstyle{every state}=[draw=none]
  \draw[white, use as bounding box] (-1.6,-0.85) rectangle (6.5,.55);
  
   \node [state, initial, initial text=,] (init) {$\mdpState{\SKIP}{\ps}$};  
   \node (initLabel) [node distance=0.5cm,gray,below of=init] {$1$};
   
   \node [state] (exit) [right of=init] {$\mdpState{\Terminated}{\ps}$};
   \node (exitLabel) [node distance=0.5cm,gray,below of=exit] {$\mdpRewSink(\ps)$};
   
   \node [state] (sink) [right of=exit] {$\mdpSinkState$};
   \node (sinkLabel) [node distance=0.5cm,gray,below of=sink] {$0$};

  \path [] 
      (init) edge [] node [above=-0.8cm] {\scriptsize{$1$}} (exit)
      (exit) edge [] node [above=-0.8cm] {\scriptsize{$1$}} (sink)
      (sink) edge [loop right] node [right=-0.8cm] {\scriptsize{$1$}} (sink)
  ;
\end{tikzpicture}
    
   Moreover, $\mdpSPaths{\mdpState{\SKIP}{\ps}}{\mdpSinkState} = \{ \mdpState{\SKIP}{\ps} \mdpState{\Terminated}{\ps} \mdpSinkState \}$.
   Thus, we have
   \begin{align*}
    & \ExpRew{\rmdpStmt{\mdpRewSink}{\ps}{\SKIP}}{\mdpSinkState} \\
    & ~=~ \sup_{\mdpSched} \sum_{\fpath \in \mdpSPaths{\mdpState{\SKIP}{\ps}}{\mdpSinkState}} \mdpPr{}{\fpath} \cdot \mdpRew(\fpath) \\
    & ~=~ \sup_{\mdpSched} (1 \cdot (\mdpRew(\mdpState{\SKIP}{\ps} \mdpState{\Terminated}{\ps} \mdpSinkState)))  \\
    & ~=~ 1 \cdot (1+\mdpRewSink(\ps)+0) = 1 + \mdpRewSink(\ps) \\
    & ~=~ \eet{\SKIP}{\mdpRewSink}(\ps).
   \end{align*}

   \paragraph{The probabilistic assignment} $\stmt = \ASSIGN{x}{\pexpr}$.
   
   For some $n \in \Nats$, the MDP $\rmdpStmt{\mdpRewSink}{\ps}{\ASSIGN{x}{\pexpr}}$ is of the following form:
   
   
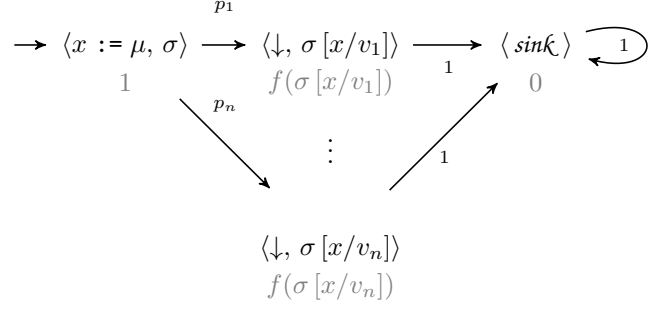
\begin{tikzpicture}[->,>=stealth',shorten >=1pt,node distance=2.7cm,semithick,minimum size=1cm]
\tikzstyle{every state}=[draw=none]
  
   \node [state, initial, initial text=] (init) {$\mdpState{\ASSIGN{x}{\pexpr}}{\ps}$};  
   \node [gray] (initLabel) [node distance=0.5cm, below of = init] {$1$};  
   
   \node [state] (p) [right of = init] {$\mdpState{\Terminated}{\ps\subst{x}{v_1}}$};
   \node [gray] (pLabel) [node distance=0.5cm, below of = p] {$\mdpRewSink(\ps\subst{x}{v_1})$};

   \node [state] (q) [below of = p] {$\mdpState{\Terminated}{\ps\subst{x}{v_n}}$};
   \node [gray] (qLabel) [node distance=0.5cm, below of = q] {$\mdpRewSink(\ps\subst{x}{v_n})$};

   \node [state] (sink) [right of=p] {$\mdpSinkState$};
   \node (sinkLabel) [node distance=0.5cm,gray,below of=sink] {$0$};
      
    \node [] (dummy1) [node distance=1.3cm, below of = p] {$\vdots$};
  \path [] 
      (init) edge [] node [above] {\scriptsize{$p_1$}} (p)
      (init) edge [] node [above] {\scriptsize{$p_n$}} (q)
      (p) edge [] node [above=-0.8cm] {\scriptsize{$1$}} (sink)
      (q) edge [] node [above=-0.8cm] {\scriptsize{$1$}} (sink)
      (sink) edge [loop right] node [right=-0.8cm] {\scriptsize{$1$}} (sink)
  ;
\end{tikzpicture}
   
   Hence, it consists of paths
   \[ \ipath = \mdpState{\ASSIGN{x}{\pexpr}}{\ps} \mdpState{\Terminated}{\ps\subst{x}{v_i}} \mdpSinkState \mdpSinkState \ldots \]
   with $\mdpPr{}{\ipath} = p_i$ for each $v_i \in \Vals$ with $\sem{\pguard}(\ps)(v_i) = p_i > 0$ and
   \[ \sum_{k=1}^{n} p_{k} = 1. \]
   Moreover,
   \[
      \mdpSPaths{\mdpState{\ASSIGN{x}{\pexpr}}{\ps}}{\mdpSinkState} = \{ \fpath_i ~|~ \fpath_i =   \mdpState{\ASSIGN{x}{\pexpr}}{\ps} \mdpState{\Terminated}{\ps\subst{x}{v_i}} \mdpSinkState \}.
   \]
  For each $\fpath_{i} \in \mdpSPaths{\mdpState{\ASSIGN{x}{\pexpr}}{\ps}}{\mdpSinkState}$, $\mdpRew(\fpath_i) = 1 + \mdpRewSink\subst{x}{v_i}$.
  Thus, we have 
   \begin{align*}
    & \ExpRew{\rmdpStmt{\mdpRewSink}{\ps}{\ASSIGN{x}{\pexpr}}}{\mdpSinkState} \\
    & ~=~ \sup_{\mdpSched} \sum_{\fpath_{i} \in \mdpSPaths{\mdpState{\ASSIGN{x}{\pexpr}}{\ps}}{\mdpSinkState}} \mdpPr{}{\fpath_{i}} \cdot \mdpRew(\fpath_{i}) \\
    & ~=~ \sum_{\fpath_{i} \in \mdpSPaths{\mdpState{\ASSIGN{x}{\pexpr}}{\ps}}{\mdpSinkState}} \mdpPr{}{\fpath_{i}} \cdot \mdpRew(\fpath_{i}) \\
    & ~=~ \sum_{\fpath_{i} \in \mdpSPaths{\mdpState{\ASSIGN{x}{\pexpr}}{\ps}}{\mdpSinkState}} \mdpPr{}{\fpath_{i}} \cdot (1 + \mdpRewSink\subst{x}{v_i})  \\
    & ~=~ 1 + \sum_{\fpath_{i} \in \mdpSPaths{\mdpState{\ASSIGN{x}{\pexpr}}{\ps}}{\mdpSinkState}} \mdpPr{}{\fpath_{i}} \cdot \mdpRewSink\subst{x}{v_i}  \\
    & ~=~ 1 + \sum_{\fpath_{i} \in \mdpSPaths{\mdpState{\ASSIGN{x}{\pexpr}}{\ps}}{\mdpSinkState}} \sem{\pguard}(\ps)(v_i) \cdot \mdpRewSink\subst{x}{v_i}  \\
    & ~=~ 1 + \ev{\sem{\pguard}(\ps)}{\lambda v_i. \mdpRewSink\subst{x}{v_i}} \\
    & ~=~ \eet{\ASSIGN{x}{E}}{\mdpRewSink}(\ps)
   \end{align*}

  \paragraph{The faulty program} $\stmt = \HALT$.
  
    The MDP $\rmdpStmt{\mdpRewSink}{\ps}{\HALT}$ is of the following form:
    
    \begin{tikzpicture}[->,>=stealth',shorten >=1pt,node distance=2.5cm,semithick,minimum size=2cm]
\tikzstyle{every state}=[draw=none]
  \draw[white, use as bounding box] (-1.6,-0.85) rectangle (6.5,.55);
  
   \node [state, initial, initial text=,] (init) {$\mdpState{\HALT}{\ps}$};  
   \node (initLabel) [node distance=0.5cm,gray,below of=init] {$0$};
   
   \node [state] (sink) [right of=init] {$\mdpSinkState$};
   \node (sinkLabel) [node distance=0.5cm,gray,below of=sink] {$0$};

  \path [] 
      (init) edge [] node [above=-0.8cm] {\scriptsize{$1$}} (sink)
      (sink) edge [loop right] node [right=-0.8cm] {\scriptsize{$1$}} (sink)
  ;
\end{tikzpicture}

   Hence, $\mdpSPaths{\mdpState{\HALT}{\ps}}{\mdpSinkState} = \{ \mdpState{\HALT}{\ps} \mdpSinkState \}$.
   Thus, we have
   \begin{align*}
    & \ExpRew{\rmdpStmt{\mdpRewSink}{\ps}{\HALT}}{\mdpSinkState} \\
    & ~=~ \sup_{\mdpSched} \sum_{\fpath \in \mdpSPaths{\mdpState{\HALT}{\ps}}{\mdpSinkState}} \mdpPr{}{\fpath} \cdot \mdpRew(\fpath) \\
    & ~=~ \sup_{\mdpSched} (1 \cdot (\mdpRew(\mdpState{\HALT}{\ps} \mdpSinkState)))  \\
    & ~=~ 1 \cdot (0+0) = 0 \\
    & ~=~ \eet{\HALT}{\mdpRewSink}(\ps).
   \end{align*}

  \paragraph{Induction Hypothesis:} 
    For all (substatements) $\stmt' \in \Stmt$ of $\stmt$ and $\ps \in \ProgramStates$ and $\mdpRewSink : \ProgramStates \to \E$, we have
   \[ 
      \ExpRew{\rmdpStmt{\mdpRewSink}{\ps}{\stmt'}}{\mdpSinkState} = \eet{\stmt'}{\mdpRewSink}(\ps).
   \]

   For the induction step, we have to consider sequential composition, conditionals, non--deterministic choice
   and loops.
   
  \paragraph{The sequential composition} $\stmt = \stmt_1;\stmt_2$.
%
%
%
  \begin{align*}
       & \ExpRew{\rmdpStmt{\mdpRewSink}{\ps}{\stmt_1;\stmt_2}}{\mdpSinkState} \\
   ~=~ & \ExpRew{\rmdpStmt{\ExpRew{\lambda \rho . \rmdpStmt{\mdpRewSink}{\rho}{\stmt_2}}{\mdpSinkState}}{\ps}{\stmt_1}}{\mdpSinkState} \tag{\autoref{thm:exprew-seq}} \\   
   ~=~& \ExpRew{\rmdpStmt{\lambda \rho . \eet{\stmt_2}{\mdpRewSink}(\rho)}{\ps}{\stmt_1}}{\mdpSinkState} \tag{I.H. on $\stmt_2$} \\   
   ~=~& \ExpRew{\rmdpStmt{\eet{\stmt_2}{\mdpRewSink}}{\ps}{\stmt_1}}{\mdpSinkState}  \\   
   ~=~& \eet{\stmt_1}{\eet{\stmt_2}{\mdpRewSink}}(\ps) \tag{I.H. on $\stmt_1$} \\
   ~=~& \eet{\stmt_1;\stmt_2}{\mdpRewSink}(\ps).
  \end{align*}

  \paragraph{The conditional} $\stmt = \ITE{\pguard}{\stmt_1}{\stmt_2}$.
  
  The MDP $\rmdpStmt{\mdpRewSink}{\ps}{\ITE{\pguard}{\stmt_1}{\stmt_2}}$ is of the following form:
  
  \begin{tikzpicture}[->,>=stealth',shorten >=1pt,node distance=2.7cm,semithick,minimum size=1cm]
\tikzstyle{every state}=[draw=none]
  \draw[white, use as bounding box] (-2.5,-3.5) rectangle (8.8,.8);

   \node [state, initial, initial text=,] (init) {$\mdpState{\ITE{\pguard}{\stmt_1}{\stmt_2}}{\ps}$};  
   \node (initLabel) [node distance=0.5cm,gray,below of=init] {$1$};

   \node [state,] (p) [node distance = 4.7cm, right of = init] {$\mdpState{\stmt_1}{\ps}$};
   \node (pLabel) [node distance=0.5cm,gray,below of=p] {$\mdpRew(\mdpState{\stmt_1}{\ps})$};
   
   \node [state,] (q) [below of = init] {$\mdpState{\stmt_2}{\ps}$};
   \node (qLabel) [node distance=0.5cm,gray,below of=q] {$\mdpRew(\mdpState{\stmt_2}{\ps})$};
   
    \node [] (dummy1) [on grid, right =2cm of p] {$\ldots$};
    \node [] (dummy2) [on grid, right =2cm of q] {$\ldots$};
  \path [] 
      (init) edge [] node [below=-0.8cm] {\scriptsize{$\probof{\pguard}{\true}$}} (p)
      (initLabel) edge [] node [right] {\scriptsize{$\probof{\pguard}{\false}$}} (q)
      (p) edge [decorate,decoration={snake, post length=2mm}] (dummy1)
      (q) edge [decorate,decoration={snake, post length=2mm}] (dummy2)
  ;
\end{tikzpicture}
  
  Thus, every path of the MDP $\rmdpStmt{\mdpRewSink}{\ps}{\ITE{\pguard}{\stmt_1}{\stmt_2}}$ starting in the initial state $\mdpState{\ITE{\pguard}{\stmt_1}{\stmt_2}}{\ps}$
  either reaches $\mdpState{\stmt_1}{\ps}$ with probability $\probofInstance{\pguard}{\true}{\ps}$ or $\mdpState{\stmt_2}{\ps}$ with probability $\probofInstance{\pguard}{\false}{\ps}$.
  These states are the initial states of the MDPs $\rmdpStmt{\mdpRewSink}{\ps}{\stmt_1}$ and $\rmdpStmt{\mdpRewSink}{\ps}{\stmt_2}$, respectively.
  Hence, 
  \begin{align*}
   & \ExpRew{\rmdpStmt{\mdpRewSink}{\ps}{\ITE{\pguard}{\stmt_1}{\stmt_2}}}{\mdpSinkState} \\
   & ~=~ 1 + \probofInstance{\pguard}{\true}{\ps} \cdot \ExpRew{\rmdpStmt{\mdpRewSink}{\ps}{\stmt_1}}{\mdpSinkState} \\
   & \qquad\qquad + \probofInstance{\pguard}{\false}{\ps} \cdot \ExpRew{\rmdpStmt{\mdpRewSink}{\ps}{\stmt_2}}{\mdpSinkState}  \\
   & ~=~ 1 + \probofInstance{\pguard}{\true}{\ps} \cdot \eet{\stmt_1}{\mdpRewSink}(\ps) 
           + \probofInstance{\pguard}{\false}{\ps} \cdot \eet{\stmt_2}{\mdpRewSink}(\ps) \tag{I.H.} \\
   & ~=~ \eet{\ITE{\pguard}{\stmt_1}{\stmt_2}}{\mdpRewSink}(\ps).
  \end{align*}  
  
 \paragraph{The non--deterministic choice} $\stmt = \NDCHOICE{\stmt_1}{\stmt_2}$.
 
 \begin{tikzpicture}[->,>=stealth',shorten >=1pt,node distance=2.7cm,semithick,minimum size=1cm]
\tikzstyle{every state}=[draw=none]
  \draw[white, use as bounding box] (-2.5,-3.5) rectangle (8.8,.8);

   \node [state, initial, initial text=,] (init) {$\mdpState{\NDCHOICE{\stmt_1}{\stmt_2}}{\ps}$};  
   \node (initLabel) [node distance=0.5cm,gray,below of=init] {$1$};

   \node [state,] (p) [node distance = 4.7cm, right of = init] {$\mdpState{\stmt_1}{\ps}$};
   \node (pLabel) [node distance=0.5cm,gray,below of=p] {$\mdpRew(\mdpState{\stmt_1}{\ps})$};
   
   \node [state,] (q) [below of = init] {$\mdpState{\stmt_2}{\ps}$};
   \node (qLabel) [node distance=0.5cm,gray,below of=q] {$\mdpRew(\mdpState{\stmt_2}{\ps})$};
   
    \node [] (dummy1) [on grid, right =2cm of p] {$\ldots$};
    \node [] (dummy2) [on grid, right =2cm of q] {$\ldots$};
  \path [] 
      (init) edge [] node [below=-0.8cm] {\scriptsize{$\actL, 1$}} (p)
      (initLabel) edge [] node [right] {\scriptsize{$\actR, 1$}} (q)
      (p) edge [decorate,decoration={snake, post length=2mm}] (dummy1)
      (q) edge [decorate,decoration={snake, post length=2mm}] (dummy2)
  ;
\end{tikzpicture}
 
 Every path of the MDP $\rmdpStmt{\mdpRewSink}{\ps}{\NDCHOICE{\stmt_1}{\stmt_2}}$ starting in the initial state $\mdpState{\NDCHOICE{\stmt_1}{\stmt_2}}{\ps}$ either reaches
 $\mdpState{\stmt_1}{\ps}$ by taking action $\actL$ or $\mdpState{\stmt_2}{\ps}$ by taking action $\actR$ with probability 1.
 Hence,
 \begin{align*}
   & \ExpRew{\rmdpStmt{\mdpRewSink}{\ps}{\NDCHOICE{\stmt_1}{\stmt_2}}}{\mdpSinkState} \\
   & ~=~  \sup_{\mdpSched} \{ \ExpRew{\rmdpStmt{\mdpRewSink}{\ps}{\stmt_1}}{\mdpSinkState}, \ExpRew{\rmdpStmt{\mdpRewSink}{\ps}{\stmt_2}}{\mdpSinkState} \} \\
   & ~=~  \sup_{\mdpSched} \{ \eet{\stmt_1}{\mdpRewSink}(\ps), \eet{\stmt_2}{\mdpRewSink}(\ps) \} \tag{I.H.} \\
   & ~=~  \max \{ \eet{\stmt_1}{\mdpRewSink}(\ps), \eet{\stmt_2}{\mdpRewSink}(\ps) \} \\
   & ~=~  \eet{\NDCHOICE{\stmt_1}{\stmt_2}}{\mdpRewSink}(\ps).
 \end{align*}

 \paragraph{The loop} $\stmt = \WHILEDO{\pguard}{\stmt'}$.
 
 For any natural number $k \geq 1$ and $\ps \in \ProgramStates$, we have
 \begin{align*}
    & \eet{\WHILE^{<k}~\pguard~\Do~\stmt'}{\mdpRewSink}(\ps) \\
    & ~=~ \eet{\ITE{\pguard}{\stmt';\BOUNDEDWHILE{k-1}{\pguard}{\stmt'}}{\EMPTY}}{\mdpRewSink}(\ps) \\
    & ~=~ 1 + \probofInstance{\pguard}{\true}{\ps} \cdot \eet{\stmt';\BOUNDEDWHILE{k-1}{\pguard}{\stmt'}}{\mdpRewSink}(\ps) \\
    & \qquad\qquad+  \probofInstance{\pguard}{\false}{\ps} \cdot \eet{\EMPTY}{\mdpRewSink}(\ps) \\
    & ~=~ 1 + \probofInstance{\pguard}{\true}{\ps} \cdot \ExpRew{\rmdpStmt{\mdpRewSink}{\ps}{\stmt';\BOUNDEDWHILE{k-1}{\pguard}{\stmt'}}}{\mdpSinkState} \tag{I.H.} \\
    & \qquad\qquad+ \probofInstance{\pguard}{\false}{\ps} \cdot \ExpRew{\rmdpStmt{\mdpRewSink}{\ps}{\EMPTY}}{\mdpSinkState}  \\
    & ~=~ \ExpRew{\rmdpStmt{\mdpRewSink}{\ps}{\BOUNDEDWHILE{k}{\pguard}{\stmt}}}{\mdpSinkState}.
 \end{align*}
 Together with the already proven proposition 
 \[ \eet{\HALT}{\mdpRewSink}(\ps) = \ExpRew{\rmdpStmt{\mdpRewSink}{\ps}{\HALT}}{\mdpSinkState}, \]
 we can establish that
 \begin{align*}
  & \eet{\WHILEDO{\pguard}{\stmt'}}{f}(\ps) \\
  & ~=~ \sup_{k \in \Nats} \eet{\BOUNDEDWHILE{k}{\pguard}{\stmt}}{\mdpRewSink}(\ps) \tag{\autoref{thm:eet-bounded-loops}} \\
  & ~=~ \sup_{k \in \Nats} \ExpRew{\rmdpStmt{\mdpRewSink}{\ps}{\BOUNDEDWHILE{k}{\pguard}{\stmt}}}{\mdpSinkState} \\
  & ~=~ \ExpRew{\rmdpStmt{\mdpRewSink}{\ps}{\WHILEDO{\pguard}{\stmt'}}}{\mdpSinkState} 
        \tag{\autoref{thm:exprew-bounded-loops-leq},~\autoref{thm:exprew-bounded-loops-geq}}.
 \end{align*} 
 \qed
\end{proof}
\subsection{\eetsymbol{} of Deterministic Programs} \label{sec:proof-thm-eet-axiomatic-soundness}

We will make use of the following Lemmas

\begin{lemma}\label{thm:det-eet-sum}
Let $\stmt_1\in \DetStmt$ terminate on $\ps \in \ProgramStates$ and let $\stmt_2 \in \DetStmt$ terminate on $\opSem{\stmt_1}{\ps}$.
Then
\[ \eet{\stmt_1;\stmt_2}{\ctert{0}}(\ps) = \eet{\stmt_1}{\ctert{0}}(\ps) + \eet{\stmt_2}{\ctert{0}}(\opSem{\stmt_1}{\ps}). \]
\end{lemma}
\begin{proof}
 Immediate by inspection of the MDP $\rmdpStmt{\ctert{0}}{\ps}{\stmt_1;\stmt_2}$ and Theorem~\ref{thm:eet-soundness}.
 \qed
\end{proof}

\begin{lemma} \label{thm:det-eet-loop}
 For each $\ps \in \ProgramStates$ with $\ExpToFun{\pguard}(\ps) = 1$ such that $\WHILEDO{\pguard}{\stmt}$ terminates on $\ps$, we have 
 \begin{align*}
  \eet{\WHILEDO{\pguard}{\stmt}}{\ctert{0}}(\ps) \geq \eet{\stmt}{\ctert{0}}(\ps) + \eet{\WHILEDO{\pguard}{\stmt}}{\ctert{0}}(\opSem{\stmt}{\ps}).
 \end{align*}
\end{lemma}

\begin{proof}
 \begin{align*}
  & \eet{\WHILEDO{\pguard}{\stmt}}{\ctert{0}}(\ps) \\
  & ~=~ \eet{\ITE{\pguard}{\stmt;\WHILEDO{\pguard}{\stmt}}{\EMPTY}}{\ctert{0}}(\ps)  \tag{\autoref{thm:eet-unroll-loop}} \\
  & ~=~ \one + \ExpToFun{\pguard}(\ps) \cdot \eet{\stmt;\WHILEDO{\pguard}{\stmt}}{\ctert{0}}(\ps) + \ExpToFun{\neg\pguard}(\ps) \cdot \zero \\
  & ~=~ \one + \eet{\stmt;\WHILEDO{\pguard}{\stmt}}{\ctert{0}}(\ps) \tag{$\ExpToFun{\pguard}(\ps) = 1$} \\
  & ~=~ \one + \eet{\stmt}{\ctert{0}}(\ps) + \eet{\WHILEDO{\pguard}{\stmt}}{\ctert{0}}(\opSem{\stmt}{\ps} \tag{\autoref{thm:det-eet-sum}} \\
  & ~\geq~ \eet{\stmt}{\ctert{0}}(\ps) + \eet{\WHILEDO{\pguard}{\stmt}}{\ctert{0}}(\opSem{\stmt}{\ps}
 \end{align*}
 \qed
\end{proof}
\begingroup
\def\thetheorem{\ref{thm:eet-axiomatic-soundness}}
\begin{theorem}[Soundness of \eetsymbol{} for deterministic programs]
 For all $\stmt \in \DetStmt$ and assertions $P,Q$, we have
 \[
  \hprove \hrtriple{P}{\stmt}{Q} \text{ implies } \nprove \htriple{P}{\stmt}{\eet{\stmt}{\ctert{0}}}{Q}.
 \]
\end{theorem}
\addtocounter{theorem}{-1}
\endgroup
\begin{proof}
 By induction on the structure of program statements $\stmt \in \DetStmt$.
 
 The base cases $\stmt = \SKIP$ and $\stmt = \ASSIGN{x}{E}$ are immediate, because
 \[ \eet{\SKIP}{\ctert{0}} = \eet{\ASSIGN{x}{E}} = 1 \] 
 and 
 $\htriple{P}{\SKIP}{1}{P}$ as well as $\htriple{P}{\ASSIGN{x}{E}}{1}{P}$ are axioms.
 
 \paragraph{Induction hypothesis}
 Assume that for each sub-statement $\stmt'$ of $\stmt$ and each pair of assertions $P,Q$, we have
 \[
  \hprove \hrtriple{P}{\stmt'}{Q} \text{ implies } \nprove \htriple{P}{\stmt'}{\eet{\stmt'}{\ctert{0}}}{Q}.
 \]
 
 For the induction step, we have to consider the sequential composition of programs, conditionals and loops. 

 \paragraph{Sequential Composition} $\stmt' = \stmt_1;\stmt_2$
 Assume that 
 \[ \hprove \hrtriple{P}{\stmt_1;\stmt_2}{Q}. \]
 Then, there exists an intermediate assertion $R$ such that 
 \[ \hprove \hrtriple{P}{\stmt_1}{R} \text{ and } \hprove \hrtriple{R}{\stmt_2}{Q}. \]
 By I.H., we know that 
 \[ \nprove \htriple{P}{\stmt_1}{\eet{\stmt_1}{\ctert{0}}}{R} \text{ and } \nprove \htriple{R}{\stmt_2}{\eet{\stmt_2}{\ctert{0}}}{Q}. \]
 Now, let $E_2' = \eet{\stmt_1;\stmt_2}{\ctert{0}} - \eet{\stmt_1}{\ctert{0}}$ and consider the triple
 \[ \htriple{P \wedge E_2' = u}{\stmt_1}{\eet{\stmt_1}{\ctert{0}}}{R \wedge \eet{\stmt_2}{\ctert{0}} \leq u} \]
 where $u$ is a fresh logical variable.
 Since $u$ does not occur in $P$, we know that for each $\ps \in \ProgramStates$ with $\ps \models P$,
 $\ps \models P \wedge E_2' = u$.
 Then, 
 \begin{align*}  
  & \ps \models P \wedge E_2' = u \text{ and } \opSem{\stmt_1}{\ps} \models R \wedge \eet{\stmt_2}{\ctert{0}}(\opSem{\stmt_1}{\ps}) \leq u \\
  ~\Leftrightarrow~ & \ps \models P \wedge E_2' = u \text{ and } \opSem{\stmt_1}{\ps} \models R \text{ and} \\
  & \eet{\stmt_2}{\ctert{0}}(\opSem{\stmt_1}{\ps}) \leq \eet{\stmt_1;\stmt_2}{\ctert{0}}(\ps) - \eet{\stmt_1}{\ctert{0}}(\ps) \tag{Def. of $E_2'$} \\
  ~\Leftrightarrow~ & \ps \models P \wedge E_2' = u \text{ and } \opSem{\stmt_1}{\ps} \models R \text{ and} \\
  &  \eet{\stmt_1}{\ctert{0}}(\ps) + \eet{\stmt_2}{\ctert{0}}(\opSem{\stmt_1}{\ps}) \leq \eet{\stmt_1;\stmt_2}{\ctert{0}}(\ps) \tag{\autoref{thm:det-eet-sum}}\\
  ~\Rightarrow~ & \models_{E} \htriple{P \wedge E_2' = u}{\stmt_1}{\eet{\stmt_1}{\ctert{0}}}{R \wedge \eet{\stmt_2}{\ctert{0}} \leq u}
 \end{align*}
 Since both the triples
 \[ \htriple{P \wedge E_2' = u}{\stmt_1}{\eet{\stmt_1}{\ctert{0}}}{R \wedge \eet{\stmt_2}{\ctert{0}} \leq u} \]
 and 
 \[ \htriple{R}{\stmt_2}{\eet{\stmt_2}{\ctert{0}}}{Q}, \]
 are valid, we may apply to rule of sequential composition to conclude 
 \begin{align*}
  & \nprove \htriple{P}{\stmt_1;\stmt_2}{\eet{\stmt_1}{\ctert{0}}+E_2'}{Q} \\
  ~\Leftrightarrow~ & \nprove \htriple{P}{\stmt_1;\stmt_2}{\eet{\stmt_1;\stmt_2}{\ctert{0}}}{Q}. \tag{Def. of $E_2'$}
 \end{align*}
 
 \paragraph{Conditionals} $\stmt' = \ITE{\pguard}{\stmt_1}{\stmt_2}$
 Assume that 
 \[ \hprove \hrtriple{P}{\ITE{\pguard}{\stmt_1}{\stmt_2}}{Q}. \]
 Then, we also know that 
 \[ \hprove \hrtriple{P \wedge \pguard}{\stmt_1}{Q} \text{ and } \hprove \hrtriple{P \wedge \neg \pguard}{\stmt_2}{Q}. \]
 By I.H., we know that 
 \begin{align*}
  & \nprove \htriple{P \wedge \pguard}{\stmt_1}{\eet{\stmt_1}{\ctert{0}}}{Q} \text{ and } \\
  & \nprove \htriple{P \wedge \neg \pguard}{\stmt_2}{\eet{\stmt_2}{\ctert{0}}}{Q}.
 \end{align*}
 Now let
 \begin{align*}
   E := & \eet{\ITE{\pguard}{\stmt_1}{\stmt_2}}{\ctert{0}} \\
      = & \one + \ExpToFun{\pguard} \cdot \eet{\stmt_1}{\ctert{0}} + \ExpToFun{\neg\pguard} \cdot \eet{\stmt_2}{\ctert{0}}. \\
 \end{align*}
 Then $E \geq \eet{\stmt_1}{\ctert{0}}$ and $E \geq \eet{\stmt_2}{\ctert{0}}$ (for the states satisfying precondition $P$).
 Hence, we can apply the rule of consequence to conclude 
 \[ \nprove \htriple{P \wedge \pguard}{\stmt_1}{E}{Q} \text{ and } \nprove \htriple{P \wedge \neg \pguard}{\stmt_2}{E}{Q}. \]
 Now, applying the rule for conditionals, yields 
 \[ \nprove \htriple{P}{\ITE{\pguard}{\stmt_1}{\stmt_2}}{E}{Q}. \]

 \paragraph{Loops} $\stmt' = \WHILEDO{\pguard}{\stmt_1}$.
 Assume that 
 \[ \hprove \hrtriple{P}{\WHILEDO{\pguard}{\stmt_1}}{Q}. \]
 Then, there exists an assertion $R(z)$ such that $P \Rightarrow \exists z \mydot R(z)$, $R(0) \Rightarrow Q$ and
 \[ \hprove \hrtriple{\exists z \mydot R(z)}{\WHILEDO{\pguard}{\stmt_1}}{R(0)}. \]
 By the \WHILE-rule of Hoare logic for total correctness, we have
 \begin{align*}
  \hprove \hrtriple{R(z+1)}{\stmt_1}{R(z)}. \tag{$\ast$}
 \end{align*}
 By I.H., we obtain
 \[ \nprove \htriple{R(z+1)}{\stmt_1}{E_1}{R(z)} \]
 where $E_1 = \eet{\stmt_1}{\ctert{0}}$.
 Now, let
 \[ E' := \eet{\WHILEDO{\pguard}{\stmt_1}}{\ctert{0}} - \eet{\stmt_1}{\ctert{0}} \]
 and 
 \[ E := \eet{\WHILEDO{\pguard}{\stmt_1}}{\ctert{0}}. \]
 
 Our goal is to apply the $\WHILE$-rule in order to show 
 \[ \nprove \htriple{\exists z \mydot R(z)}{\WHILEDO{\pguard}{\stmt_1}}{E}{R(0)}. \]
 We first check the side conditions of this rule for our choice of $E'$, i.e. we show that
 $R(0) \Rightarrow \neg \pguard \wedge E \geq 1$ as well as $R(z+1) \Rightarrow \pguard \wedge E \geq E_1 + E'$ are valid.
 
 If $R(0)$ is valid then $\neg \pguard$ is valid due to $(\ast)$ and
 \[ \eval{E}(\ps) = \eet{\WHILEDO{\pguard}{\stmt_1}}{\ctert{0}}(\ps) = 1 \geq 1. \]
 Furthermore, if $R(z+1)$ is valid for some $z \in \Nats$, then $\pguard$ is valid by $(\ast)$ and for each $\ps \in \ProgramStates$, we have
 \begin{align*}
    & \eval{E}(\ps) ~=~ \eet{\WHILEDO{\pguard}{\stmt_1}}{\ctert{0}}(\ps) \\
    & ~=~ \eet{\WHILEDO{\pguard}{\stmt_1}}{\ctert{0}}(\ps) - \eet{\stmt_1}{\ctert{0}}(\ps) + \eet{\stmt_1}{\ctert{0}}(\ps) \\
    & ~=~ \eval{E'}(\ps) + \eval{E_1}(\ps).
 \end{align*}
 
 Hence, the side conditions of the $\WHILE$-rule hold. In order to apply the rule, we also have to show 
 validity of the triple
 \[ \htriple{R(z+1) + E' = u}{\stmt_1}{\eet{\stmt_1}{\ctert{0}}}{R(z) \wedge E \leq u} \]
 where $u$ is a fresh logical variable.
 Since $u$ does not occur in $P$, we know that for each $\ps \in \ProgramStates$ with $\ps \models P$,
 $\ps \models R(z+1) \wedge E' = u$.
 Then 
 \begin{align*}  
  & \ps \models R(z+1) \wedge E' = u \text{ and } \opSem{\stmt_1}{\ps} \models R(z) \wedge E(\opSem{\stmt_1}{\ps}) \leq u \\
  ~\Leftrightarrow~ & \ps \models R(z+1) \wedge E' = u \text{ and }
                      \opSem{\stmt_1}{\ps} \models R(z) \text{ and} \\
  &  E(\opSem{\stmt_1}{\ps}) \leq E'(\ps) \tag{Def. of $u$} \\
  ~\Leftrightarrow~ & \ps \models R(z+1) \wedge E' = u \text{ and }
                      \opSem{\stmt_1}{\ps} \models R(z) \text{ and} \\
  &  E(\opSem{\stmt_1}{\ps}) \leq \eet{\WHILEDO{\pguard}{\stmt_1}}{\ctert{0}}(\ps) -\eet{\stmt_1}{\ctert{0}}(\ps) \tag{Def. of $E'$} \\
  ~\Leftrightarrow~ & \ps \models R(z+1) \wedge E' = u \text{ and }
                      \opSem{\stmt_1}{\ps} \models R(z) \text{ and} \\
  &  E(\opSem{\stmt_1}{\ps}) + \eet{\stmt_1}{\ctert{0}}(\ps) \leq \eet{\WHILEDO{\pguard}{\stmt_1}}{\ctert{0}}(\ps) \\
  &  \text{(\autoref{thm:det-eet-loop})} \\
  ~\Rightarrow~ & \models_{E} \htriple{R(z+1) \wedge E' = u}{\stmt_1}{\eet{\stmt_1}{\ctert{0}}}{R(z) \wedge E \leq u}.
 \end{align*}
 Thus we may apply the $\WHILE$-rule to conclude 
 \begin{align*}
   \nprove \htriple{\exists z \mydot R(z)}{\WHILEDO{\pguard}{\stmt_1}}{\eet{\WHILEDO{\pguard}{\stmt_1}}{\ctert{0}}}{R(0)}.
 \end{align*}
 By assumption, the implications $P \Rightarrow \exists z \mydot R(z)$ as well as $R(0) \Rightarrow Q$ are valid, i.e.
 \begin{align*}
   \nprove \htriple{P}{\WHILEDO{\pguard}{\stmt_1}}{\eet{\WHILEDO{\pguard}{\stmt_1}}{\ctert{0}}}{Q}.
 \end{align*}
 \qed
\end{proof}
\begingroup
\def\thetheorem{\ref{thm:det-eet-completeness}}
\begin{theorem}[Completeness of \eetsymbol{} w.r.t. Nielson]
    For all $\stmt \in \DetStmt$, assertions $P,Q$ and deterministic expressions $E$,
  $\nprove \htriple{P}{\stmt}{E}{Q}$ implies that there exists a natural number $k$ such that for all $\ps \in \ProgramStates$ satisfying $P$, we have 
  \[\eet{\stmt}{\ctert{0}}(\ps) ~\leq~ k \cdot (\eval{E}(\ps))~. \]
\end{theorem}
\addtocounter{theorem}{-1}
\endgroup
\begin{proof}
  We show the following claim by induction on the structure of program statements:
  For all assertions $P,Q$, $\stmt \in \Stmt_{\det}$ and arithmetic expressions $E$, 
  \[ \nprove \htriple{P}{\stmt}{E}{Q} \]
  implies that there exists $k \in \Nats$ such that for all $\ps \in \ProgramStates$, we have 
  \[ \eet{\stmt}{\ctert{0}}(\ps) \leq k \cdot \eval{E}(\ps). \]
  Clearly this implies \autoref{thm:det-eet-completeness}.
  
  \paragraph{The effectless program} $\stmt = \SKIP$.
  Assume 
  \[ \nprove \htriple{P}{\SKIP}{E}{Q}. \] 
  Then there exists an assertion $R$ such that
  \[ \nprove \htriple{R}{\SKIP}{\one}{R} \]
  and 
  $P \Rightarrow P' \wedge 1 \leq k \cdot E$ are $R \Rightarrow Q$ are valid for some $k \in \Nats$.
  Hence there exists a $k \in \Nats$ such that $\eet{\SKIP}{\ctert{0}}(\ps) = 1 \leq k \cdot \eval{E_1}(\ps)$ for each $\ps \in \ProgramStates$.
  
  \paragraph{The assignment} $\stmt = \ASSIGN{x}{E}$.
  Analogous to $\SKIP$.
  
  \paragraph{The sequential composition} $\stmt = \stmt_1;\stmt_2$.
  Assume 
  \[ \nprove \htriple{P}{\stmt_1;\stmt_2}{E}{Q}. \] 
  Then there exists an assertion $R$ and arithmetic expressions $E_1,E_2,E_2'$ such that 
  \[  \nprove \htriple{P \wedge E_2' = u}{\stmt_1}{E_1}{R \wedge E_2 \leq u} \] 
  and 
  \[  \nprove \htriple{R}{\stmt_2}{E_2}{Q} \] 
  for some fresh logical variable $u$.
  By I.H. there exists a $k \in \Nats$ such that for all $\ps \in \ProgramStates$,
  \begin{align*}
   \eet{\stmt_1}{\ctert{0}}(\ps) \leq k \cdot \eval{E_1}(\ps) \text{ and } \eet{\stmt_2}{\ctert{0}}(\ps) \leq k \cdot \eval{E_2}(\ps). \tag{$\ast$}
  \end{align*}
  In particular, 
  \begin{align*}
     \eet{\stmt_2}{\ctert{0}}(\opSem{\stmt_1}{\ps}) \leq  k \cdot \eval{E_2}(\opSem{\stmt_1}{\ps}) \leq k \cdot \eval{E_2'}(\ps).  \tag{$\dag$}
  \end{align*}
  Hence
  \begin{align*}
    & \eet{\stmt_1;\stmt_2}{\ctert{0}}(\ps) \\
    & ~=~ \eet{\stmt_1}{\ctert{0}}(\ps) + \eet{\stmt_2}{\ctert{0}}(\opSem{\stmt_1}{\ps}) \tag{\autoref{thm:det-eet-sum}}\\
    & ~\leq~ k \cdot \eval{E_1}(\ps) + \eet{\stmt_2}{\ctert{0}}(\opSem{\stmt_1}{\ps}) \tag{by $\ast$} \\
    & ~\leq~ k \cdot \eval{E_1}(\ps) + k \cdot \eval{E_2'}(\ps) \tag{by $\dag$} \\ 
    & ~=~ k \cdot (\eval{E_1}(\ps) + \eval{E_2'}(\ps)).
  \end{align*}

  \paragraph{Conditionals} $\stmt' = \ITE{\pguard}{\stmt_1}{\stmt_2}$
  Assume 
  \[ \nprove \htriple{P}{\ITE{\pguard}{\stmt_1}{\stmt_2}}{E}{Q}. \] 
  Then
  \[
    \nprove \htriple{P \wedge \pguard}{\stmt_1}{E}{Q} \text{ and } \nprove \htriple{P \wedge \neg \pguard}{\stmt_2}{E}{Q}.
  \]
  By I.H. there exists a $k \in \Nats$ such that for each $\ps \in \ProgramStates$ and $i \in \{1,2\}$
  \[ \eet{\stmt_i}{\ctert{0}}(\ps) \leq k \cdot \eval{E_1}(\ps). \tag{$\ast$} \]
  Thus
  \begin{align*}
    & \eet{\ITE{\pguard}{\stmt_1}{\stmt_2}}{\zero}(\ps) \\
    & ~=~ \one + \ExpToFun{\pguard}(\ps) \cdot \eet{\stmt_1}{\ctert{0}}(\ps) + \ExpToFun{\neg\pguard}(\ps) \cdot \eet{\stmt_2}{\ctert{0}}(\ps) \tag{\autoref{table:eet-rules}} \\
    & ~\leq~ k + \ExpToFun{\pguard}(\ps) \cdot k \cdot \eval{E_1}(\ps) + \ExpToFun{\neg\pguard}(\ps) \cdot k \cdot \eval{E_1}(\ps) \tag{by $\ast$} \\
    & ~\leq~ (3 \cdot k) \cdot \eval{E_1}(\ps).
  \end{align*}

  \paragraph{Loops} $\stmt' = \WHILEDO{\pguard}{\stmt_1}$. 
  Assume 
  \[ \nprove \htriple{P}{\WHILEDO{\pguard}{\stmt_1}}{E}{Q}. \]
  Then there exists an assertion $R(z)$ such that $P \Rightarrow \exists z \mydot R(z)$ 
  and $R(0) \Rightarrow Q$ are valid.
  Furthermore, there exists $z \in \Nats$ such that
  \begin{align*}
   \nprove \htriple{R(z+1) \wedge E' = u}{\stmt_1}{E_1}{R(z) \wedge E \leq u} \tag{$\ast$}  
  \end{align*}
  for some fresh logical variable $u$.
  Additionally, for each $z \in \Nats$, the side conditions 
  \begin{align*}
  R(z+1) \Rightarrow \pguard \wedge E \geq E_1 + E' \text{ as well as } R(0) \Rightarrow \neg \pguard \wedge E \geq 1 \tag{$\dag$}
  \end{align*}
  are valid.
  Our goal is to show that there exists $k \in \Nats$ such that for all $\ps \in \ProgramStates$,
  \begin{align*} 
    \eet{\WHILEDO{\pguard}{\stmt_1}}{\ctert{0}}(\ps) \leq k \cdot \eval{E}(\ps).  
    \tag{$\spadesuit$}  
  \end{align*}
  By I.H. there exists a $k' \in \Nats$ such that for each $\ps \in \ProgramStates$, 
  \begin{align*} 
    \eet{\stmt_1}{\ctert{0}}(\ps) \leq k' \cdot \eval{E_1}(\ps). 
    \tag{$\clubsuit$}  
  \end{align*}

  We now show by complete induction over $z \in \Nats$ that for all $\ps \in \ProgramStates$ with
  $\ps \models R(z)$ and 
  \begin{align*}
     \nprove \htriple{R(z)}{\WHILEDO{\pguard}{\stmt_1}}{E}{R(0)}, 
  \end{align*}
  we have $(k'+1) \cdot \eval{E}(\ps) \geq \eet{\WHILEDO{\pguard}{\stmt_1}}{\ctert{0}}(\ps)$.
  
  For the base case $z=0$ the side condition $R(0) \Rightarrow \neg \pguard \wedge E \geq 1$
  yields 
  \begin{align*}
    & (k'+1) \cdot  \eval{E}(\ps) \\
    & ~\geq~ \one \tag{by $\dag$} \\
    & ~=~ \one + \ExpToFun{\pguard}(\ps) \cdot \eet{\stmt_1}{\ctert{0}}(\ps) + \ExpToFun{\neg\pguard}(\ps) \cdot \zero \\
    & ~=~ \eet{\WHILEDO{\pguard}{\stmt_1}}{\ctert{0}}(\ps). \tag{\autoref{table:eet-rules}}
  \end{align*}
  
  Now let $\ps \models R(z+1)$.
  Then the side condition $R(z+1) \Rightarrow \pguard \wedge E \geq E_1 + E'$ yields 
  \begin{align*}
   & (k'+1) \cdot \eval{E}(\ps) \\
   & ~\geq~ (k'+1) \cdot (\eval{E_1}(\ps) + \eval{E'}(\ps)) \tag{by $\dag$} \\
   & ~\geq~ (k'+1) \cdot \eval{E_1}(\ps) + (k'+1) \cdot \eval{E}(\opSem{\stmt_1}{\ps}) \tag{Postcondition of ($\ast$)} \\
   & ~\geq~ (k'+1) \cdot \eval{E_1}(\ps) + \eet{\WHILEDO{\pguard}{\stmt_1}}{\ctert{0}}(\opSem{\stmt_1}{\ps}) \\
   & \qquad \text{(I.H., $\opSem{\stmt_1}{\ps} \models R(z)$)} \\
   & ~\geq~ \eval{E_1}(\ps) + \eet{\WHILEDO{\pguard}{\stmt_1}}{\ctert{0}}(\opSem{\stmt_1}{\ps}) \tag{by $\clubsuit$} \\
   & ~=~ \eet{\WHILEDO{\pguard}{\stmt_1}}{\ctert{0}}(\ps) \tag{\autoref{thm:det-eet-loop}}
  \end{align*}
  which completes the inner induction and implies $(\spadesuit)$ for $k = k'+1$.
  \qed
\end{proof}

\section{Omitted Calculations}
\subsection{Invariant Verification for the Random Walk}
\label{sec:app-random-walk}

First we verify that 
\begin{align*}
	I_n ~=~ \one + \sum_{k=0}^{n} \,\eval{x \,{>}\, k} \cdot a_{n,k}
\end{align*}
is a lower
$\omega$--invariant of the loop with respect to $\ctert{0}$ if for all $n \geq 0$,
\begin{align}
 a_{0,0} &\:=\: 1 \label{eq:rwalk-rec1}\\[0.5ex]
 a_{n+1,0} &\:=\: 2 + \tfrac{1}{2} \cdot (a_{n,0} + a_{n,1})\label{eq:rwalk-rec2}\\[0.5ex]
a_{n+1,k} &\:=\: \tfrac{1}{2} \cdot (a_{n,k-1} + a_{n,k+1}), \text{ for all } 1 \,{\leq}\,
  k \,{\leq}\, n{+}1\label{eq:rwalk-rec3}\\[0.5ex]
a_{n,k} &\:=\: 0, \text{ for all } k > n~. \label{eq:rwalk-rec4}
\end{align}
Let $F$ be the characteristic functional of the loop with respect to
\ctert{0}. Then for the first condition $F(\zero) \succeq I_0$ we have
\begin{align*}
F(\zero) 	~=~
  &\one + \eval{x \leq 0} \cdot \zero + \eval{x > 0} \cdot \eet{C}{\zero}\\
~=~	&\one + \eval{x > 0} \cdot \Bigl(\one + \tfrac{1}{2} \cdot
          \zero\subst{x}{x-1} + \tfrac{1}{2} \cdot
          \zero\subst{x}{x+1}\Bigr) \\
~=~	&\one + \eval{x > 0} \cdot \one\\
~=~	&\one + \eval{x > 0} \cdot a_{0,0} = I_0 \tag{Eq.~\ref{eq:rwalk-rec1}}
\end{align*}
For the second condition $F(I_n) \succeq I_{n+1}$, consider
\begin{align*}
F(I_n) ~=~%
				&\one + \eval{x \leq 0} \cdot \zero + \eval{x >
                                  0} \cdot \eet{C}{I_n}\\
  ~=~	&\one + \eval{x > 0} \cdot \Bigl(\one + \tfrac{1}{2} \cdot
          I_n\subst{x}{x-1} + \tfrac{1}{2} \cdot I_n\subst{x}{x+1}\Bigr) \\
  ~=~	&\one + \eval{x > 0} \cdot \left(\ctert{2} + \tfrac{1}{2} \cdot
          \sum_{k=0}^{n} \,\eval{x{-}1 \,{>}\, k} \cdot a_{n,k} + \tfrac{1}{2} \cdot \sum_{k=0}^{n} \,\eval{x{+}1 \,{>}\, k} \cdot a_{n,k}\right) \\
  ~=~	&\one + \eval{x > 0} \cdot \left(\ctert{2} + \tfrac{1}{2} \cdot
          \sum_{k=1}^{n+1} \,\eval{x \,{>}\, k} \cdot a_{n,k-1} + \tfrac{1}{2} \cdot \sum_{k=-1}^{n-1} \,\eval{x\,{>}\, k} \cdot a_{n,k+1}\right) \\
  ~=~	&\one + \eval{x > 0} \cdot \biggl(\ctert{2} + \tfrac{1}{2} \cdot \eval{x
          \,{>}\, {-}1} \cdot a_{n,0} + \tfrac{1}{2} \cdot \eval{x
          \,{>}\, 0} \cdot a_{n,1} \biggr.\\
  & \biggl. \qquad \qquad \qquad \quad\  +\: \tfrac{1}{2} \cdot \sum_{k=1}^{n-1} \,\eval{x \,{>}\, k} \cdot
  (a_{n,k-1} + a_{n,k+1}) \biggr.\\
  & \biggl. \qquad \qquad \qquad \quad\  +\: \tfrac{1}{2} \cdot \eval{x
          \,{>}\, n} \cdot a_{n,n-1} + \tfrac{1}{2} \cdot \eval{x
          \,{>}\, n+1} \cdot a_{n,n}  \biggr)%
\intertext{We distribute $\eval{x \,{>}\, 0}$ and use the fact that $\eval{x
          \,{>}\, 0} \cdot \eval{x \,{>}\, a} = \eval{x \,{>}\, \max \{0,a\}}$
    to obtain}
  ~=~	&\one +  \eval{x \,{>}\, 0} \,{\cdot}\,  \bigl(
          \ctert{2} + \tfrac{1}{2} \,{\cdot}\, ( a_{n,0} \,{+}\, a_{n,1} )\bigr) + \tfrac{1}{2} \,{\cdot} \sum_{k=1}^{n-1} \,\eval{x \,{>}\, k} \cdot
  (a_{n,k-1} \,{+}\, a_{n,k+1}) \\
  & \qquad \qquad \quad  +\: \tfrac{1}{2} \cdot \eval{x
          \,{>}\, n} \cdot a_{n,n-1} + \tfrac{1}{2} \cdot \eval{x
          \,{>}\, n+1} \cdot a_{n,n} \\
  ~=~	&\one +  \eval{x \,{>}\, 0} \,{\cdot}\,  a_{n+1,0} + \tfrac{1}{2} \,{\cdot} \sum_{k=1}^{n-1} \,\eval{x \,{>}\, k} \cdot
  (a_{n,k-1} \,{+}\, a_{n,k+1}) \tag{Eq.~\ref{eq:rwalk-rec2}}\\
  & \qquad \qquad \quad  +\: \tfrac{1}{2} \cdot \eval{x
          \,{>}\, n} \cdot a_{n,n-1} + \tfrac{1}{2} \cdot \eval{x
          \,{>}\, n+1} \cdot a_{n,n}  \\
  ~=~	&\one +  \eval{x \,{>}\, 0} \,{\cdot}\,  a_{n+1,0} + \tfrac{1}{2} \,{\cdot} \sum_{k=1}^{n-1} \,\eval{x \,{>}\, k} \cdot
  (a_{n,k-1} \,{+}\, a_{n,k+1}) \tag{Eq.~\ref{eq:rwalk-rec4}} \\
  & \qquad\qquad\quad  +\: \tfrac{1}{2} \cdot \eval{x
          \,{>}\, n} \cdot (a_{n,n-1} + a_{n,n+1}) + \tfrac{1}{2} \cdot \eval{x
          \,{>}\, n+1} \cdot (a_{n,n} + a_{n,n+2}) \\
  ~=~	&\one +  \eval{x \,{>}\, 0} \,{\cdot}\,  a_{n+1,0} + \tfrac{1}{2} \,{\cdot} \sum_{k=1}^{n+1} \,\eval{x \,{>}\, k} \cdot
  (a_{n,k-1} \,{+}\, a_{n,k+1}) \\
  ~=~	&\one +  \eval{x \,{>}\, 0} \,{\cdot}\,  a_{n+1,0} + \sum_{k=1}^{n+1} \,\eval{x \,{>}\, k} \cdot
  a_{n+1,k} \tag{Eq.~\ref{eq:rwalk-rec3}}\\
  ~=~	&\one +  \sum_{k=0}^{n+1} \,\eval{x \,{>}\, k} \cdot
  a_{n+1,k} ~=~ I_{n+1}
\end{align*}
Now we show that 
\[
a_{n,k} = \frac{1}{2^{n}} \left[- \binom{n}{\lfloor\!\frac{n-k}{2}\!\rfloor} + 2 \,
\sum_{i=0}^{n-k}
2^i \binom{n-i}{\lfloor\!\frac{n-i-k}{2}\!\rfloor} \right]
\]
satisfies the above recursion. Here, we assume that $\binom{n}{m}$ is $0$
whenever $m < 0$. Equations~\ref{eq:rwalk-rec1} and \ref{eq:rwalk-rec4} are
immediate. 
For \autoref{eq:rwalk-rec2}, i.e.\ $a_{n+1,0} = 2 + \nicefrac{1}{2} \cdot (a_{n,0} + a_{n,1})$, we make use of the identity
\begin{align*}
  {k \choose \lfloor \frac{k+1}{2} \rfloor} ~ = ~ {k \choose \lfloor \frac{k}{2} \rfloor} \tag{$\bigstar$}
\end{align*}
which is shown by a simple case analysis on $k$ being even or odd. Then
\begin{align*}
 a_{n+1,0} ~=~ & \frac{1}{2^{n+1}} \left[ - {n+1 \choose \lfloor \frac{n+1}{2} \rfloor} + 2 \sum_{i=0}^{n+1} 2^i {n+1-i \choose \lfloor \frac{n+1-i}{2} \rfloor} \right] \tag{Def. $a_{n+1,0}$} \\
           ~=~ & \frac{1}{2^{n+1}} \Bigg[
                     - {n \choose \lfloor \frac{n+1}{2} \rfloor} - {n \choose \lfloor \frac{n+1}{2} \rfloor -1} \tag*{$\left(\text{\small ${{k+1}\choose{\ell+1}} = {{k}\choose{\ell}} + {{k}\choose{\ell+1}}$}\right)$} \\
               & \qquad \qquad + 2 \sum_{i=0}^{n+1} 2^i \left( {n-i \choose \lfloor \frac{n+1-i}{2} \rfloor} + {n-i \choose \lfloor \frac{n+1-i}{2} \rfloor - 1} \right)
                 \Bigg] \\
	~=~ & \frac{1}{2^{n+1}} \Bigg[
                     - {n \choose \lfloor \frac{n+1}{2} \rfloor} - {n \choose \lfloor \frac{n+1}{2} \rfloor -1} \\
               & \qquad \qquad + 2^{n+2} + 2 \sum_{i=0}^{n} 2^i \left( {n-i \choose \lfloor \frac{n+1-i}{2} \rfloor} + {n-i \choose \lfloor \frac{n+1-i}{2} \rfloor - 1} \right)
                 \Bigg] \\
           ~=~ & 2 + \frac{1}{2} \cdot \Bigg( \frac{1}{2^n} \left[ - {n \choose \lfloor \frac{n-1}{2} \rfloor} + 2 \sum_{i=0}^{n-1} 2^i {n-i \choose \lfloor \frac{n-i-1}{2} \rfloor} \right] \\
               & \qquad \qquad + \frac{1}{2^n} \left[ - {n \choose \lfloor \frac{n+1}{2} \rfloor} + 2 \sum_{i=0}^{n} 2^i {n-i \choose \lfloor \frac{n+1-i}{2} \rfloor} \right] \Bigg) \tag*{$\left(\text{by {\small ${{0}\choose{-1}} = 0$} and {\small $\lfloor \frac{n+1}{2} \rfloor  - 1 = \lfloor \frac{n-1}{2} \rfloor$}}\right)$}  \\
           ~=~ & 2 + \frac{1}{2} \left( a_{n,1} + \frac{1}{2^n} \left[ - {n \choose \lfloor \frac{n+1}{2} \rfloor} + 2 \sum_{i=0}^{n} 2^i {n-i \choose \lfloor \frac{n+1-i}{2} \rfloor} \right] \right) \tag{Def. $a_{n,1}$} \\
           ~=~ & 2 + \frac{1}{2} \left( a_{n,1} + \frac{1}{2^n} \left[ - {n \choose \lfloor \frac{n}{2} \rfloor} + 2 \sum_{i=0}^{n} 2^i {n-i \choose \lfloor \frac{n-i}{2} \rfloor} \right] \right) \tag{using $\bigstar$} \\
           ~=~ & 2 + \frac{1}{2} \left( a_{n,1} + a_{n,0} \right) \tag{Def. $a_{n,0}$}
\end{align*}

It remains to show \autoref{eq:rwalk-rec3}, i.e.\ $a_{n+1,k} = \tfrac{1}{2} \cdot (a_{n,k-1} + a_{n,k+1}), \text{ for all } 1 \leq  k \leq n+1$:
\begin{align*}
 a_{n+1,k} ~=~ & \frac{1}{2^{n+1}} \left[ - {n+1 \choose \lfloor \frac{n+1-k}{2} \rfloor} + 2 \sum_{i=0}^{n+1-k} 2^i {n+1-i \choose \lfloor \frac{n+1-i-k}{2} \rfloor} \right] \tag{Def. $a_{n+1,k}$} \\
           ~=~ & \frac{1}{2^{n+1}} \left[ - {n+1 \choose \lfloor \frac{n+1-k}{2} \rfloor} + 2^{n+2-k} + 2 \sum_{i=0}^{n-k} 2^i {n+1-i \choose \lfloor \frac{n+1-i-k}{2} \rfloor} \right] \tag*{$\left(\text{\small ${{m}\choose{0}} = 1$}\right)$}   \\
           ~=~ & \frac{1}{2^{n+1}} \Bigg[ - {n \choose \lfloor \frac{n-(k-1)}{2} \rfloor} - {n \choose \lfloor \frac{n-(k+1)}{2} \rfloor} + 2^{n+2-k} \tag*{$\left(\text{\small ${{k+1}\choose{\ell+1}} = {{k}\choose{\ell}} + {{k}\choose{\ell+1}}$}\right)$} \\
               & \qquad \qquad + 2 \sum_{i=0}^{n-k} 2^i {n-i \choose \lfloor \frac{n-i-(k-1)}{2} \rfloor} + 2 \sum_{i=0}^{n-k} 2^i {n-i \choose \lfloor \frac{n-i-(k+1)}{2} \rfloor} \Bigg] \\
           ~=~ & \frac{1}{2^{n+1}} \Bigg[ - {n \choose \lfloor \frac{n-(k-1)}{2} \rfloor} - {n \choose \lfloor \frac{n-(k+1)}{2} \rfloor} + 2^{n+2-k} \tag*{$\left(\text{\small ${{m}\choose{-1}} = 0$}\right)$} \\
               & \qquad \qquad + 2 \sum_{i=0}^{n-k} 2^i {n-i \choose \lfloor \frac{n-i-(k-1)}{2} \rfloor} + 2 \sum_{i=0}^{n-(k+1)} 2^i {n-i \choose \lfloor \frac{n-i-(k+1)}{2} \rfloor} \Bigg] \\
           ~=~ & \frac{1}{2} \left( a_{n,k+1} \tag{Def. $a_{n,k+1}$} + \frac{1}{2^n} \left[ - {n \choose \lfloor \frac{n-(k-1)}{2} \rfloor} + 2 \sum_{i=0}^{n-(k-1)} 2^i {n-i \choose \lfloor \frac{n-i-(k-1)}{2} \rfloor} \right] \right) \\
           ~=~ & \frac{1}{2} \left( a_{n,k+1} + a_{n,k-1} \right) \tag{Def. $a_{n,k-1}$}
\end{align*}

Finally we prove that $\lim_{n \rightarrow \infty} a_{n,0} = \infty$. The crux of
the proof is showing that for all $n \geq 2$, $a_{n,0} \geq 1 + \mathcal{H}_{\lfloor\nicefrac n
  2\rfloor}$ where $\mathcal{H}_m$ denotes the $m$-th
Harmonic number, i.e.\
\begin{align*}
	\mathcal{H}_m ~=~ \sum_{k=1}^m \tfrac{1}{k}~.
\end{align*}
The result then follows since $\lim_{m \rightarrow \infty}
\mathcal{H}_m = \infty$. Calculations go as follows:
\begin{align*}
a_{n,0}
&~=~ \frac{1}{2^{n}} \left[- \binom{n}{\lfloor\!\frac{n}{2}\!\rfloor} + 2 \,
\sum_{i=0}^{n}
2^i \binom{n-i}{\lfloor\!\frac{n-i}{2}\!\rfloor} \right] \\
&~=~ \frac{1}{2^{n}} \left[- \binom{n}{\lfloor\!\frac{n}{2}\!\rfloor} + 2 \,
\sum_{k=0}^{n}
2^{n-k} \binom{k}{\lfloor\!\frac{k}{2}\!\rfloor} \right]
  \tag{Take $k=n-i$}\\
&~\geq~ \frac{1}{2^{n}} \left[- 2^n + 2 \,
\sum_{k=0}^{n}
2^{n-k} \binom{k}{\lfloor\!\frac{k}{2}\!\rfloor} \right] \tag*{$\left(\text{${{n}\choose{\lfloor\!\nicefrac{n}{2}\!\rfloor}} \leq 2^n$}\right)$}  \\
&~=~ -1 + 2 \sum_{k=0}^{n}
2^{-k} \binom{k}{\lfloor\!\frac{k}{2}\!\rfloor}\\
&~\geq~ -1 + 2 \sum_{j=0}^{\lfloor\nicefrac n 2\rfloor}
2^{-2j} \binom{2j}{j} \tag{Keep only even $k$'s}\\
&~\geq~ 1 + 2 \sum_{j=1}^{\lfloor\nicefrac n 2\rfloor}
2^{-2j} \binom{2j}{j} \tag{Extract $j=0$ out of the sum}\\
&~\geq~ 1 + 2 \sum_{j=1}^{\lfloor\nicefrac n 2\rfloor}
2^{-2j} \frac{2^{2j-1}}{\sqrt{j}} \tag*{$\left(\text{Stirling approximation: $\binom{2j}{j} \geq \frac{2^{2j-1}}{\sqrt{j}}$}\right)$} \\
&~=~ 1 + \sum_{j=1}^{\lfloor\nicefrac n 2\rfloor}
\frac{1}{\sqrt{j}}\\
&~\geq~ 1 + \sum_{j=1}^{\lfloor\nicefrac n 2\rfloor}
\frac{1}{j} = 1 + \mathcal{H}_{\lfloor\nicefrac n 2\rfloor}~. \tag{$\sqrt{j} \leq j$}
\end{align*}
\subsection{Invariant Verification for the Inner Loop of the Coupon Collector Algorithm}
\label{sec:app-coupon-collector-inner}

By definition of $\eetsymbol$ (cf. \autoref{table:eet-rules}), $\eet{C_\textnormal{in}}{f} = \lfp F_{f}(X)$, where
the characteristic functional $F_{f}(X)$ is given by:
\begin{align*}
 & F_{f}(X) \\
 & ~=~ \one + \cond{\cp{i} = 0} \cdot f + \cond{\cp{i} \neq 0} \cdot \eet{\APPASSIGN{i}{\mathtt{Unif}[1 \ldots N]}}{X} \\
 & ~=~ \one + \cond{\cp{i} = 0} \cdot f + \cond{\cp{i} \neq 0} \cdot \left( \one + \frac{1}{N} \cdot \sum_{k=1}^{N} X[i/k] \right) \\
\end{align*}

For simplicity, let
\begin{align*}
 G(f) = \sum_{j=1}^{N} \cond{\cp{j}=0}\cdot f[i/j].
\end{align*}

Moreover, recall the $\omega$--invariant of the inner loop proposed in \autoref{subsec:coupon}:

\begin{align*}
J_\textnormal{n}^f 	~=~ 	\ctert{1} ~+~ & \cond{\cp{i} = 0} \cdot f \\
        ~+~ & \cond{\cp{i} \neq 0} \cdot \sum_{k = 0}^{n}\left(\frac{\nocol}{N}\right)^k \left(2 + \frac{1}{N} \cdot \sum_{j=1}^{N} \eval{\cp{j}=0} \cdot f[i/j]\right) \\
        ~=~ 	\ctert{1} ~+~ & \cond{\cp{i} = 0} \cdot f \\
        ~+~ & \cond{\cp{i} \neq 0} \cdot \sum_{k = 0}^{n}\left(\frac{\nocol}{N}\right)^k \left(2 + \frac{G(f)}{N} \right) \\
\end{align*}

Our goal is to to apply \autoref{thm:omega-inv} to show that $J_{n}^{f}$ is a lower as well as an upper $\omega$--invariant.
We first show $F_{f}(\zero) = J_\textnormal{0}^{f}$:

\begin{align*}
 F_{f}(\zero) & ~=~ \one + \cond{\cp{i} = 0} \cdot f + \cond{\cp{i} \neq 0} \cdot \left( \one + \frac{1}{N} \cdot \sum_{k=1}^{N} \zero[i/k] \right) \\
          & ~=~ \one + \cond{\cp{i} = 0} \cdot f + \cond{\cp{i} \neq 0} \\
          & ~=~ J_\textnormal{0}^{f}
\end{align*}

Furthermore, we have to prove $F_{f}(J_\textnormal{n}^{f}) = J_\textnormal{n+1}^{f}$:

\begin{align*}
 & F_{f}(J_\textnormal{n}^{f}) \\
 & ~=~ \one + \cond{\cp{i} = 0} \cdot f + \cond{\cp{i} \neq 0} \cdot \left( \one + \frac{1}{N} \cdot \sum_{k=1}^{N} {J_\textnormal{n}^{f}}[i/k] \right) \tag{Def. $F_{f}$} \\
 & ~=~ \one + \cond{\cp{i} = 0} \cdot f 
            + \cond{\cp{i} \neq 0} \cdot ( \one + \frac{1}{N} \cdot \sum_{k=1}^{N} ( 
            \ctert{1} + \cond{\cp{k} = 0} \cdot f[i/k] \\
 &   \qquad + \cond{\cp{k} \neq 0} \cdot \sum_{\ell = 0}^{n}\left(\frac{\nocol}{N}\right)^{\ell} \cdot \left(2 + \frac{G(f)[i/k]}{N} \right)
            ) ) \tag{Def. $J_\textnormal{n}^f$} \\
 & ~=~ \one + \cond{\cp{i} = 0} \cdot f 
            + \cond{\cp{i} \neq 0} \cdot ( \one + \frac{1}{N} \cdot \sum_{k=1}^{N} ( 
            \ctert{1} + \cond{\cp{k} = 0} \cdot f[i/k] \\
 &   \qquad + \cond{\cp{k} \neq 0} \cdot \sum_{\ell = 0}^{n}\left(\frac{\nocol}{N}\right)^{\ell} \cdot \left(2 + \frac{G(f)}{N} \right)
            ) ) \tag{$k$ does not occur in $G(f)$} \\
 & ~=~ \one + \cond{\cp{i} = 0} \cdot f + \ctert{2} \cdot \cond{\cp{i} \neq 0} + \frac{\cond{\cp{i} \neq 0}}{N} \cdot \sum_{k=1}^{N} ( 
            \cond{\cp{k} = 0} \cdot f[i/k] \\
 &   \qquad + \cond{\cp{k} \neq 0} \cdot \sum_{\ell = 0}^{n}\left(\frac{\nocol}{N}\right)^{\ell} \cdot \left(2 + \frac{G(f)}{N} \right)
            ) \\
 & ~=~ \one + \cond{\cp{i} = 0} \cdot f + \ctert{2} \cdot \cond{\cp{i} \neq 0} \\
 &   \qquad + \frac{\cond{\cp{i} \neq 0}}{N} \cdot \sum_{k=1}^{N} ( \cond{\cp{k} = 0} \cdot f[i/k]) \\
 &   \qquad + \frac{\cond{\cp{i} \neq 0}}{N} \cdot \sum_{k=1}^{N} \cond{\cp{k} \neq 0} \cdot \sum_{\ell = 0}^{n}\left(\frac{\nocol}{N}\right)^{\ell} \cdot \left(2 + \frac {G(f)} N\right) \\
 & ~=~ \one + \cond{\cp{i} = 0} \cdot f + \ctert{2} \cdot \cond{\cp{i} \neq 0} \\
 &   \qquad + \frac{\cond{\cp{i} \neq 0}}{N} \cdot \sum_{k=1}^{N} ( \cond{\cp{k} = 0} \cdot f[i/k]) \\
 &   \qquad + \frac{\cond{\cp{i} \neq 0} \nocol}{N} \cdot \sum_{\ell = 0}^{n}\left(\frac{\nocol}{N}\right)^{\ell} \cdot \left(2 + \frac {G(f)} N\right) \tag{Def. $\nocol$} \\
 & ~=~ \one + \cond{\cp{i} = 0} \cdot f + \cond{\cp{i} \neq 0} \left( 2 + \frac{G(f)}{N} \right) \\
 &   \qquad + \frac{\cond{\cp{i} \neq 0} \nocol}{N} \cdot \sum_{\ell = 0}^{n}\left(\frac{\nocol}{N}\right)^{\ell} \cdot \left(2 + \frac {G(f)} N\right) \tag{Def. $G$} \\
 & ~=~ \one + \cond{\cp{i} = 0} \cdot f + \cond{\cp{i} \neq 0} \left( 2 + \frac{G(f)}{N} \right) \\
 &   \qquad + \cond{\cp{i} \neq 0} \cdot \sum_{\ell = 1}^{n+1}\left(\frac{\nocol}{N}\right)^{\ell} \cdot \left(2 + \frac {G(f)} N\right) \\
 & ~=~ \one + \cond{\cp{i} = 0} \cdot f + \cond{\cp{i} \neq 0} \cdot \sum_{\ell = 0}^{n+1}\left(\frac{\nocol}{N}\right)^{\ell} \cdot \left(2 + \frac {G(f)} N\right) \\
 & ~=~ J_\textnormal{n+1}^{f}
\end{align*}

Now, by \autoref{thm:omega-inv}, we obtain
\begin{align*}
 J^g = \lim_{n \to \infty} J_\textnormal{n}^g ~\preceq~ \eet{\stmt_{\textnormal{in}}}{g} ~\preceq~ \lim_{n \to \infty} J_\textnormal{n}^g = J^g.
\end{align*}

\subsection{Invariant Verification for the Outer Loop of the Coupon Collector Algorithm}
\label{sec:app-coupon-collector-outer}

We start by computing the characteristic functional $H$ of loop $C_\mathit{out}$ with
respect to run--time $\ctert{0}$:
\begin{align*}
 & H(Y) ~=~ \ctert{1} + \cond{x \leq 0} \cdot \zero + \cond{x > 0} \cdot \eet{\stmt_{\textnormal{in}};\ASSIGN{\cp{i}}{1};\ASSIGN{x}{x-1}}{Y} \\
 & ~=~ \ctert{1}  + \cond{x > 0} \cdot J^Y[x/x-1,\, \cp{i}/1] \tag{replace $g$ by $Y$ in ($\dagger$)}\\
 & ~=~ \ctert{1}  + \cond{x > 0} \cdot (3  + \cond{\cp{i} = 0} \cdot Y[x/x-1, \cp{i}/1] \tag{Def. $J^Y$} \\
 & ~+~  \cond{\cp{i} \neq 0} 
          \sum_{k=0}^{\infty} \left(\frac{\nocol}{N}\right)^{k} \left( 2 + \frac{1}{N} \sum_{j=1}^{N} \eval{\cp{j}=0} \cdot Y[x/x-1,\, \cp{i}/1,\, i/j] \right)
\end{align*}

Next we note a useful relationship between the number of collected and the number of missing coupons:
\begin{align*}
 N - \nocol ~=~ \sum_{i=1}^{N} (1 - \cond{\cp{i} \neq 0}) ~=~ \sum_{i=1}^{N} \cond{\cp{i} = 0}~. \tag{$\spadesuit$}
\end{align*}

Recall the $\omega$--invariant proposed for the outer loop in \autoref{subsec:coupon}:
\begin{align*}
 I_{n} 
 & ~=~ 1 + \sum_{\ell = 0}^{n} \cond{x > \ell} \cdot \left( 3 + \cond{n \neq 0} + 2 \cdot \sum_{k=0}^{\infty} \left( \frac{\nocol+\ell}{N} \right)^{k} \right) \\
 & \qquad - 2 \cdot \cond{\cp{i}=0} \cdot \cond{x > 0} \cdot \sum_{k=0}^{\infty} \left( \frac{\nocol}{N} \right)^{k}                      
\end{align*}
In order to keep calculations readable, let 
\begin{align*}
 {^{n} K_{{i}}^{{j}}} := \sum_{\ell = {i}}^{n} \cond{x > \ell {+j}}
              \cdot \left( 3 + \cond{n \neq 0} + 2 \cdot \sum_{k=0}^{\infty} \left( \frac{\nocol{+j}+\ell}{N} \right)^{k} \right).                
\end{align*}
Then, our proposed invariant can be written as
\[ I_{n} = 1 + {^{n}K_{0}^{0}} - 2 \cdot \cond{\cp{i}=0} \cdot \cond{x > 0} \cdot \sum_{k=0}^{\infty} \left( \frac{\nocol}{N} \right)^{k}. \]

Our goal is to to apply \autoref{thm:omega-inv} to show that $I_{n}$ is a lower as well as an upper $\omega$--invariant.
We first show $H(\zero) = I_\textnormal{0}$:
 
\begin{align*}
  H(\zero) & ~=~ 1 + \cond{x > 0} \cdot (3 + \cond{\cp{i} = 0} \cdot \zero[x/x-1, \cp{i}/1] \\
 &  \qquad    + 
          \cond{\cp{i} \neq 0}
          \cdot 
          \sum_{k=0}^{\infty} \left(\frac{\nocol}{N}\right)^{k} \\
 &        \cdot 
          \left(
              2 + \sum_{j=1}^{N} \frac{\cond{\cp{j} = 0}}{N} \cdot \zero[x/x-1, \cp{i}/1, i/j]
          \right) \\
 &        ) \\
 & ~=~ 1 + \cond{x > 0} \cdot \left(3 + 0 + \cond{\cp{i} \neq 0} \cdot \sum_{k=0}^{\infty} \left(\frac{\nocol}{N}\right)^{k} \cdot (2 + 0)\right) \\
 & ~=~ 1 + \cond{x > 0} \cdot \left(3 + 2 \cdot \cond{\cp{i} \neq 0} \cdot \sum_{k=0}^{\infty} \left(\frac{\nocol}{N}\right)^{k} \right) \\
 & ~=~ 1 + \cond{x > 0} \cdot \left(3 + 2 \cdot \sum_{k=0}^{\infty} \left(\frac{\nocol}{N}\right)^{k} \right) \tag{by $\spadesuit$} \\
 & \qquad - \cond{x > 0} \cdot \cond{\cp{i} = 0} \cdot \left( 2 \cdot \sum_{k=0}^{\infty} \left(\frac{\nocol}{N}\right)^{k}\right) \\
 & ~=~ 1 + {^{0}K_{0}^{0}} - 2 \cdot \cond{\cp{i}=0} \cdot \cond{x > 0} \cdot \sum_{k=0}^{\infty} \left( \frac{\nocol}{N} \right)^{k} \tag{Def. $^{0}K_{0}^{0}$} \\
 & ~=~ I_{0} \tag{Def. $I_{0}$}
\end{align*}

Before proving the reamining proof obligation, we note that
\begin{align*}
 & \cond{\cp{i} = 0} \cdot \nocol[\cp{i} / 1] \\
 & ~=~ \cond{\cp{i} = 0} \cdot \left( \sum_{j=1}^{N} \cond{\cp{j} \neq 0}[\cp{i}/1] \right) \tag{Def. $\nocol$} \\
 & ~=~ \cond{\cp{i} = 0} \cdot \left( \sum_{j=1}^{i-1} \cond{\cp{j} \neq 0} + \cond{\cp{i} \neq 0}[\cp{i}/1] + \sum_{j=i+1}^{N} \cond{\cp{j} \neq 0} \right) \\
 & ~=~ \cond{\cp{i} = 0} \cdot \left( 1 + \sum_{j=1}^{N} \cond{\cp{j} \neq 0} \right) \tag{ $\cond{\cp{i} = 0} \cdot \cond{\cp{i} \neq 0} = 0$ }\\
 & ~=~ \cond{\cp{i} = 0} \cdot \left( 1 + \nocol \right). \tag{Def. $\nocol$}
\end{align*}

As a result of this, we obtain
\begin{align*}
  & \cond{\cp{i} = 0} \cdot I_{n}[x/x-1, \cp{i}/1] \tag{$\clubsuit$} \\
  & ~=~ \cond{\cp{i} = 0} \cdot I_{n}[x/x-1, \cp{i}/1, i/j] \\
  & ~=~ \cond{\cp{i} = 0} \cdot (1+{^n K_{0}^{1}}).
\end{align*}

We are now in a position to show $H(I_{n}) = I_{n+1}$. 
\begin{align*}
 & H(I_{n}) \\
 & ~=~ 1 + \cond{x > 0} \cdot (3 + \cond{\cp{i} = 0} \cdot I_{n}[x/x-1, \cp{i}/1] \tag{Def. $H(Y)$} \\
 & \qquad + 
          \cond{\cp{i} \neq 0}
          \cdot 
          \sum_{k=0}^{\infty} \left(\frac{\nocol}{N}\right)^{k} \\
 & \qquad       \cdot 
          \left(
              2 + \sum_{j=1}^{N} \frac{\cond{\cp{j} = 0}}{N} \cdot I_{n}[x/x-1, \cp{i}/1, i/j]
          \right) \\
 & \qquad ) \\
 & \\
 & ~=~ 1 + \cond{x > 0} \cdot (3 + \cond{\cp{i} = 0} \cdot {(1+{^n K_{0}^{1}})} \tag{by $\clubsuit$}  \\
 & \qquad       + \cond{\cp{i} \neq 0} \cdot \sum_{k=0}^{\infty} \left(\frac{\nocol}{N}\right)^{k}
          \cdot \left(
              2 + \sum_{j=1}^{N} \frac{\cond{\cp{j} = 0}}{N} \cdot {(1+{^n K_{0}^{1}})}
          \right) \\
 & \qquad       ) \\
 & ~=~ 1 + \cond{x > 0} \cdot (3 + \cond{\cp{i} = 0} \cdot {(1+{^{n+1} K_{1}^{0}})} \tag{$^n K_{0}^{1} = ^{n+1} K_{1}^{0}$}  \\
 & \qquad       + \cond{\cp{i} \neq 0} \cdot \sum_{k=0}^{\infty} \left(\frac{\nocol}{N}\right)^{k}
          \cdot \left(
              2 + \sum_{j=1}^{N} \frac{\cond{\cp{j} = 0}}{N} \cdot {(1+{^{n+1} K_{1}^{0}})}
          \right) \\
 & \qquad       ) \\
 & \\
 & ~=~ 1 + \cond{x > 0} \cdot (3 + \cond{\cp{i} = 0} \cdot {(1+{^{n+1} K_{1}^{0}})}  \tag{by $\spadesuit$} \\
 & \qquad     + \cond{\cp{i} \neq 0} \cdot \sum_{k={0}}^{\infty} \left(\frac{\nocol}{N}\right)^{k} \cdot \left(
            2 + \left(1 - \frac{\nocol}{N}\right) \cdot \left(1+{^{n+1} K_{1}^{0}} \right)
          \right) \\
 & \qquad ) \\
 & \\
 & ~=~ 1 + \cond{x > 0} \cdot (3 + \cond{\cp{i} = 0} \cdot {(1+{^{n+1} K_{1}^{0}})}  \\
 & \qquad     + \cond{\cp{i} \neq 0} \cdot \sum_{k={0}}^{\infty} \left(\frac{\nocol}{N}\right)^{k} \cdot \left(
            3 + {^{n+1} K_{1}^{0}} - \frac{\nocol}{N} \left(1 + {^{n+1} K_{1}^{0}}\right) \right) \\
 & \qquad ) \\
 & \\
 & ~=~ 1 + \cond{x > 0} \cdot (4 + {^{n+1} K_{1}^{0}}  \\
 & \qquad     + \cond{\cp{i} \neq 0} \cdot \sum_{k={1}}^{\infty} \left(\frac{\nocol}{N}\right)^{k} \cdot \left(
            3 + {^{n+1} K_{1}^{0}} - \frac{\nocol}{N} \left(1 + {^{n+1} K_{1}^{0}}\right) \right) \\
 & \qquad     + \cond{\cp{i} \neq 0} \cdot \left( 2 - \frac{\nocol}{N} \left(1 + {^{n+1} K_{1}^{0}}\right) \right) ) \\
 & \\
 & ~=~ 1 + \cond{x > 0} \cdot (4 + {^{n+1} K_{1}^{0}}  \\
 & \qquad     + \cond{\cp{i} \neq 0} \cdot \sum_{k={1}}^{\infty} \left(\frac{\nocol}{N}\right)^{k} \cdot \left(
            (1 - \frac{\nocol}{N}) \cdot \left(1 + {^{n+1} K_{1}^{0}}\right) \right) \\
 & \qquad     + \cond{\cp{i} \neq 0} \cdot \left( 2 - \frac{\nocol}{N} \left(1 + {^{n+1} K_{1}^{0}}\right) \right) + 2 \cdot \cond{\cp{i} \neq 0} \cdot \sum_{k={1}}^{\infty} \left(\frac{\nocol}{N}\right)^{k} ) \\
 & \\
 & ~=~ 1 + \cond{x > 0} \cdot (4 + {^{n+1} K_{1}^{0}}  \\
 & \qquad     + \cond{\cp{i} \neq 0} \cdot \left(1 + {^{n+1} K_{1}^{0}}\right) \cdot \left(\sum_{k={1}}^{\infty} \left(\frac{\nocol}{N}\right)^{k} - \sum_{k={1}}^{\infty} \left(\frac{\nocol}{N}\right)^{k+1}\right) \\
 & \qquad     + \cond{\cp{i} \neq 0} \cdot \left( 2 - \frac{\nocol}{N} \left(1 + {^{n+1} K_{1}^{0}}\right) \right) + 2 \cdot \cond{\cp{i} \neq 0} \cdot \sum_{k={1}}^{\infty} \left(\frac{\nocol}{N}\right)^{k} ) \\
 & \\
 & ~=~ 1 + \cond{x > 0} \cdot (4 + {^{n+1} K_{1}^{0}} + \cond{\cp{i} \neq 0} \cdot \left(1 + {^{n+1} K_{1}^{0}}\right) \cdot \left(\frac{\nocol}{N}\right) \\
 & \qquad     + \cond{\cp{i} \neq 0} \cdot \left( 2 - \frac{\nocol}{N} \left(1 + {^{n+1} K_{1}^{0}}\right) \right)  + 2 \cdot \cond{\cp{i} \neq 0} \cdot \sum_{k={1}}^{\infty} \left(\frac{\nocol}{N}\right)^{k} ) \\
 & \\
 & ~=~ 1 + \cond{x > 0} \cdot \left(4 + {^{n+1} K_{1}^{0}} + 2 \cdot \cond{\cp{i} \neq 0} + 2 \cdot \cond{\cp{i} \neq 0} \cdot \sum_{k={1}}^{\infty} \left(\frac{\nocol}{N}\right)^{k} \right) \\
 & \\
 & ~=~ 1 + \cond{x > 0} \cdot \left(4 + {^{n+1} K_{1}^{0}} + 2 \cdot \cond{\cp{i} \neq 0} \cdot \sum_{k={0}}^{\infty} \left(\frac{\nocol}{N}\right)^{k} \right) \\
 & \\
 & ~=~ 1 + {^{n+1} K_{1}^{0}}
         + \cond{x > 0} \left(4 + 2 \cdot \cond{\cp{i} \neq 0} \cdot \sum_{k={0}}^{\infty} \left(\frac{\nocol}{N}\right)^{k} \right) \\
 & \\
 & ~=~ 1 + {^{n+1} K_{1}^{0}} + \cond{x > 0} \left(4 + 2 \cdot \cond{\cp{i} \neq 0} \cdot \sum_{k={0}}^{\infty} \left(\frac{\nocol}{N}\right)^{k} \right) \\
 & \qquad + (2-2) \cdot \cond{x > 0} \cdot \cond{\cp{i} = 0} \cdot \sum_{k={0}}^{\infty} \left(\frac{\nocol}{N}\right)^{k} \\
 & \\
 & ~=~ 1 + {^{n+1} K_{1}^{0}}
         + \cond{x > 0} \left( 4 + 2 \cdot \sum_{k={0}}^{\infty} \left(\frac{\nocol}{N}\right)^{k} \right) \\
 & \qquad - 2 \cdot \cond{x > 0} \cdot \cond{\cp{i} = 0} \cdot \sum_{k={0}}^{\infty} \left(\frac{\nocol}{N}\right)^{k} \\
 & \\
 & ~=~ 1 + {^{n+1} K_{1}^{0}}
         + \cond{x > 0} \left( 3 + \cond{n+1 \neq 0} + 2 \cdot \sum_{k={0}}^{\infty} \left(\frac{\nocol}{N}\right)^{k} \right) \\
 & \qquad - 2 \cdot \cond{x > 0} \cdot \cond{\cp{i} = 0} \cdot \sum_{k={0}}^{\infty} \left(\frac{\nocol}{N}\right)^{k} \\
 & \\
 & ~=~ 1 + {^{n+1} K_{0}^{0}} - 2 \cdot \cond{x > 0} \cdot \cond{\cp{i} = 0} \cdot \sum_{k={0}}^{\infty} \left(\frac{\nocol}{N}\right)^{k} \tag{Def. ${^{n+1} K_{0}^{0}}$}\\
 & \\
 & ~=~ I_{n+1} \tag{Def. of $I_{n}$}
\end{align*}

Now, by \autoref{thm:omega-inv}, we obtain
\begin{align*}
 I = \lim_{n \to \infty} I_\textnormal{n} ~\preceq~ \eet{\stmt_{\textnormal{out}}}{\zero} ~\preceq~ \lim_{n \to \infty} I_\textnormal{n} = I.
\end{align*}

\end{document}